\documentclass[aps,pre,twocolumn,groupedaddress]{revtex4-2}
\usepackage{changepage}   % for the adjustwidth environment
\usepackage{frcursive}
\usepackage{setspace}
\usepackage{sidecap}
\usepackage{multirow}

\usepackage{wrapfig}
\usepackage{colortbl}
\usepackage{etex}
\usepackage{bbding}
\usepackage{dingbat}
\usepackage[ruled]{algorithm2e}
\usepackage{amsmath}
\usepackage{amsthm}
\usepackage[newmultline]{empheq}
\usepackage{algorithmic}
\usepackage{amsfonts}
\usepackage{amssymb}
\usepackage{latexsym}
\usepackage{graphicx}
\usepackage{color}
\usepackage{balance}
\usepackage{textcomp}
\usepackage{nicefrac}
\usepackage{wasysym}
\usepackage[inline]{enumitem}
\usepackage{gensymb}
\usepackage{url}
\usepackage{fancybox}
\usepackage{longtable}
\usepackage[scr=boondox,  % heavily sloped
            cal=esstix]   % slightly sloped
           {mathalpha}
\usepackage{frcursive}

\usepackage{pifont,epsfig,rotating}

\usepackage{etoolbox}

\DeclareMathAlphabet{\mathpzc}{OT1}{pzc}{m}{it}

\usepackage{etaremune}

\usepackage{mathtools}

\usepackage[final,authormarkup=none,deletedmarkup=sout,addedmarkup=uwave,todonotes={textsize=tiny}]{changes}
\setdeletedmarkup{{\color{red}\sout{#1}}}
\usepackage{ctable}

\definecolor{dgreyblue}{rgb}{0.26,0.3,0.46}             %89.25 102 158.10

\newcommand{\bee}{{\mathbf{e}}}
\newcommand{\cA}{\mathcal{A}}
\newcommand{\cB}{\mathcal{B}}
\newcommand{\cC}{\mathcal{C}}

\newcommand{\cE}{\mathcal{E}}
\newcommand{\cH}{\mathcal{H}}

\newcommand{\cT}{\mathcal{T}}
\newcommand{\cTa}{\cT\!\mathrm{ail}}
\newcommand{\cHe}{\cH\mathrm{ead}}

\newcommand{\R}{{\mathbb R}}  %ams bold
\newcommand{\PP}{{\mathbb P}}  %ams bold
\newcommand{\PPL}{{\mathbb P}_{\mathrm{left}}}  %ams bold
\newcommand{\PPR}{{\mathbb P}_{\mathrm{right}}}  %ams bold

\renewcommand{\text}[1]{\hbox{\rm \ #1\ \/}}

\newcommand{\be}[1]{\begin{equation}\label{#1}}
\newcommand{\ee}{\end{equation}}
\newcommand{\beqn}{\begin{eqnarray*}}
\newcommand{\eeqn}{\end{eqnarray*}}
\newcommand{\beq}{\begin{eqnarray}}
\newcommand{\eeq}{\end{eqnarray}}
\newcommand{\ben}{\begin{enumerate}}
\newcommand{\een}{\end{enumerate}}
\newcommand{\bi}{\begin{itemize}}
\newcommand{\ei}{\end{itemize}}
\newcommand{\eps}{\varepsilon}

\newcommand{\IE}{{\em i.e.}\xspace}
\newcommand{\tx}{^{\rm th}}

\newtheorem{theorem}{Theorem}

\newenvironment{proof-sketch}{{\noindent\bf Sketch of Proof.\ }}{\hfill{\Pisymbol{pzd}{113}}\vspace{0.1in}}

\newcommand{\NP}{\mathsf{NP}}
\newcommand{\LP}{\mathsf{LP}}

\renewcommand{\deg}{\mathsf{deg}}

\newcommand{\cS}{\mathcal{S}}

\newcommand{\nbr}{\mathsf{Nbr}}

\newcommand{\cP}{\mathcal{P}}
\newcommand{\EA}{{\em et al.}\xspace}
\newcommand{\TB}{\vspace{-0.1ex}}\newcommand{\TiE}{\setlength{\itemsep}{-1ex}}

\newcommand{\junk}[1]{}

\newcommand{\EG}{{\it e.g.}\xspace}
\newcommand{\FI}[1]{Fig.~\ref{#1}\xspace}

\newcommand{\dist}{\mathrm{dist}}

\newcommand{\eqdef}{\stackrel{\mathrm{def}}{=}}
\definecolor{columbiablue}{rgb}{0.61, 0.87, 1.0}

\allowdisplaybreaks

\usepackage{soul}
\begin{document}

%\preprint{}

%%\title{Finding Influential Cores via Normalized Ricci Flows in Directed and Undirected Hypergraphs Encoding Biological and Social Interactions}
\title{Finding Influential Cores via Normalized Ricci Flows in Directed and Undirected Hypergraphs with Applications}



\author{Prithviraj Sengupta}
%%\email[]{psengu4@uic.edu}
\email[]{prithvi1096@gmail.com}
\homepage{www.linkedin.com/in/prithviraj-sengupta/}
\author{Nazanin Azarhooshang}
%%\email[]{nazarh2@uic.edu}
\email[]{nazanin.azarhoushang@gmail.com}
\homepage{www.linkedin.com/in/nazaninazarhooshang/}
\affiliation{
Department of Computer Science, 
%%\\
University of Illinois Chicago
\\
Chicago, IL 60607, USA
}

\author{R{\'{e}}ka Albert}
\email{rza1@psu.edu}
\homepage{www.ralbert.me}
\affiliation{
Department of Physics
\\
Pennsylvania State University 
\\
University Park, PA 16802, USA  
}

\author{Bhaskar DasGupta}
\email{bdasgup@uic.edu}
%%\homepage{www.cs.uic.edu/\textasciitilde dasgupta}
\homepage{http://bdasgup.github.io/}
\affiliation{
Department of Computer Science 
\\
University of Illinois Chicago
\\
Chicago, IL 60607, USA
}

%%\author{}
%\thanks{}
%%\affiliation{}


\date{\today}

%%%%%%%%%%%%%%%%%%%%%%%%%%%%%%%%%%%%%%%%%%%%%%%%%%%%%%%%%%%%%%%%%%%%%%%%%%%%%%%%%%%%%%
\begin{abstract}
Many biological and social systems are 
naturally represented as edge-weighted directed or undirected hypergraphs
since they exhibit \emph{group} interactions involving \emph{three or more}
system units 
as opposed to pairwise interactions that can be incorporated in graph-theoretic representations.
However, finding 
{influential cores}
in hypergraphs 
is still {not} as extensively studied as their graph-theoretic counter-parts.
To this end, we develop and implement a hypergraph-curvature guided discrete time diffusion process 
with suitable topological surgeries and edge-weight re-normalization procedures
for {both} 
undirected and directed 
weighted 
hypergraphs
to find influential cores.
We successfully apply 
our framework 
for directed hypergraphs
to 
seven metabolic hypergraphs 
and 
our framework 
for undirected hypergraphs
to 
two social (co-authorship) hypergraphs
to find influential cores, thereby demonstrating the practical feasibility 
of our approach. 
In addition, we prove a theorem showing that a certain edge weight re-normalization procedure in
a prior research work for Ricci flows for 
edge-weighted graphs has the 
{undesirable} outcome of modifying the edge-weights to negative numbers, 
thereby rendering the procedure {impossible} to use.
To the best of our knowledge, this seems to be 
{one of}
the first articles that formulates algorithmic approaches for finding 
core(s) of 
{(weighted or unweighted)}
directed 
hypergraphs. 
\end{abstract}
%%%%%%%%%%%%%%%%%%%%%%%%%%%%%%%%%%%%%%%%%%%%%%%%%%%%%%%%%%%%%%%%%%%%%%%%%%%%%%%%%%%%%%

%\keywords{}

\maketitle


%%%%%%%%%%%%%%%%%%%%%%%%%%%%%%%%%%%%%%%%%%%%%%%%%%%%%%%%%%%%%%%%%%%%%%%%%%%%%%%%%%%%%%
\section{\label{sec-intro}Introduction}
%%%%%%%%%%%%%%%%%%%%%%%%%%%%%%%%%%%%%%%%%%%%%%%%%%%%%%%%%%%%%%%%%%%%%%%%%%%%%%%%%%%%%%

Useful insights for many complex systems 
are often obtained by representing them as graphs and 
analyzing them using 
{graph-theoretic} 
and 
{combinatorial} 
tools~\cite{DL16,Newman-book,Albert-Barabasi-2002}.
Such graphs may vary in diversity from 
simple undirected graphs 
to edge-labeled directed graphs.
In such graphs, nodes represent the 
%\emph{basic} 
basic
units of the system (\EG, system variables) and edges represent 
relationships (\EG, correlations) between 
%\emph{pairs}
pairs of such basic units.
Once such graphs are constructed, they can be analyzed using graph-theoretic
measures such as 
{degree-based measures} 
(\EG, degree distributions), 
{connectivity-based measures} 
(\EG, clustering coefficients),
{geodesic-based measures} 
(\EG, betweenness centralities)
and other more novel network measures such as in~\cite{rich-club,LM07,ADGGHPSS11,bassett-et-al-2011}
to give meaningful insights into the properties and the dynamics of the system.
However, many real-world systems exhibit \emph{group} 
(\IE, 
{higher order}) 
interactions involving 
{three or more}
system units~\cite{BATTISTON20201,Battetal21,doi10113720M1355896}.
One way to handle these
higher order interactions is to 
%\emph{encode} 
encode them by a suitable combination of pairwise interactions
and then simply use the existing graph-theoretic tools (\EG, see~\cite{net10-1,DESZ07}).
While such approaches have been successful in the context of many real-world networks, 
they obviously do not encode the higher-order interactions in their full generalities.
A more direct approach would be to use \emph{directed} or \emph{undirected} hypergraphs to encode 
these interactions, and this is the approach we follow in this article.
Although the 
{theory} 
of hypergraphs has been considerably developed during the last few decades 
(\EG, see~\cite{Berg89}), applications of hypergraphs to real-world networks face their own challenges.
Sometimes it is 
%{not} 
not
clear how to generalize a concept from graphs to hypergraphs such that it 
best serves its purpose in the corresponding application,
and some computationally tractable 
graph-theoretic algorithms 
%{may} 
may
become intractable when generalized to hypergraphs (\EG, the maximum 
matching problem for graphs is polynomial-time solvable whereas the $3$-dimensional matching problem
for hypergraphs is $\NP$-complete~\cite{GJ79}).

Suitable notions of curvatures are natural measures of shapes of higher dimensional objects in 
mainstream physics and mathematics~\cite{book99,Berger12}.
There have been several attempts to extend these curvature measures to 
graphs and hypergraphs.
Two major notions of curvatures of graphs can be obtained via extending Forman's discretization~\cite{F03}
of Ricci curvature for (polyhedral or CW) complexes 
to undirected graphs (the ``\emph{Forman-Ricci curvature}'')~\cite{Sree1,Sree2,Weber17,DJY20,Samal18,CATAD21}
and via Ollivier's discretization of manifold Ricci curvature to 
undirected edge-weighted graphs (the ``\emph{Ollivier-Ricci curvature}'')~\cite{Oll11,Oll09,Oll10,Oll07}.
Both Ollivier-Ricci curvature and Forman-Ricci curvature assign a number to each edge of
the given graph, but the numbers are calculated in very \emph{different} ways since 
they capture \emph{different} metric properties of a Riemannian manifold; 
some comparative analysis of these two measures
can be found in~\cite{CATAD21,Samal18}.
Recently, 
the Ollivier-Ricci curvature
and 
the Forman-Ricci curvature
measures have been generalized in a few ways to 
{unweighted} 
directed and undirected hypergraphs~\cite{AGE18,EJ20,CDR23,Ak22,LRSJ21}.

A curvature-guided diffusion process called 
Ricci flow, along with ``topological surgery'' procedures to avoid topological singularities,
was originally introduced in the context of a 
Riemannian manifold by Hamilton~\cite{Ha81} to provide a continuous change of the metric of the manifold
to provide a continuous transformation (homeomorphism) of one manifold to another manifold.
One of the most ground-breaking application of this technique was done by Perelman~\cite{Pe02} 
to solve the 
{Poincar\'{e} conjecture},
which asserts that 
any three-dimensional manifold that is closed, connected, and has a trivial fundamental group is homeomorphic to the 
three-dimensional sphere.
In the context of a weighted undirected graphs, these techniques were extended in papers  
such as~\cite{NLLG19,SJB19,LBL22,WJS16,Weber17,axioms11090486,NLGG18}
to 
{iteratively} 
and 
{synchronously} 
change the weights of the edges of the graph.
Motivated by the fact that 
connected sum decomposition can be detected by the geometric Ricci flows in manifolds,
these techniques were then used to find communities or modules mostly in the context of \emph{undirected social graphs}. 
To prevent lack of convergence of graph Ricci flows within reasonable time, edge-weight 
{re-normalization} 
methods after every iteration were suggested and investigated in~\cite{LBL22}.

In this article we 
devise a computational framework to 
detect 
\emph{influential cores}
in \emph{both} 
undirected and directed 
weighted 
hypergraphs
by 
formulating and using 
a hypergraph-curvature guided discrete time diffusion process 
with suitable topological surgeries and edge-weight re-normalization procedures.
We
demonstrate the practical feasibility 
of our approach
by 
{successfully} 
applying 
our computational framework 
for directed hypergraphs
to 
seven metabolic hypergraphs 
and 
our computational framework 
for undirected hypergraphs
to 
two social (co-authorship) hypergraphs
to find influential cores.
In addition, we prove a theorem showing that a certain edge weight re-normalization procedure in~\cite{LBL22} for Ricci flows for 
edge-weighted graphs has an 
{undesirable} 
outcome of modifying the edge-weights to negative numbers, 
thereby rendering the procedure \emph{impossible} to use.

%%%%%%%%%%%%%%%%%%%%%%%%%%%%%%%%%%%%%%%%%%%%%%%%%%%%%%%%%%%%%%%%%%%%%%%%%%%%%%%%%%%%%%
\smallskip
\paragraph{Motivation for finding core(s) of a complex system}
%%%%%%%%%%%%%%%%%%%%%%%%%%%%%%%%%%%%%%%%%%%%%%%%%%%%%%%%%%%%%%%%%%%%%%%%%%%%%%%%%%%%%%
%%
Broadly speaking, the core of a complex system (also studied under the name ``cohesive subgraph''
in the network science literature~\cite{KIM2023290}) is a smaller sub-system 
that contributes \emph{significantly} to the functioning of the 
{overall} 
system, and 
thus focussing the analysis of the smaller sub-system, which could be \emph{easier} 
than analyzing the system as a whole, may reveal important characteristics of the 
overall system.
For example, in the context of brain graphs, 
identifications of  core(s) of the graph where neurons strongly interact with each other
provide 
effective characterizations of these graph topologies;
such cores 
{are} 
very important for various brain functions and cognition~\cite{fornito2016fundamentals}.
As another example, cores in attributed social networks
can be directly utilised for a 
{recommendation system}~\cite{KIM2023290}.

For systems represented by \emph{undirected graphs}, prior research works have used 
{various} 
definitions of 
what actually constitutes a core, such as 
via 
{modularity}~\cite{annurev-content-journals-annurev-psych}, 
via 
{rich clubs}~\cite{harrispo16},
or using 
{information-theoretic} 
ideas~\cite{KITAZONO2020232}.

Our definition for cores of hypergraphs 
requires sufficient connectivity and cohesiveness, nontrivial size, and a large centrality, 
evidenced by a large loss of short paths in the network when the core is removed. 
The specific definition and constraints are given in Section~\ref{sec-quality}.
%%%
The specific usefulness of finding cores 
for metabolic (directed) hypergraphs 
and 
for co-authorship (undirected) hypergraphs 
are discussed in Section~\ref{sec-finmod-dir-int}
and in Section~\ref{sec-finmod-undir-int}, respectively.
{\bf To the best of our knowledge, this seems to be 
one of
the first articles that formulates algorithmic approaches for finding 
core(s) of 
(weighted or unweighted)
directed 
hypergraphs}. 
%%%
Note that since graphs are special cases of hypergraphs, our methodologies are also applicable to graphs; 
however, in this article our focus is on hypergraphs that are \emph{not} graphs.

%%%%%%%%%%%%%%%%%%%%%%%%%%%%%%%%%%%%%%%%%%%%%%%%%%%%%%%%%%%%%%%%%%%%%%%%%%%%%%%%%%%%%%
\smallskip
\paragraph{Finding core$($s$)$ vs.\ modular decomposition}
%%%%%%%%%%%%%%%%%%%%%%%%%%%%%%%%%%%%%%%%%%%%%%%%%%%%%%%%%%%%%%%%%%%%%%%%%%%%%%%%%%%%%%
%%
Note that finding a core of a system is \emph{different} than the modular 
decomposition of graphs that has been extensively studied in the network
science literature~\cite{NG04,LN08,N06,DD13,D14}: the overall goal 
of graph decomposition (partitioning) into modules (also called communities 
or clusters) is to partition the \emph{entire} node set into modules and requires the 
optimization of a \emph{joint} fitness function of these modules to evaluate the quality of the decomposition.
In particular, the 
centrality parameters (quantifying loss of short paths when removing the core, 
see Section~\ref{sec-quality})
are \emph{not} relevant for typical modular decomposition applications 
and the 
size constraints
are unimportant if the joint fitness function is satisfactory 
(\EG, papers such as~\cite{DD13} 
show that a good approximation to Newman's modularity value 
may be obtained even though the corresponding modules will \emph{not} satisfy our
size constraints).

%%%%%%%%%%%%%%%%%%%%%%%%%%%%%%%%%%%%%%%%%%%%%%%%%%%%%%%%%%%%%%%%%%%%%%%%%%%%%%%%%%%%%%
\subsection{\label{sec-basic-defn}Basic Definitions and Notations}
%%%%%%%%%%%%%%%%%%%%%%%%%%%%%%%%%%%%%%%%%%%%%%%%%%%%%%%%%%%%%%%%%%%%%%%%%%%%%%%%%%%%%%

A weighted \emph{directed} hypergraph $H=(V,E,w)$ 
consists of a node set $V$, a set $E$ of (directed) hyperedges and a hyperedge-weight function 
$w:E\mapsto \R_{\ge 0}$.
A \emph{directed} hyperedge $e\in E$ is an ordered pair $(\cTa_e,\cHe_e)$ 
where 
$\emptyset \subset \cTa_e\subset V$
is the \emph{tail}, 
$\emptyset \subset \cHe_e\subset V$
is the \emph{head} and $\cTa_e \neq \cHe_e$; for convenience we will also denote the
hyperedge by $\cTa_e\to \cHe_e$.
For a node $x \in V$, the \emph{in-degree} $\deg_{x}^{in}$ is the number of 
``incoming'' hyperedges, 
\IE, the number of hyperedges $e'$ such that $x\in \cHe_{e'}$, and 
the \emph{out-degree} $\deg_{x}^{out}$ is the number of ``outgoing'' hyperedges, 
\IE, the number of hyperedges $e'$ such that $x\in \cTa_{e'}$.
A (directed) \emph{path} 
$\cP_{x,y}$ from node $x$ to node $y$ 
is an alternating sequence 
$(x=v_1,e_1,\dots,v_k,e_k,v_{k+1}=y)$
of \emph{distinct} nodes and directed hyperedges such that
$v_i\in\cTa_{e_i}$ and 
$v_{i+1}\in\cHe_{e_i}$
for each $i\in\{1,\dots,k\}$; 
the \emph{length} of the path is $\sum_{i=1}^k w(e_i)$.
We will denote by $\dist_H(u,v)$ 
the \emph{minimum} length of any path from $u$ to $v$; note that
$\dist_H(u,v)$ 
need \emph{not} be same as 
$\dist_H(v,u)$. 
A directed hypergraph is \emph{weakly connected} provided
for every pair of nodes $x$ and $y$ there is \emph{either} a path from $x$ to $y$ 
\emph{or} a path from $y$ to $x$.
\textbf{We assume from now onwards that our original (input) directed hypergraph is weakly connected}.

A weighted \emph{undirected} hypergraph $H=(V,E,w)$ 
consists of a node set $V$, a set $E$ of (undirected) hyperedges and a hyperedge-weight function 
$w:E\mapsto \R_{\ge 0}$.
A (undirected) hyperedge $e\in E$ is a subset of nodes $\cA_e$ 
where 
$\emptyset\subset\cA_e\subseteq V$.
For a node $x \in V$, the \emph{degree} $\deg_{x}$ is the number of 
the number of hyperedges $e'$ such that $x\in \cA_{e'}$.
An (undirected) \emph{path} 
$\cP_{x,y}$ between nodes $x$ and $y$ 
is an alternating sequence 
$(x=v_1,e_1,\dots,v_k,e_k,v_{k+1}=y)$
of \emph{distinct} nodes and undirected hyperedges such that
$v_i,v_{i+1}\in\cA_{e_i}$
for each $i\in\{1,\dots,k\}$; 
the \emph{length} of the path is $\sum_{i=1}^k w(e_i)$.
We will denote by $\dist_H(u,v)$ 
the \emph{distance} (\IE, \emph{minimum} length of any path) between $u$ and $v$.
A undirected hypergraph is \emph{connected} provided
for every pair of nodes there is a path between them.
\textbf{We assume from now onwards that our original (input) undirected hypergraph is connected}.

%%%%%%%%%%%%%%%%%%%%%%%%%%%%%%%%%%%%%%%%%%%%%%%%%%%%%%%%%%%%%%%%%%%%%%%%%%%%%%%%%%%%%%
\subsection{\label{sec-prior}Discussions on Prior Relevant Research}
%%%%%%%%%%%%%%%%%%%%%%%%%%%%%%%%%%%%%%%%%%%%%%%%%%%%%%%%%%%%%%%%%%%%%%%%%%%%%%%%%%%%%%

There are some prior published articles that deal with finding cores or core-like structures  
for \emph{undirected} hypergraphs, \emph{mostly} for unweighted 
undirected hypergraphs~\cite{doi:10.1137/22M1480926,pmlr-v84-chien18a,doi:10.1126-sciadv.adg9159,Carletti_2021,Erikss21,Krits24,doi:10.1080-01621459-2021-2002157,Eriksson2022,MIPB23,MIPB24}
but some also on weighted undirected hypergraphs~\cite{YiqunXu2023}.
However, there seem to be very few peer-reviewed prior articles dealing with finding 
cores for \emph{directed} (and more so for \emph{weighted} directed) hypergraphs where cores are defined in the same sense as
used in this article; the authors themselves were unable to get any relevant peer-reviewed prior works via google search.
For example, the articles by Pretolani~\cite{PRETOLANI2013226}
and by Volpentesta~\cite{VOLPENTESTA2008390}
investigate intriguing but \emph{different} 
concepts that 
refer to
sub-structures in unweighted directed hypergraphs based on hyperpaths.
%%%
Note that although for graphs replacing an undirected
edge by two directed edges make the core finding problem for undirected graphs solvable from 
the core finding problem for directed graphs, a similar trick cannot directly be used for 
hypergraphs since the head or tail may contain more than one node (\IE, 
a core finding algorithm for directed hypergraphs does not 
readily translate
to an algorithm for finding cores 
in undirected hypergraphs).
A strength of our framework is that
we have a single overall algorithmic method (albeit with different calculations of certain quantities) that can be adopted 
for \emph{both} directed and undirected hypergraphs. 

Most of the prior works on finding cores for undirected unweighted hypergraphs
involve iteratively identifying a set of nodes such that the degree of every node in the 
set is at least $k$ for a given $k$. 
We go deeper than these prior works by defining the
centrality quality parameters of the core 
as stated via Equations~\eqref{eq-undir-discon} and \eqref{eq-undir-dist}.
These centrality quality parameters,
which are combinatorial analogs of the information loss quantifications 
used in prior research works such as~\cite{KITAZONO2020232}, 
are significant in real-world applications to biological and social networks,
as mentioned in prior publications such as~\cite{ADM14} in the context of undirected unweighted graphs
(\EG, see~Figure~5 for biological networks and Figures 7--8 for social networks 
in~\cite{ADM14}
with associated texts),
and are closely related to the concept of structural holes~\cite{Burt92} in social networks.

%%%%%%%%%%%%%%%%%%%%%%%%%%%%%%%%%%%%%%%%%%%%%%%%%%%%%%%%%%%%%%%%%%%%%%%%%%%%%%%%%%%%%%
\section{\label{sec-methods-and-materials}Methods and Materials}
%%%%%%%%%%%%%%%%%%%%%%%%%%%%%%%%%%%%%%%%%%%%%%%%%%%%%%%%%%%%%%%%%%%%%%%%%%%%%%%%%%%%%%

%%%%%%%%%%%%%%%%%%%%%%%%%%%%%%%%%%%%%%%%%%%%%%%%%%%%%%%%%%%%%%%%%%%%%%%%%%%%%%%%%%%%%%
\subsection{\label{sec-def-curv}Definition of Curvatures of Weighted Hypergraphs}
%%%%%%%%%%%%%%%%%%%%%%%%%%%%%%%%%%%%%%%%%%%%%%%%%%%%%%%%%%%%%%%%%%%%%%%%%%%%%%%%%%%%%%

\newcommand{\emd}{{\sc Emd}}
\newcommand{\ric}{{\sc Ric}}

Curvatures of weighted hypergraphs are defined in somewhat different ways depending on whether
the hypergraph is directed or undirected. However, both definitions use a common paradigm of 
\emd\ 
(\emph{Earth Mover's Distance}, 
also known as the \emph{ $L_1$ transportation distance}, the \emph{ $L_1$ Wasserstein distance} or 
the \emph{Monge-Kantorovich-Rubinstein distance}~\cite{Ma72,Gui1,Gui2,Vil03})
defined on the hypergraph in the following manner based on the notations and terminologies in~\cite{ASD20}.
Let $H=(V,E,w)$ be a directed or undirected hypergraph.
Suppose that we have
two probability distributions 
$\PPL$ and $\PPR$ over the set of nodes $V$, \IE, two real numbers 
$0\leq \PPL(v),\PPR(v)\leq 1$ for every node $v\in V$ with 
$
\sum_{v\in V}\PPL(v)=
\sum_{v\in V}\PPR(v)=1
$.
We can think of 
$\PPL(v)$ as the total amount of ``earth'' (dirt) at node $v$ that need to be moved to other nodes, 
and $\PPR(v)$ as the \emph{maximum} total amount of earth node $v$ can store.
The cost of transporting \emph{one} unit of earth from node $u$ to node $v$ is 
$\dist_H(u,v)$, and 
the goal is to move \emph{all} units of earth (determined by $\PPL$) while simultaneously
satisfying \emph{all} storage requirements (dictated by $\PPR$) and \emph{minimizing} the total transportation cost.
Letting the real variable $z_{u,v}\in[0,1]$ denote the amount of shipment from node $u$ to node $v$ in an optimal solution, 
\emd\ for 
the two probability distributions $\PP_1$ and $\PP_2$ on
$V$ is the \emph{linear programming} ($\LP$) problem shown in \FI{f1}
which can be solved in polynomial time.
We will use the notation 
\emd$_H(\PPL,\PPR)$ 
to denote the value of the objective function in an optimal solution 
of the $\LP$ in \FI{f1}.

%%%%%%%%%%%%%%%%%%%%%%%%%%%%%%%%%%%%%%%%%%%%%%%%%%%%%%%%%%%%%%%%%%%%%%%%%%%%%%%%%%%%%%%
\begin{figure*}
%%%%%%%%%%%%%%%%%%%%%%%%%%%%%%%%%%%%%%%%%%%%%%%%%%%%%%%%%%%%%%%%%%%%%%%%%%%%%%%%%%%%%%%%
\begin{tabular}{ r l r}
\\
\toprule
\multicolumn{3}{c}{{\bf variables}: $z_{u,v}$ for every pair of nodes $u,v\in V$} 
\\
[5pt]
  \emph{minimize} &  
	      $\sum\limits_{u\in V}\sum\limits_{v\in V'} \dist_H(u,v)\, z_{u,v}$ 
	                &
        (* minimize total transportation cost *)
\\
[5pt]
  \emph{subject to} & & 
\\
[5pt]
	 & 
   $\sum\limits_{v\in V} z_{u,v} = \PPL(u),\,\,$ for each $u\in V$
	 & 
  (* ship from $u$ as much as it has *)
\\
[5pt]
	 & 
   $\sum\limits_{u\in V} z_{u,v} = \PPR(v),\,\,$ for each $v\in V$
	&
  (* ship to $v$ as much as it can store *)
\\
[5pt]
	& 
  $z_{u,v}\geq 0,\,\,$ for each $u,v\in V$
	& 
\\
\bottomrule
\end{tabular}
%%%%%%%%%%%%%%%%%%%%%%%%%%%%%%%%%%%%%%%%%%%%%%%%%%%%%%%%%%%%%%%%%%%%%%%%%%%%%%%%%%%%%%%%
\caption{\label{f1}$\LP$-formulation for \emd\ on hypergraph $H=(V,E,w)$ corresponding 
to distributions $\PPL$ and $\PPR$.
Comments are enclosed by (* and *).}
%%%%%%%%%%%%%%%%%%%%%%%%%%%%%%%%%%%%%%%%%%%%%%%%%%%%%%%%%%%%%%%%%%%%%%%%%%%%%%%%%%%%%%%%
\end{figure*}   
%%%%%%%%%%%%%%%%%%%%%%%%%%%%%%%%%%%%%%%%%%%%%%%%%%%%%%%%%%%%%%%%%%%%%%%%%%%%%%%%%%%%%%%%

Given a hypergraph $H$ and an edge $e$ of $H$, the curvature value of $e$ is then 
computed as follows:
%%%%%%%%%%%%%%%%%%%%%%%%%%%%%%%%%%%%%%%%%%%%%%%%%%%%%%%%%%%%%%%%%%%%%%%%%%%%%%%%%%%%%%%%
\begin{enumerate}[label=$\triangleright$]
\item
Fix appropriate distributions for $\PPL$ and $\PPR$.
%%%%%%%%%%%%%%%%%%%%%%%%%%%%%%%%%%%%%%%%%%%%%%%%%%%%%%%%%%%%%%%%%%%%%%%%%%%%%%%%%%%%%%%%
\item
Use a formula 
for Ricci curvature of the hyperedge $e$
The formula is \emph{different} depending on whether the hypergraph is directed or undirected,
and shown below:
%%%%%%%%%%%%%%%%%%%%%%%%%%%%%%%%%%%%%%%%%%%%%%%%%%%%%%%%%%%%%%%%%%%%%%%%%%%%%%%%%%%%%%%%
\begin{enumerate}[label=$\triangleright$]
\item
For directed hypergraphs,
the Ricci curvature of the hyperedge $e$ is calculated as:
%%%%%%%%%%%%%%%%%%%%%%%%%%%%%%%%%%%%%%%%%%%%%%%%%%%%%%%%%%%%%%%%%%%%%%%%%%%%%%%%%%%%%%%%%
\begin{empheq}[box=\Ovalbox]{gather}
\mbox{\ric}(e) = 1 - \frac { \mbox{\emd}_H(\PPL,\PPR) } { w(u,v) }
\label{eq1}
\end{empheq}
%%%%%%%%%%%%%%%%%%%%%%%%%%%%%%%%%%%%%%%%%%%%%%%%%%%%%%%%%%%%%%%%%%%%%%%%%%%%%%%%%%%%%%%%
An informal intuitive understanding of the connection of \emd\ to 
Ricci curvature in the above formula, as explained in prior research works such as~\cite{Oll11}
in the context of graphs, is as follows.
The Ricci curvature at a point $x$ in 
a smooth Riemannian manifold
can be thought of 
transporting a small ball centered at $x$ along that direction and measuring the ``distortion'' 
of that ball. 
In~\eqref{eq1}
the role of the direction is captured by the hyperedge $(u,v)$,
the roles of the balls at the two nodes are played by the distributions $\PPL$ and $\PPR$,
and the role of the distortion due to transportation is captured by the \emd\ measure. 
%%%%%%%%%%%%%%%%%%%%%%%%%%%%%%%%%%%%%%%%%%%%%%%%%%%%%%%%%%%%%%%%%%%%%%%%%%%%%%%%%%%%%%%%
\item
For undirected hypergraphs, 
the Ricci curvature of the hyperedge $e$ is calculated as:
%%%%%%%%%%%%%%%%%%%%%%%%%%%%%%%%%%%%%%%%%%%%%%%%%%%%%%%%%%%%%%%%%%%%%%%%%%%%%%%%%%%%%%%%%
\begin{empheq}[box=\Ovalbox]{gather}
\mbox{\ric}(e) = 1 - \frac{1}{\binom{|\cA_e|}{2} } \times 
\sum_{ \substack{p,q\in \cA_e \\ p\neq q  }} \mbox{\emd}_H(\PPL^p,\PPR^q)
\label{eq2}
\end{empheq}
%%%%%%%%%%%%%%%%%%%%%%%%%%%%%%%%%%%%%%%%%%%%%%%%%%%%%%%%%%%%%%%%%%%%%%%%%%%%%%%%%%%%%%%%%
The calculation for the curvature averages out weighted-lazy random walk probabilities 
over all pairs of distinct nodes in $e$.
There is one special case not covered by the above definition but may occur in our
undirected co-authorship hypergraphs: namely when $\cA_e=\{u\}$ for some node $u$ 
corresponding to a paper written by \emph{just one} author. For this case we treat the 
hyperedge as a self-loop from $u$ to $u$ giving an \emd\ value of zero.
\end{enumerate}
%%%%%%%%%%%%%%%%%%%%%%%%%%%%%%%%%%%%%%%%%%%%%%%%%%%%%%%%%%%%%%%%%%%%%%%%%%%%%%%%%%%%%%%%
\end{enumerate}
%%%%%%%%%%%%%%%%%%%%%%%%%%%%%%%%%%%%%%%%%%%%%%%%%%%%%%%%%%%%%%%%%%%%%%%%%%%%%%%%%%%%%%%%
The exact calculations of $\PPL$ and $\PPR$ are somewhat different depending on 
whether the hypergraph is directed or undirected, and this is described in next two sections.
Let $H=(V,E,w)$ be the weighted (directed or undirected) hypergraph, and $e\in E$ be the hyperedge considered.

%%%%%%%%%%%%%%%%%%%%%%%%%%%%%%%%%%%%%%%%%%%%%%%%%%%%%%%%%%%%%%%%%%%%%%%%%%%%%%%%%%%%%%
\subsubsection{\label{sec-ricci-dir}Calculations of $\PPL$ and $\PPR$ for a Weighted Directed Hypergraph}
%%%%%%%%%%%%%%%%%%%%%%%%%%%%%%%%%%%%%%%%%%%%%%%%%%%%%%%%%%%%%%%%%%%%%%%%%%%%%%%%%%%%%%

The distributions $\PPL$ and $\PPR$ are determined by the nodes in 
$\cTa_e$ and $\cHe_e$, respectively.
$\PPL$ is determined in the following manner
by adopting the calculations in~\cite{EJ20}: 
%%%%%%%%%%%%%%%%%%%%%%%%%%%%%%%%%%%%%%%%%%%%%%%%%%%%%%%%%%%%%%%%%%%%%%%%%%%%%%%%%%%%%%%%%%%%%%%%%%%%%%%%%%%%%%%%
\begin{enumerate}[label=$\triangleright$]
%%%%%%%%%%%%%%%%%%%%%%%%%%%%%%%%%%%%%%%%%%%%%%%%%%%%%%%%%%%%%%%%%%%%%%%%%%%%%%%%%%%%%%%%%%%%%%%%%%%%%%%%%%%%%%%%
\item
Initially, $\PPL(u)=0$ for all $u\in V$. In our subsequent steps, we will add to these values as appropriate.
%%%%%%%%%%%%%%%%%%%%%%%%%%%%%%%%%%%%%%%%%%%%%%%%%%%%%%%%%%%%%%%%%%%%%%%%%%%%%%%%%%%%%%%%%%%%%%%%%%%%%%%%%%%%%%%%
\item
We divide the total probability $1$ equally among the nodes in $\cTa_e$, thus ``allocating'' a value of 
$(| \cTa_e |)^{-1}$ 
to each node in question.
%%%%%%%%%%%%%%%%%%%%%%%%%%%%%%%%%%%%%%%%%%%%%%%%%%%%%%%%%%%%%%%%%%%%%%%%%%%%%%%%%%%%%%%%%%%%%%%%%%%%%%%%%%%%%%%%
\item
For every node $x\in\cTa_e$ with 
$\deg_x^{in}=0$, 
we add 
$(| \cTa_e |)^{-1}$ 
to 
$\PPL(x)$.
%%%%%%%%%%%%%%%%%%%%%%%%%%%%%%%%%%%%%%%%%%%%%%%%%%%%%%%%%%%%%%%%%%%%%%%%%%%%%%%%%%%%%%%%%%%%%%%%%%%%%%%%%%%%%%%%
\item
For every node $x\in\cTa_e$ with 
$\deg_x^{in}>0$, 
we perform the following:
%%%%%%%%%%%%%%%%%%%%%%%%%%%%%%%%%%%%%%%%%%%%%%%%%%%%%%%%%%%%%%%%%%%%%%%%%%%%%%%%%%%%%%%%%%%%%%%%%%%%%%%%%%%%%%%%
\begin{enumerate}[label=$\triangleright$]
%%%%%%%%%%%%%%%%%%%%%%%%%%%%%%%%%%%%%%%%%%%%%%%%%%%%%%%%%%%%%%%%%%%%%%%%%%%%%%%%%%%%%%%%%%%%%%%%%%%%%%%%%%%%%%%%
\item
We divide the probability 
$(| \cTa_e |)^{-1}$ 
equally among the hyperedges 
$e'$ such that $x\in\cHe_{e'}$, 
thus ``allocating'' a value of 
$(| \cTa_e |\times\deg_x^{in})^{-1}$ 
to each hyperedge in question.
%%%%%%%%%%%%%%%%%%%%%%%%%%%%%%%%%%%%%%%%%%%%%%%%%%%%%%%%%%%%%%%%%%%%%%%%%%%%%%%%%%%%%%%%%%%%%%%%%%%%%%%%%%%%%%%%
\item
For each such hyperedge 
$e'$ such that $x\in\cHe_{e'}$, 
we divide the allocated value equally among the nodes in 
$\cTa_{e'}$ and add these values to the probabilities of these nodes.
In other words, 
for every node $y\in\cTa_{e'}$
we add 
$(| \cTa_e |\times\deg_x^{in} \times |\cTa_{e'}|)^{-1}$ 
to 
$\PPL(y)$.
%%%%%%%%%%%%%%%%%%%%%%%%%%%%%%%%%%%%%%%%%%%%%%%%%%%%%%%%%%%%%%%%%%%%%%%%%%%%%%%%%%%%%%%%%%%%%%%%%%%%%%%%%%%%%%%%
\end{enumerate}
%%%%%%%%%%%%%%%%%%%%%%%%%%%%%%%%%%%%%%%%%%%%%%%%%%%%%%%%%%%%%%%%%%%%%%%%%%%%%%%%%%%%%%%%%%%%%%%%%%%%%%%%%%%%%%%%
\end{enumerate}
%%%%%%%%%%%%%%%%%%%%%%%%%%%%%%%%%%%%%%%%%%%%%%%%%%%%%%%%%%%%%%%%%%%%%%%%%%%%%%%%%%%%%%%%%%%%%%%%%%%%%%%%%%%%%%%%
Note that the final probability for each node is calculated by 
summing all the contributions from each bullet point.
In closed form, 
$\PPL^x(y)$ is given by:
%%%%%%%%%%%%%%%%%%%%%%%%%%%%%%%%%%%%%%%%%%%%%%%%%%%%%%%%%%%%%%%%%%%%%%%%%%%%%%%%%%%%%%%%%
\begin{multline*}
\PPL^x(y)= 
\frac{
\delta_{\deg_y^{in},0} \times \delta_{ |\cTa_e|-1, | \cTa_e \setminus \{y\} | }
}{| \cTa_e |} 
\\
+
\sum_{\begin{subarray}{l} 
          x\in\cTa_e
					\\
					x\in\cHe_{e'}
          \\
          y\in\cTa_{e'}
      \end{subarray}
         }
\frac{
1 - \delta_{\deg_x^{in},0}
}
{
| \cTa_e |\times\deg_x^{in} \times |\cTa_{e'}|
}
\end{multline*}
%%%%%%%%%%%%%%%%%%%%%%%%%%%%%%%%%%%%%%%%%%%%%%%%%%%%%%%%%%%%%%%%%%%%%%%%%%%%%%%%%%%%%%%%%
where $\delta(i,j)$ is the 
Kronecker delta function, \IE,
\[
\delta_{i,j}=
\left\{
\begin{array}{r l}
1, & \mbox{if $i=j$} 
\\
0, & \mbox{otherwise} 
\end{array}
\right.
\]
$\PPR$ is determined in a symmetric manner. 
The details are provided in the appendix for the sake of completeness.

%%%%%%%%%%%%%%%%%%%%%%%%%%%%%%%%%%%%%%%%%%%%%%%%%%%%%%%%%%%%%%%%%%%%%%%%%%%%%%%%%%%%%%
\subsubsection{\label{sec-ricci-undir}Calculations of $\PPL$ and $\PPR$ for an Undirected Hypergraph}
%%%%%%%%%%%%%%%%%%%%%%%%%%%%%%%%%%%%%%%%%%%%%%%%%%%%%%%%%%%%%%%%%%%%%%%%%%%%%%%%%%%%%%

%%%%%%%%%%%%%%%%%%%%%%%%%%%%%%%%%%%%%%%%%%%%%%%%%%%%%%%%%%%%%%%%%%%%%%%%%%%%%%%%%%%%%%
\begin{figure}
\caption{\label{fig-new-6}An illustration of the 
calculations of $\PPL$ and $\PPR$ for an undirected hypergraph,
as outlined in Section~\ref{sec-ricci-undir}, where the node set is 
$\{ s_1, s_2, s_3, s_4, s_5, s_6, s_7, s_8 \} $
and the four hyperedges are 
$\{s_1,s_2,s_3,s_4\}$, $\{s_1,s_5,s_7\}$, 
$\{s_5,s_6\}$ and $\{s_2,s_8\}$.}
\includegraphics[width=0.45\textwidth]{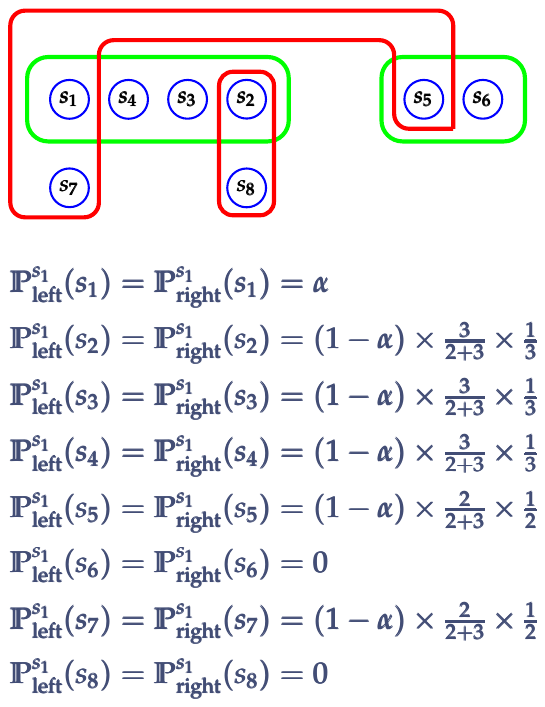}
\end{figure}  
%%%%%%%%%%%%%%%%%%%%%%%%%%%%%%%%%%%%%%%%%%%%%%%%%%%%%%%%%%%%%%%%%%%%%%%%%%%%%%%%%%%%%%

For this case, 
$\PPL^x(y) = \PPR^x(y)$ for all 
$x,y\in\cA_e$ since the hypergraph is undirected.  
Let $0<\alpha<1$ be a parameter that encodes the ``laziness'' of the random walk.
Then, $\PPL^x$ is determined in the following manner
by adopting the calculations in~\cite{CDR23}
(see \FI{fig-new-6} for an illustration):
%%%%%%%%%%%%%%%%%%%%%%%%%%%%%%%%%%%%%%%%%%%%%%%%%%%%%%%%%%%%%%%%%%%%%%%%%%%%%%%%%%%%%%%%%%%%%%%%%%%%%%%%%%%%%%%%
\begin{enumerate}[label=$\triangleright$]
%%%%%%%%%%%%%%%%%%%%%%%%%%%%%%%%%%%%%%%%%%%%%%%%%%%%%%%%%%%%%%%%%%%%%%%%%%%%%%%%%%%%%%%%%%%%%%%%%%%%%%%%%%%%%%%%
\item
Initially, 
$\PPL^x(x)=\alpha$ and 
$\PPL^x(y)=0$ 
for all $y\neq x$. In our subsequent steps, we will add to these values as appropriate.
%%%%%%%%%%%%%%%%%%%%%%%%%%%%%%%%%%%%%%%%%%%%%%%%%%%%%%%%%%%%%%%%%%%%%%%%%%%%%%%%%%%%%%%%%%%%%%%%%%%%%%%%%%%%%%%%
\item
We divide and allocate the remaining total probability $1-\alpha$ among all the hyperedges that contain $x$ proportionally to 
the cardinality of these hyperedges \emph{excluding} the node $x$, \IE, 
a hyperedge $e'$ such that $x\in\cA_e$ gets 
$
\delta_{e'}\eqdef\frac{|\cA_{e'}|-1}{\sum_{x\in \cA_{e''}} \left( |\cA_{e''}|-1 \right)}
$.
Then, 
we divide the allocated value $\delta_{e'}$ equally among the nodes in 
$\cA_{e'}$ \emph{excluding} $x$ and add these values to the probabilities of these nodes.
In other words, 
for every node $y\in\cA_{e'}$ such that $y\neq x$ 
we add 
$\frac{\delta_{e'}}{|\cA_{e'}|-1}$
to 
$\PPL^x(y)$.
%%%%%%%%%%%%%%%%%%%%%%%%%%%%%%%%%%%%%%%%%%%%%%%%%%%%%%%%%%%%%%%%%%%%%%%%%%%%%%%%%%%%%%%%%%%%%%%%%%%%%%%%%%%%%%%%
\end{enumerate}
%%%%%%%%%%%%%%%%%%%%%%%%%%%%%%%%%%%%%%%%%%%%%%%%%%%%%%%%%%%%%%%%%%%%%%%%%%%%%%%%%%%%%%%%%%%%%%%%%%%%%%%%%%%%%%%%
In closed form, 
$\PPL^x(y)$ is given by:
%%%%%%%%%%%%%%%%%%%%%%%%%%%%%%%%%%%%%%%%%%%%%%%%%%%%%%%%%%%%%%%%%%%%%%%%%%%%%%%%%%%%%%%%%
\begin{multline*}
\PPL^x(y) = 
%%%%%%%%%%%%%%%%%%%%%%%%%%%%%%%%%%%%%%%%%%%%%%%%%%%%%%%%%%%%%%%%%%%%%%%%%%%%%%%%%%%%%%%%%
\\
%%%%%%%%%%%%%%%%%%%%%%%%%%%%%%%%%%%%%%%%%%%%%%%%%%%%%%%%%%%%%%%%%%%%%%%%%%%%%%%%%%%%%%%%%
\left\{
\begin{array}{l l}
\alpha, & \mbox{if $x=y$}
\\
%%%%%%%%
\begin{array}{l}
(1-\alpha)
\times
\!\!\!
{\displaystyle \sum_{x,y\in e'} }
\frac{|\cA_{e'}|-1}{\sum_{x\in \cA_{e''}} (|\cA_{e''}|-1) }
\times
\frac{1}{|\cA_{e'}|-1}
\\
[4pt]
\displaystyle
\hspace*{0.3in}
=
(1-\alpha)
\times
\frac{
\big| \, \{e' \,|\, x,y\in\cA_{e'} \} \, \big|
}
{
\sum_{x\in \cA_{e''}} (|\cA_{e''}|-1)
}
\end{array}
%%%%%%%%
, 
& \mbox{otherwise}
\end{array}
\right.
\end{multline*}
%%%%%%%%%%%%%%%%%%%%%%%%%%%%%%%%%%%%%%%%%%%%%%%%%%%%%%%%%%%%%%%%%%%%%%%%%%%%%%%%%%%%%%%%%
The parameter $\alpha$ was used in the context of Ricci flows on 
undirected graphs by Ni \EA~\cite{NLLG19}, and by Lai \EA~\cite{LBL22}. 
These works suggested using a non-zero value for $\alpha$, but the exact
choice of $\alpha$ was left to the specific application in question. 
Thus, we decided to use a non-zero value of $\alpha$ and found that 
a value of $\alpha=0.1$ works best in our applications.

%%%%%%%%%%%%%%%%%%%%%%%%%%%%%%%%%%%%%%%%%%%%%%%%%%%%%%%%%%%%%%%%%%%%%%%%%%%%%%%%%%%%%%
\subsection{\label{sec-flow-surg}Ricci Flow, Weight-renormalization, Topological Surgery
and Flow Convergence for Weighted Hypergraphs}
%%%%%%%%%%%%%%%%%%%%%%%%%%%%%%%%%%%%%%%%%%%%%%%%%%%%%%%%%%%%%%%%%%%%%%%%%%%%%%%%%%%%%%

Before proceeding with the technical descriptions, we first provide a brief informal intuition 
behind the proposed approach.
The curvature values of the hyperedges provide a (positive or negative) value to each hyperedge of 
the hypergraph. 
The Ricci flow is an iterative process that produces a sequence of hypergraphs, where 
each iteration of the Ricci flow process dynamically alters the weights of hyperedges based on their current weights and 
curvature values.
On average, these weight alterations tend to 
increase the weights of the hyperedges connecting the core(s) to the rest of the hypergraph 
while decreasing the weights of those hyperedges within the core(s). 
The weight re-normalization procedure after every iteration ensures that the weights of the hyperedges 
do not increase in an unbounded manner
and instead eventually converge to some ``steady-state'' values.
The goal of the topological surgery procedure performed once in a few iterations is to remove the hyperedges 
connecting the core(s) to the rest of the hypergraph so that at the end of our Ricci flow process,
when the hyperedge weights have converged to some stable values,
we can recover the core(s) from the connected components of the hypergraph.
This process of topological surgeries and hyperedge weight updates via Ricci flows 
is somewhat analogous to the Newman-Girvan's algorithm~\cite{comm5} which,  in the context of undirected graphs, 
iteratively removes edges of high betweenness centrality  and recomputes the edge betweenness centrality.

We now present the precise technical descriptions of these concepts.
Let $t=0,1,2,\cdots$ 
is the discrete iteration index, 
and let 
$H^{(0)}=(V^{(0)},E^{(0)},w^{(0)}),
H^{(1)}=(V^{(1)},E^{(1)},w^{(1)}),
H^{(2)}=(V^{(2)},E^{(2)},w^{(2)}),\cdots$
denote the sequence of hypergraphs produced by the Ricci flow
with 
$H^{(0)}$ being equal to the original (starting) hypergraph 
$H=(V,E,w)$.
The Ricci flow equation (\emph{without} topological surgeries and hyperedge-weight renormalization)
is as follows~\cite{NLLG19,SJB19,LBL22}: 
%%%%%%%%%%%%%%%%%%%%%%%%%%%%%%%%%%%%%%%%%%%%%%%%%%%%%%%%%%%%%%%%%%%%%%%%%%%%%%%%
\begin{empheq}[box=\Ovalbox]{gather}
w^{(t+1)}(e)
=
w^{(t)}(e) - 
w^{(t)}(e) 
\times
\mbox{\ric}^{(t)}(e)
\label{eq-ricci}
\end{empheq}
%%%%%%%%%%%%%%%%%%%%%%%%%%%%%%%%%%%%%%%%%%%%%%%%%%%%%%%%%%%%%%%%%%%%%%%%%%%%%%%%
where 
$\mbox{\ric}^{(t)}(e)$
is the curvature value based on the edge-weights 
$W^{(t)}
=
\left\{ w^{(t)}(e) \,|\, e\in E^{(t)} \right\}
$.
Note that if $w^{(t)}(e)=0$ for some value $t=t_0$ then 
$w^{(t)}(e)$ stays zero for all $t>t_0$ based on~\eqref{eq-ricci},
so we may simply remove such edges $e$ from all $E^{(t)}$ with $t\geq t_0$.
Also note that 
$w^{(t)}(e)\geq 0$ for all $t$
since 
$\mbox{\ric}^{(t)}(e)\leq 1$
for all $t$.

Unfortunately, as observed in papers such as~\cite{LBL22}, there \emph{are} problematic aspects to 
the Ricci flow equation in~\eqref{eq-ricci}.  
In particular, in applications such as in this article, we would like the iterations
to eventually \emph{converge} (within a reasonable time), but 
it can be easily seen that there exists hypergraphs for which 
the hyperedge weights may keep on increasing in successive iterations.
As a remedy for graphs only, 
Lai, Bai and Lin~\cite{LBL22}
proposed changing~\eqref{eq-ricci} to~\eqref{eq-ricci-norm} as shown below
such that the Ricci flow is ``normalized''
in the sense that
the sum of edge weights of the graph remain the same and therefore edge weights \emph{cannot} become 
arbitrarily large:
%%%%%%%%%%%%%%%%%%%%%%%%%%%%%%%%%%%%%%%%%%%%%%%%%%%%%%%%%%%%%%%%%%%%%%%%%%%%%%%%
\begin{empheq}[box=\Ovalbox]{multline}
%%%%%%%%%%%%%%%%%%%%%%%%%%%%%%%%%%%%%%%%%%%%%%%%%%%%%%%%%%%%%%%%%%%%%%%%%%%%%%%%
w^{(t+1)}(e)
=
w^{(t)}(e) - 
w^{(t)}(e) 
\times
\mbox{\ric}^{(t)}(e)
\\
\textstyle
+
\frac{ 
       s \times w^{(t)}(e) \times 
			       \left(
              \sum_{h\in E} 
               \left( 
                  w^{(t)}(h) 
                  \times
                  \mbox{\ric}^{(t)}(h)
	             \right)
             \right)
		 }
    { \sum_{h\in E} \left( 
         w^{(0)}(h) 
         \times
         \mbox{\ric}^{(0)}(h)
       \right)
		}
\label{eq-ricci-norm}
%%%%%%%%%%%%%%%%%%%%%%%%%%%%%%%%%%%%%%%%%%%%%%%%%%%%%%%%%%%%%%%%%%%%%%%%%%%%%%%%
\end{empheq}
%%%%%%%%%%%%%%%%%%%%%%%%%%%%%%%%%%%%%%%%%%%%%%%%%%%%%%%%%%%%%%%%%%%%%%%%%%%%%%%%
In the above equation, $s>0$ is a \emph{constant} (called ``step size'' in~\cite{LBL22}).
Unfortunately, as we will prove in 
Theorem
in Section~\ref{sec-res-normprob},
there are infinitely many graphs
for which 
$w^{(1)}(e)$ will become negative thus rendering the iterative process in~\eqref{eq-ricci-norm}
\emph{impossible} to execute beyond the first step.
Instead, in our algorithm we 
perform hyperedge-weight re-normalization by 
applying a sigmoidal function to hyperedge weights to ensure that \emph{all} hyperedge weights do \emph{not} exceed $1$, 
\IE, 
for $t\geq 1$ we replace 
$w^{(t)}(e)$
by 
$\frac{1}{1+\bee^{-w^{(t)} (e)}}$.
{\bf In the sequel, unless explicitly mentioned otherwise, when we refer to weight 
$\pmb{w^{(t)}(e)}$
for $\pmb{t\geq 1}$ we refer to the weight after re-normalization.}

Unfortunately, there is \emph{no} closed-form analytical solution of Equation~\eqref{eq-ricci}
yet and it is \emph{not} even clear if such a solution is possible.
Here we introduce our topological surgery operation, and 
provide an \emph{informal} intuition behind our approach of finding cores 
by using Ricci flows with 
topological surgery.
Note that since all hyperedge weights are non-negative at \emph{all} times, 
Equation~\eqref{eq-ricci}
shows that 
$w^{(t+1)}(e)<w^{(t)}(e)$ 
if 
$\mbox{\ric}^{(t)}(e)>0$ 
and 
$w^{(t+1)}(e)>w^{(t)}(e)$ 
if 
$\mbox{\ric}^{(t)}(e)<0$, \IE,  
each iteration of the Ricci flow process dynamically alters the weights of hyperedges based on their curvature values,
increasing the weights of those with negative curvature while decreasing the weights of those with positive curvature. 
Using the observation
in~\cite{SJB19} that states (quoted verbatim) 
``{positively curved edges are well connected in the sense that none of them are essential for the proper transport
operation}'',
which is also supported by research works such as~\cite{NLLG19,SJB19} on graphs, 
it follows that 
this effect should on an average lead to 
pairs of nodes within a core being connected by hyperedges with a smaller weight whereas
pairs of nodes outside cores being connected by hyperedges with a larger weight.
Consequently, 
we can use the following ``topological surgery'' method
to isolate the core(s) from the remaining parts of the hypergraph:
\emph{remove hyperedges with substantial weights
following every several iterations of the Ricci flow}. 
This strategic manipulation enhances the clarity of core structures within the network. 
Moreover, since each hyperedge in a hypergraph typically involves many nodes, 
surgical removal of hyperedges may disconnect more nodes (as compared to graphs in which
each edge always involves two nodes), thus leading to the survival of a few very well-connected
sets of nodes as cores.
See \FI{fig-intuit} for a visual illustration of some of these intuitions.

%%%%%%%%%%%%%%%%%%%%%%%%%%%%%%%%%%%%%%%%%%%%%%%%%%%%%%%%%%%%%%%%%%%%%%%%%%%%%%%%%%%%%%
\begin{figure*}
\caption{\label{fig-intuit}A visual illustration of the intuition behind our approach of finding cores 
by using Ricci flows with topological surgery as discussed in 
Section~\ref{sec-flow-surg}. The notation ``$\approx 0''$ refers to a function $f(n)$ such that 
$\lim_{n\to\infty}f(n)=0$.
The nodes are colored blue and red for visual clarity: 
red nodes are involved in cliques of hyperedges of two nodes (i.e., cliques of edges) and all the blue nodes together with an equal
number of red nodes appear in a single hyperedge 
of $2n$ nodes.
The cliques are enclosed by dotted black bounding boxes for visual clarity 
(the cliques do not correspond to hyperedges).
The second figure from top indicates the hypergraph after one iteration of Ricci flow but before the weight renormalization. The thicknesses of the red edges are reduced to indicate 
the decrease of their weights from approximately $1$ to approximately $0$ and the thickness of
the black hyperedge is increased to indicate an increase of its weight.
The third figure from top shows that the black hyperedge of $2n$ nodes gets deleted 
as a result of weight renormalization and topological surgery, thus giving us 
the $n$ cores corresponding to the $n$ cliques.
}
\includegraphics[width=0.75\textwidth]{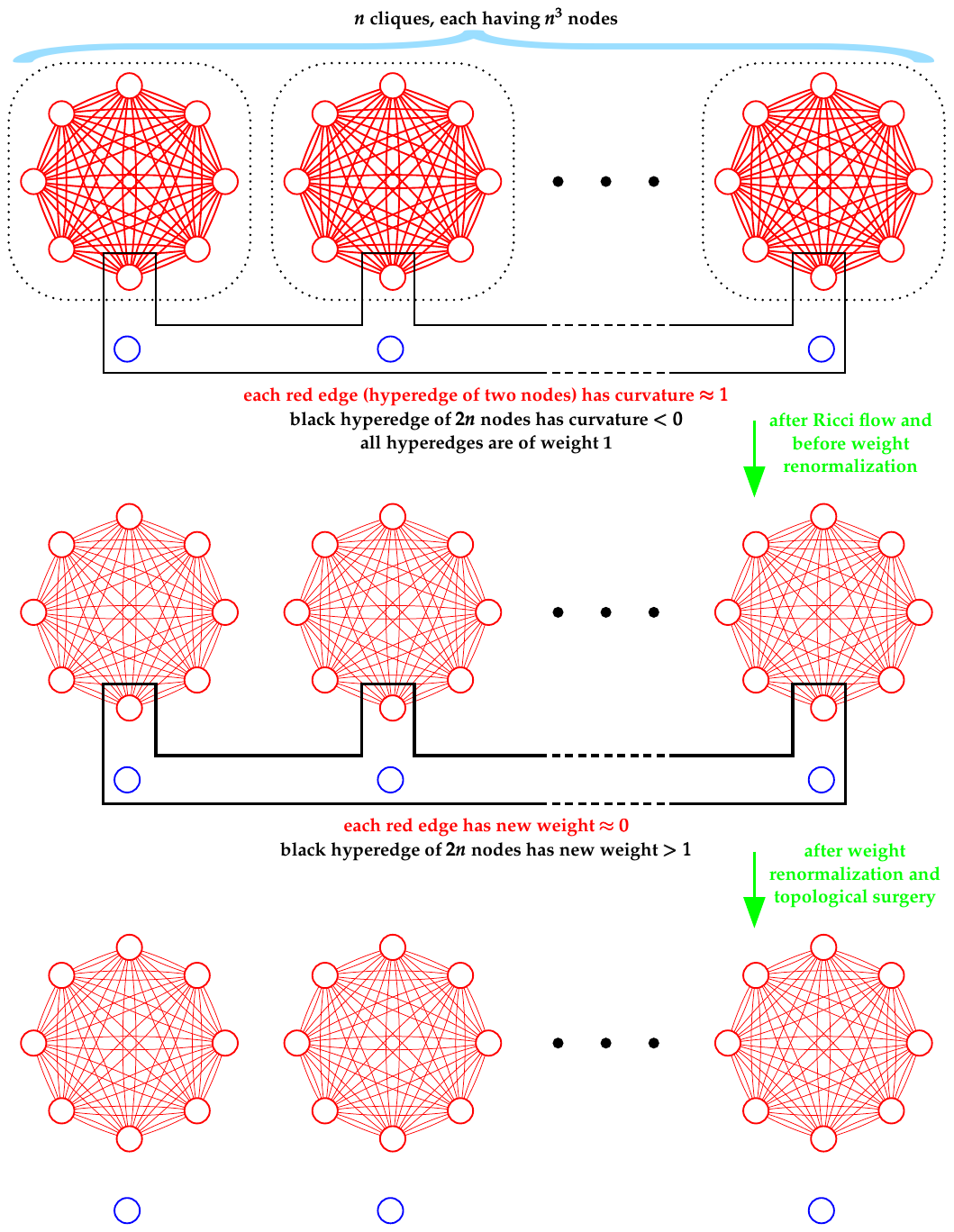}
\end{figure*}  
%%%%%%%%%%%%%%%%%%%%%%%%%%%%%%%%%%%%%%%%%%%%%%%%%%%%%%%%%%%%%%%%%%%%%%%%%%%%%%%%%%%%%%

For our experiments, 
we did surgery every $2$ iterations and ran our algorithm for a total of $40$ iterations in total for every hypergraph. 
The amount of hyperedges to be removed is an adjustable parameter that can be tweaked accordingly to get cores of different sizes. 
For our experiments, we set our ``surgery amount'' to be 
those hyperedges that have weights in the largest 
$8\%$ of the weights of all hyperedges in the previous iteration.
These combinations of adjustable parameters provided us with a reasonable combination of 
rapid rate of convergence and acceptable core quality parameters.

To check if the edge-weights have converged to a fixed-point 
we use a standard convergence criterion in which the \emph{average} of the \emph{absolute} differences 
of edge-weights in successive iterations is sufficiently small, \IE, 
%%%%%%%%%%%%%%%%%%%%%%%%%%%%%%%%%%%%%%%%%%%%%%%%%%%%%%%%%%%%%%%%%%%%%%%%%%%%%%%%%%%%%%
\begin{empheq}[box=\Ovalbox]{gather}
\Delta_{\mathrm{AVE}}=
\frac{1}{|E^{(t)}|} \times \sum_{e\in E^{(t)}} |w^{(t+1)}(e)-w^{(t)}(e)|
\label{eq-converg}
\end{empheq}
%%%%%%%%%%%%%%%%%%%%%%%%%%%%%%%%%%%%%%%%%%%%%%%%%%%%%%%%%%%%%%%%%%%%%%%%%%%%%%%%%%%%%%
is at most $\eps$
for some small real number $\eps\ge 0$.
We use 
$\eps=0.005$ for directed hypergraph applications 
and 
$\eps=0.000005$ for undirected hypergraph applications.
To check the dispersion of these absolute differences around mean we calculate the standard
deviation, \IE, 
%%%%%%%%%%%%%%%%%%%%%%%%%%%%%%%%%%%%%%%%%%%%%%%%%%%%%%%%%%%%%%%%%%%%%%%%%%%%%%%%%%%%%%
\begin{empheq}[box=\Ovalbox]{gather}
\Delta_{\mathrm{STD}}=
\!\!
\sqrt{
\frac{1}{|E^{(t)}|}  \times \!\!\! \sum_{e\in E^{(t)}} \!\!\! 
     \left( \Delta_{AVE} - |w^{(t+1)}(e)-w^{(t)}(e)| \right)^2
}
\label{eq-converg-2}
\end{empheq}
%%%%%%%%%%%%%%%%%%%%%%%%%%%%%%%%%%%%%%%%%%%%%%%%%%%%%%%%%%%%%%%%%%%%%%%%%%%%%%%%%%%%%%

%%%%%%%%%%%%%%%%%%%%%%%%%%%%%%%%%%%%%%%%%%%%%%%%%%%%%%%%%%%%%%%%%%%%%%%%%%%%%%%%%%%%%%
\subsection{\label{sec-quality}Quality Measures of Influential cores}
%%%%%%%%%%%%%%%%%%%%%%%%%%%%%%%%%%%%%%%%%%%%%%%%%%%%%%%%%%%%%%%%%%%%%%%%%%%%%%%%%%%%%%

A {\em core} 
is a subset of nodes 
that are more connected to 
each other
as opposed to the rest of the hypergraph.
In this article, 
following the approach in~\cite{ADM14}
we also want our cores for hypergraphs to be \emph{central} and \emph{influential} in the sense that 
removal of the nodes in the core \emph{significantly} disrupts short paths 
between nodes \emph{not} in the core.
This leads to the following quality constraints and parameters for a core
that generalizes similar conventions used by the network science community
for graphs.

%%%%%%%%%%%%%%%%%%%%%%%%%%%%%%%%%%%%%%%%%%%%%%%%%%%%%%%%%%%%%%%%%%%%%%%%%%%%%%%%%%%%%%
\subsubsection{\label{sec-conn}Connectivity Constraint}
%%%%%%%%%%%%%%%%%%%%%%%%%%%%%%%%%%%%%%%%%%%%%%%%%%%%%%%%%%%%%%%%%%%%%%%%%%%%%%%%%%%%%%

For an undirected hypergraph 
(respectively, directed hypergraph) 
a core is a connected component
(respectively, a weakly connected component)
of the hypergraph.
This is a bare minimum constraint that a core should satisfy.

%%%%%%%%%%%%%%%%%%%%%%%%%%%%%%%%%%%%%%%%%%%%%%%%%%%%%%%%%%%%%%%%%%%%%%%%%%%%%%%%%%%%%%
\subsubsection{\label{sec-size}Size Constraint}
%%%%%%%%%%%%%%%%%%%%%%%%%%%%%%%%%%%%%%%%%%%%%%%%%%%%%%%%%%%%%%%%%%%%%%%%%%%%%%%%%%%%%%

The size (number of nodes) of a core should be \emph{non-trivial}, \IE, neither too small nor too large.
For example, a core containing more than $50\%$ of all nodes 
or containing only $5$ nodes 
is \emph{hardly} an interesting core.
For the real hypergraphs investigated in this paper, 
our algorithm \emph{always} produces 
only one or two cores of non-trivial size
(\emph{cf}.\ Table~\ref{tab-finmod-dir} and
Table~\ref{tab-finmod-undir}) in the sense that each of the remaining connected components 
has \emph{very} few nodes.

%%%%%%%%%%%%%%%%%%%%%%%%%%%%%%%%%%%%%%%%%%%%%%%%%%%%%%%%%%%%%%%%%%%%%%%%%%%%%%%%%%%%%%
\subsubsection{\label{sec-cohes}Cohesiveness Measure}
%%%%%%%%%%%%%%%%%%%%%%%%%%%%%%%%%%%%%%%%%%%%%%%%%%%%%%%%%%%%%%%%%%%%%%%%%%%%%%%%%%%%%%

This goal of quantifying this measure is to ensure that the nodes in the core should 
be connected more among themselves as opposed to nodes 
outside the core.
For our hypergraphs, we quantify this in the following manner.

%%%%%%%%%%%%%%%%%%%%%%%%%%%%%%%%%%%%%%%%%%%%%%%%%%%%%%%%%%%%%%%%%%%%%%%%%%%%%%%%%%%%%%
\bigskip
\noindent
\textbf{Undirected hypergraph}
\medskip
%%%%%%%%%%%%%%%%%%%%%%%%%%%%%%%%%%%%%%%%%%%%%%%%%%%%%%%%%%%%%%%%%%%%%%%%%%%%%%%%%%%%%%

For an undirected hypergraph $H=(V,E,w)$, a non-empty proper subset $V'$ of $V$ and 
a node $x\in V$, 
let the notation
$H\setminus V'$ denote the hypergraph $(V\setminus V',E',w)$ 
obtained by removing the nodes in $V'$ from $V$ and removing any hyperedge $e$ 
with 
$\cA_e\cap V' \neq\emptyset$
from $E$.
Let 
the notation
$
\deg_x(H) 
$
denote the degree
of $x$ in $H$.
%%%%
For undirected graphs, 
several measures of cohesiveness have been used 
in prior published literatures~\cite{KOUJAKU2016143},
\EG, 
via \emph{distance}, 
via \emph{degree}, 
via \emph{density}, 
\emph{etc}.
%%%
However, most prior published articles dealing with finding cores or core-like structures  
for undirected hypergraphs
use a simple 
``degree'' cohesive measure by extending the 
concept of a $k$-core (or its minor variations)
from undirected graphs to undirected 
hypergraphs~\cite{doi:10.1137/22M1480926,pmlr-v84-chien18a,doi:10.1126-sciadv.adg9159,Carletti_2021,Erikss21,Krits24,doi:10.1080-01621459-2021-2002157,Eriksson2022,MIPB23,MIPB24,YiqunXu2023}.
In the notations and terminologies of this article, 
a $k$-core of
an undirected hypergraph $H=(V,E,w)$
is a subset of nodes $\cS$ such that 
$\deg_x(H \setminus (V\setminus \cS) )\geq k$
for every node $x\in\cS$; 
larger values of $k$ signify a better core.
In particular, a $1$-core trivially exists in any hypergraph and therefore is not considered a core at all.

We however believe that the above-mentioned $k$-core measure is not very suitable for the type of undirected 
hypergraphs studied in this article (namely, co-author hypergraphs (see Section~\ref{sec-coauth-data})
and similar other social interaction hypergraphs).
The condition of a $k$-core is too strict and changes the quality of the core too abruptly. 
For example, if a node $x$ in the core $\cS$ 
satisfies
$
\deg_x(H \setminus (V\setminus \cS) )= k-10
$
then $\cS$ defines a $(k-10)$-core even if every remaining node $y$ in $\cS$ satisfies 
$
\deg_y(H \setminus (V\setminus \cS) )=k
$.
Density-based measures centered on average degrees have been used 
extensively in the computer science literature for 
graphs~\cite{BONCHI202134,10.1145-3366423-3380140,doi:10.1137/1.9781611977073.64,10.14778/3554821.3554895,9201321,10184843,10.14778/3551793.3551826,10.1145/3483940,FPK01,10.1145/1806689.1806719}, and in one instance for undirected 
hypergraphs~\cite{10.1145/3485447.3512158}.
Following these research works, 
we consider a measure of cohesion based on 
average degrees.
%%%%%
Note that for the example mentioned above
the average degree changes from 
$\frac{k}{|\cS|}$ 
to 
$
\frac{k\times(|\cS|-1) + 1\times (k-10)}{|\cS|}
=
k-\frac{10}{|\cS|}
$,
showing a 
more gradual 
deterioration of the quality with an increasing value of $|\cS|$.
However, one should also account for
the nodes outside of $\cS$.
For example, for the node $x$ it is possible that 
either
\textbf{(\emph{a})}
$
\deg_x(H \setminus (V\setminus \cS) )
=
\deg_x(H)=k-10
$, 
or 
\textbf{(\emph{b})}
$
\deg_x(H \setminus (V\setminus \cS) )
=k-10
$
but 
$\deg_x(H)=k-\alpha$ for some $\alpha<10$. 
Our cohesiveness measure should indicate better cohesiveness for case 
\textbf{(\emph{a})}
as opposed to case  
\textbf{(\emph{b})}, 
and 
for case \textbf{(\emph{b})}
our cohesiveness measure should indicate worse cohesiveness with increasing $\alpha$.
Thus, we use the following measure for cohesiveness of a core $\cS$:
%%%%%%%%%%%%%%%%%%%%%%%%%%%%%%%%%%%%%%%%%%%%%%%%%%%%%%%%%%%%%%%%%%%%%%%%%%%%%%%%%%%%%%
\begin{empheq}[box=\Ovalbox]{gather}
\mathscr{r}^{deg}
=
\frac {
\sum_{x\in\cS} \frac { \deg_x(H\setminus(V\setminus \cS) ) } { \deg_x(H) }
}
{ |\cS| }
\label{eq-r}
\end{empheq}
%%%%%%%%%%%%%%%%%%%%%%%%%%%%%%%%%%%%%%%%%%%%%%%%%%%%%%%%%%%%%%%%%%%%%%%%%%%%%%%%%%%%%%
In the above equation, 
$
H\setminus(V\setminus \cS)
$
is the sub-hypergraph of $H$ induced by the nodes in $\cS$, 
$\deg_x(H\setminus(V\setminus \cS) )$
is the degree of node $x$ in this induced sub-hypergraph, 
the ratio 
${ \deg_x(H\setminus(V\setminus \cS) ) } / { \deg_x(H) }\in[0,1]$
provides a measure of how much the node $x$ is connected 
to only nodes in $\cS$ when we exclude all connections to 
nodes outside of $\cS$ via hyperedges, and 
the entire equation averages out the ratio over nodes in $\cS$.
Simple calculations show that 
for case \textbf{(\emph{b})}
of our example
$
\mathscr{r}^{deg}
=\frac{ (k-1)\times 1 + 1 \times \frac{k-10}{k-\alpha} } { k }
=
1 - \frac{10 - \alpha} {k\times (k-\alpha)}
$,
which decreases with increasing $\alpha$, as desired.
%%%%%%%%%%%%%%
Note that 
obviously 
$0\leq \mathscr{r}^{deg} \leq 1$.
%%%%%%%%
Any value of 
$\mathscr{r}^{deg}$
in the range $(\nicefrac{1}{2},1]$ 
with a statistical significance indicator $p$-value 
(see Section~\ref{sec-pcalc})
below $10^{-5}$ 
indicates a valid core
since in that case the nodes in the core on average are connected more to other
nodes in the core as opposed to nodes outside the core
and 
this property is not satisfied by the null hypothesis model; 
in other words, a core found by any algorithm will be considered to
be invalid if either
$\mathscr{r}^{deg}\leq 0.5$
or if the $p$-value 
is greater than or equal to $10^{-5}$. 
Higher values of
$\mathscr{r}^{deg}$
indicate better cohesiveness of the core.

%%%%%%%%%%%%%%%%%%%%%%%%%%%%%%%%%%%%%%%%%%%%%%%%%%%%%%%%%%%%%%%%%%%%%%%%%%%%%%%%%%%%%%
\bigskip
\noindent
\textbf{Directed hypergraph}
\medskip
%%%%%%%%%%%%%%%%%%%%%%%%%%%%%%%%%%%%%%%%%%%%%%%%%%%%%%%%%%%%%%%%%%%%%%%%%%%%%%%%%%%%%%

For a directed hypergraph $H=(V,E,w)$, a non-empty proper subset $V'$ of $V$ and 
a node $x\in V$, 
let the notation
$H\setminus V'$ denote the hypergraph $(V\setminus V',E',w)$ 
obtained by removing the nodes in $V'$ from $V$ and removing any hyperedge $e$ 
with 
$(\cTa_e\cup\cHe_e)\cap V'\neq\emptyset$
from $E$,
and 
the notations
$
\deg_x^{in}(H), 
\deg_x^{in}(H\setminus V'), 
\deg_x^{in}(out)
$
and $\deg_x^{out}(H\setminus V')$
denote the corresponding in-degrees and out-degrees
of $x$ in $H$ and $H\setminus V'$.
%%%%%%%%%%%%%%
As we mentioned already, we could not find existing peer-reviewed published 
materials for finding cores in directed weighted or unweighted hypergraphs, and 
thus \emph{no} prior cohesive measures for directed hypergraphs were available.
Our cohesiveness quality measures for directed hypergraphs are a direct generalization 
the cohesiveness quality measure for undirected hypergraphs as stated in Equation~\eqref{eq-r}.
%%%
Due to the directionality of a hyperedge in a directed hypergraph, 
we get \emph{two} cohesive parameters.
For a directed hypergraph $H=(V,E,w)$
our cohesiveness measures 
for a core $\cS$ are 
the following two 
values:
%%%%%%%%%%%%%%%%%%%%%%%%%%%%%%%%%%%%%%%%%%%%%%%%%%%%%%%%%%%%%%%%%%%%%%%%%%%%%%%%%%%%%%
\begin{empheq}[box=\Ovalbox]{gather}
\mathscr{r}_{\mathrm{in}}^{deg}
=
\frac {
\sum_{x\in\cS} \frac { \deg_x^{in}(H\setminus(V \setminus \cS) ) } { \deg_x^{in}(H) }
}
{ |\cS| }
\label{eq-rin}
\end{empheq}
%%%%%%%%%%%%%%%%%%%%%%%%%%%%%%%%%%%%%%%%%%%%%%%%%%%%%%%%%%%%%%%%%%%%%%%%%%%%%%%%%%%%%%
\begin{empheq}[box=\Ovalbox]{gather}
\mathscr{r}_{\mathrm{out}}^{deg}
=
\frac {
\sum_{x\in\cS} \frac { \deg_x^{out}(H\setminus(V \setminus \cS) ) } { \deg_x^{out}(H) }
}
{ |\cS| }
\label{eq-rout}
\end{empheq}
%%%%%%%%%%%%%%%%%%%%%%%%%%%%%%%%%%%%%%%%%%%%%%%%%%%%%%%%%%%%%%%%%%%%%%%%%%%%%%%%%%%%%%
In the above two equations,
$
H\setminus(V\setminus \cS)
$
is the directed sub-hypergraph of $H$ induced by the nodes in $\cS$, 
$\deg_x^{in}(H\setminus(V \setminus \cS) )$ (respectively, $\deg_x^{out}(H\setminus(V \setminus \cS) )$)
is the in-degree (respectively, out-degree) of node $x$ in this induced directed sub-hypergraph, 
the ratio ${ \deg_x^{in}(H\setminus(V \setminus \cS) ) } / { \deg_x^{in}(H) }\in[0,1]$
(respectively, the ratio ${ \deg_x^{out}(H\setminus(V \setminus \cS) ) } / { \deg_x^{in}(H) }\in[0,1]$)
provides a measure of how much 
the node $x$ is connected only to nodes in $\cS$ 
(respectively, nodes in $\cS$ are connected to the node $x$) 
when we exclude all connections to 
nodes outside of $\cS$ via directed hyperedges, and 
the entire equation averages out the ratio over nodes in $\cS$.
For example, 
$\frac { \deg_x^{in}(H\setminus(V \setminus \cS) ) } { \deg_x^{in}(H) }=0.7$
indicates that $70\%$ of the 
the in-degree of node $x$ is contributed by hyperedges that contain only nodes from $\cS$ both in their head and tail and 
only $30\%$ of the hyperedges contributing to the in-degree of $x$ have one or more outside nodes either in their head or in their tail.
Note that obviously 
$0\leq \mathscr{r}_{\mathrm{in}}^{deg},\, \mathscr{r}_{\mathrm{out}}^{deg} \leq 1$.
Again for a similar reason as in the undirected case, 
values of \emph{both} 
$\mathscr{r}_{\mathrm{in}}^{deg}$ and 
$\mathscr{r}_{\mathrm{out}}^{deg}$
in the range $(\nicefrac{1}{2},1]$ 
with a statistical significance indicator $p$-value 
(see Section~\ref{sec-pcalc})
below $10^{-5}$ 
indicate a valid core;
in other words, a core found by any algorithm will be considered to
be invalid if the values of at least one of 
$\mathscr{r}_{\mathrm{in}}^{deg}$
or
$\mathscr{r}_{\mathrm{out}}^{deg}$
is at most $0.5$ 
or if at least one of their $p$-values  
is greater than or equal to $10^{-5}$. 
Higher values of
$\mathscr{r}_{\mathrm{in}}^{deg}$ and 
$\mathscr{r}_{\mathrm{out}}^{deg}$
indicate better cohesiveness of the core.

%%%%%%%%%%%%%%%%%%%%%%%%%%%%%%%%%%%%%%%%%%%%%%%%%%%%%%%%%%%%%%%%%%%%%%%%%%%%%%%%%%%%%%
\subsubsection{\label{sec-centrality}Centrality Measure}
%%%%%%%%%%%%%%%%%%%%%%%%%%%%%%%%%%%%%%%%%%%%%%%%%%%%%%%%%%%%%%%%%%%%%%%%%%%%%%%%%%%%%%

These types of measures are a ``combinatorial analog'' of the \emph{information loss} quantifications 
used in prior research works such as~\cite{KITAZONO2020232}, 
and 
are significant in real-world applications to biological and social networks
as mentioned in prior publications such as~\cite{ADM14} in the context of undirected unweighted graphs.
For our hypergraphs, these measures quantify centrality and influential nature of the core in the sense that 
removal of the nodes in the core significantly disrupts short paths 
between nodes not in the core.
In the definitions below 
we use the notation $H\setminus V'$ as defined in 
Section~\ref{sec-cohes}.

%%%%%%%%%%%%%%%%%%%%%%%%%%%%%%%%%%%%%%%%%%%%%%%%%%%%%%%%%%%%%%%%%%%%%%%%%%%%%%%%%%%%%%
\bigskip
\noindent
\textbf{Directed hypergraph}
\medskip
%%%%%%%%%%%%%%%%%%%%%%%%%%%%%%%%%%%%%%%%%%%%%%%%%%%%%%%%%%%%%%%%%%%%%%%%%%%%%%%%%%%%%%

Let 
$H=(V,E,w)$ 
be the  directed hypergraph 
and 
$\emptyset\subset\cS\subset V$
be the core we are evaluating.
Our \emph{first} goal is to 
measure the fraction of ordered pairs of nodes for which there \emph{was} a path in the given input hypergraph but there 
\emph{no longer} is a path after removing the core. 
Let $\zeta$ denote the number of ordered pairs $(u,v)$ of nodes $u,v$ not in \emph{any} core
for which there was a (directed) path from $u$ to $v$ in $H$ but no (directed) path in $H\setminus\cS$.
We then calculate the 
following quantity:
%%%%%%%%%%%%%%%%%%%%%%%%%%%%%%%%%%%%%%%%%%%%%%%%%%%%%%%%%%%%%%%%%%%%%%%%%%%%%%%%%%%%%%
\begin{empheq}[box=\Ovalbox]{gather}
\mathscr{r}_{\mathrm{disconnected}}^{\mathrm{ordered\_pairs}}
=
\frac{\zeta}{
|V\setminus \cS| \times
(|V\setminus \cS| -1)
}
\label{eq-dir-discon}
\end{empheq}
%%%%%%%%%%%%%%%%%%%%%%%%%%%%%%%%%%%%%%%%%%%%%%%%%%%%%%%%%%%%%%%%%%%%%%%%%%%%%%%%%%%%%%
In addition, our \emph{second} goal is to 
measure the \emph{average percentage} increase in the length of paths among ordered pairs of nodes that remain connected
both before and after removing the core.
This is done as follows.
Let $\xi$ be the number of 
ordered pair $(u,v)$ of nodes $u,v$ not in \emph{any} core
for which there was a (directed) path from $u$ to $v$ in both $H$ and $H\setminus\cS$.
We then calculate the 
following quantity:
%%%%%%%%%%%%%%%%%%%%%%%%%%%%%%%%%%%%%%%%%%%%%%%%%%%%%%%%%%%%%%%%%%%%%%%%%%%%%%%%%%%%%%
\begin{empheq}[box=\Ovalbox]{gather}
\mathscr{r}_{\mathrm{dist\_stretch}}^{directed}
=
\frac{1}{\xi} \times 
\hspace*{-0.6in}
%%\frac {
	 \displaystyle
   \sum_{ \substack { 
	                   \,
										 (u,v)\in V\setminus\cS: 
										 \\ 
										 u\neq v 
										 \\
										 \hspace*{0.6in}
					            \dist_{H\setminus\cS}(u,v) < \infty
										 \\
										 \hspace*{0.5in}
					             \dist_{H}(u,v) < \infty
										 } } 
%%\hspace*{-0.2in}
\hspace*{-0.45in}
	         \frac {
					        \dist_{H\setminus\cS}(u,v)
					       }
								 {
					        \dist_{H}(u,v)
								 }
	      %%}
%%{
%%\xi
%%}
\label{eq-dir-dist}
\end{empheq}
%%%%%%%%%%%%%%%%%%%%%%%%%%%%%%%%%%%%%%%%%%%%%%%%%%%%%%%%%%%%%%%%%%%%%%%%%%%%%%%%%%%%%%
Note that every ordered pair of nodes from $V\setminus \cS$
appears in the calculation of either 
$\mathscr{r}_{\mathrm{disconnected}}^{\mathrm{ordered\_pairs}}$
or 
$\mathscr{r}_{\mathrm{dist\_stretch}}^{directed}$
but \emph{not} both, the reason being that 
%%%%%%%
incorporating the ordered pair of nodes used in~\eqref{eq-dir-discon}
in~\eqref{eq-dir-dist} instead 
would have make the value of 
$\mathscr{r}_{\mathrm{dist\_stretch}}^{directed}$
become $\infty$.
Note that 
$\mathscr{r}_{\mathrm{disconnected}}^{\mathrm{ordered\_pairs}}$ is at most $1$, 
and 
$\mathscr{r}_{\mathrm{dist\_stretch}}^{directed}$ is at least $1$
(since edge removal does not decrease the distance values).

The values of 
$\mathscr{r}_{\mathrm{dist\_stretch}}^{directed}$ and 
$\mathscr{r}_{\mathrm{disconnected}}^{\mathrm{ordered\_pairs}}$
are considered to be valid only if 
their statistical significance indicator $p$-values 
(see Section~\ref{sec-pcalc})
are below $10^{-5}$, and 
higher values of both 
$\mathscr{r}_{\mathrm{dist\_stretch}}^{directed}$ and 
$\mathscr{r}_{\mathrm{disconnected}}^{\mathrm{ordered\_pairs}}$ indicate stronger central and influential quality.
Note that a small value of 
$\mathscr{r}_{\mathrm{disconnected}}^{\mathrm{ordered\_pairs}}$
does not signify a weak-quality core as long as 
$\mathscr{r}_{\mathrm{dist\_stretch}}^{directed}$
is sufficiently large; 
however if 
$\mathscr{r}_{\mathrm{dist\_stretch}}^{directed}$
is not sufficiently large then 
$\mathscr{r}_{\mathrm{disconnected}}^{\mathrm{ordered\_pairs}}$
must be sufficiently large to signify the centrality of the core.
In this article, we adopt the strict criterion that 
for a valid core
either 
$\mathscr{r}_{\mathrm{dist\_stretch}}^{directed}$
must be at least $\nicefrac{3}{2}$ (\IE, shortest paths are stretched by at least $50\%$), or 
if 
$\mathscr{r}_{\mathrm{dist\_stretch}}^{directed}$
is below $\nicefrac{3}{2}$ 
then 
$\mathscr{r}_{\mathrm{disconnected}}^{\mathrm{ordered\_pairs}}$
must be at least $\nicefrac{1}{2}$
(\IE, at least $50\%$ of the ordered pairs of nodes are disconnected);
in other words, a core found by any algorithm will be considered to
be invalid if 
both 
$\mathscr{r}_{\mathrm{dist\_stretch}}^{directed}<\nicefrac{3}{2}$
and 
$\mathscr{r}_{\mathrm{disconnected}}^{\mathrm{ordered\_pairs}}<\nicefrac{1}{2}$.

%%%%%%%%%%%%%%%%%%%%%%%%%%%%%%%%%%%%%%%%%%%%%%%%%%%%%%%%%%%%%%%%%%%%%%%%%%%%%%%%%%%%%%
\bigskip
\noindent
\textbf{Undirected hypergraph}
\medskip
%%%%%%%%%%%%%%%%%%%%%%%%%%%%%%%%%%%%%%%%%%%%%%%%%%%%%%%%%%%%%%%%%%%%%%%%%%%%%%%%%%%%%%

Let 
$H=(V,E,w)$ 
be the undirected hypergraph 
and 
$\emptyset\subset\cS\subset V$
be the core we are evaluating.
Our \emph{first} goal is to 
measure the fraction of pairs of nodes for which there \emph{was} a path in the given input hypergraph but there 
\emph{no longer} is a path after removing the core. 
Let $\zeta$ denote the number of 
unordered pairs $\{u,v\}$ of nodes $u,v$ not in \emph{any} core
for which there is no path between them in $H\setminus\cS$.
We then calculate the 
following quantity:
%%%%%%%%%%%%%%%%%%%%%%%%%%%%%%%%%%%%%%%%%%%%%%%%%%%%%%%%%%%%%%%%%%%%%%%%%%%%%%%%%%%%%%
\begin{empheq}[box=\Ovalbox]{gather}
\mathscr{r}_{\mathrm{disconnected}}^{\mathrm{unordered\_pairs}}
=
\frac{\zeta}{
\binom{|V\setminus \cS|}{2}
}
\label{eq-undir-discon}
\end{empheq}
%%%%%%%%%%%%%%%%%%%%%%%%%%%%%%%%%%%%%%%%%%%%%%%%%%%%%%%%%%%%%%%%%%%%%%%%%%%%%%%%%%%%%%
In addition, our \emph{second} goal is to 
measure the \emph{average percentage} increase in the length of paths among unordered pairs of nodes that remain connected
both even after removing the core.
This is done as follows.
Let $\xi$ be the number of 
unordered pair $\{u,v\}$ of nodes $u,v$ not in \emph{any} core
for which there was still a path between $u$ to $v$ in $H\setminus\cS$.
We then calculate the 
following quantity:
%%%%%%%%%%%%%%%%%%%%%%%%%%%%%%%%%%%%%%%%%%%%%%%%%%%%%%%%%%%%%%%%%%%%%%%%%%%%%%%%%%%%%%
\begin{empheq}[box=\Ovalbox]{gather}
\mathscr{r}_{\mathrm{dist\_stretch}}^{undirected}
=
\frac{1}{\xi} \times 
\hspace*{-0.6in}
%%\frac {
	 \displaystyle
   \sum_{ \substack { 
										 \hspace*{0.1in}
	                   \{u,v\}\in V\setminus\cS: 
										 \\ 
										 u\neq v 
										 \\
										 \hspace*{0.6in}
					            \dist_{H\setminus\cS}(u,v) < \infty
										 } } 
\hspace*{-0.4in}
	         \frac {
					        \dist_{H\setminus\cS}(u,v)
					       }
								 {
					        \dist_{H}(u,v)
								 }
	      %%}
%%{
%%\xi
%%}
\label{eq-undir-dist}
\end{empheq}
%%%%%%%%%%%%%%%%%%%%%%%%%%%%%%%%%%%%%%%%%%%%%%%%%%%%%%%%%%%%%%%%%%%%%%%%%%%%%%%%%%%%%%
Note that every (unordered) pair of nodes from $V\setminus \cS$
appears in the calculation of either 
$\mathscr{r}_{\mathrm{disconnected}}^{\mathrm{unordered\_pairs}}$ 
or 
$\mathscr{r}_{\mathrm{dist\_stretch}}^{undirected}$ 
but \emph{not} both, as 
incorporating the pair of nodes used in~\eqref{eq-undir-discon}
in~\eqref{eq-undir-dist} instead
would yield
$\mathscr{r}_{\mathrm{dist\_stretch}}^{undirected}=\infty$
Note that  
$\mathscr{r}_{\mathrm{disconnected}}^{\mathrm{unordered\_pairs}}$ is at most $1$, 
and
$\mathscr{r}_{\mathrm{dist\_stretch}}^{undirected}$ is at least $1$
(since edge removal does not decrease the distance values).
%%%
The values of 
$\mathscr{r}_{\mathrm{dist\_stretch}}^{undirected}$ and 
$\mathscr{r}_{\mathrm{disconnected}}^{\mathrm{unordered\_pairs}}$
are considered to be valid only if 
their statistical significance indicator $p$-values 
(see Section~\ref{sec-pcalc})
are below $10^{-5}$, and 
higher values of both 
$\mathscr{r}_{\mathrm{dist\_stretch}}^{undirected}$ and 
$\mathscr{r}_{\mathrm{disconnected}}^{\mathrm{unordered\_pairs}}$ indicate stronger central and influential quality.
Note that a small value of 
$\mathscr{r}_{\mathrm{disconnected}}^{\mathrm{unordered\_pairs}}$
does not signify a weak-quality core as long as 
$\mathscr{r}_{\mathrm{dist\_stretch}}^{undirected}$
is sufficiently large;
however if 
$\mathscr{r}_{\mathrm{dist\_stretch}}^{undirected}$
is not sufficiently large then 
$\mathscr{r}_{\mathrm{disconnected}}^{\mathrm{unordered\_pairs}}$
must be sufficiently large to justify centrality of the core.
In this article, we adopt the strict criterion that 
for a valid core
either 
$\mathscr{r}_{\mathrm{dist\_stretch}}^{undirected}$
must be at least $\nicefrac{3}{2}$ (\IE, shortest paths are stretched by at least $50\%$), or 
if 
$\mathscr{r}_{\mathrm{dist\_stretch}}^{undirected}$
is below $\nicefrac{3}{2}$ 
then 
$\mathscr{r}_{\mathrm{disconnected}}^{\mathrm{unordered\_pairs}}$
must be at least $\nicefrac{1}{2}$
(\IE, at least $50\%$ of pairs of nodes are disconnected);
in other words, a core found by any algorithm will be considered to
be invalid if 
both 
$\mathscr{r}_{\mathrm{dist\_stretch}}^{undirected}<\nicefrac{3}{2}$
and 
$\mathscr{r}_{\mathrm{disconnected}}^{\mathrm{unordered\_pairs}}<\nicefrac{1}{2}$.

%%%%%%%%%%%%%%%%%%%%%%%%%%%%%%%%%%%%%%%%%%%%%%%%%%%%%%%%%%%%%%%%%%%%%%%%%%%%%%%%%%%%%%
\subsubsection{\label{sec-pcalc}Statistical Significance Measure:
Calculations of $p$-values for Core Quality Parameters}
%%%%%%%%%%%%%%%%%%%%%%%%%%%%%%%%%%%%%%%%%%%%%%%%%%%%%%%%%%%%%%%%%%%%%%%%%%%%%%%%%%%%%%

Statistical significance ($p$-value) calculations for core quality parameters require 
a null hypothesis model corresponding a random hypergraph similar in some essential 
characteristics to the one studied. 
We explain below why two most common methods for generating random graphs 
used by the network science community for $p$-value calculations \emph{fail} to generalize
to hypergraphs:
%%%%%%%%%%%%%%%%%%%%%%%%%%%%%%%%%%%%%%%%%%%%%%%%%%%%%%%%%%%%%%%%%%%%%%%%%%%%%%%%%%%%%%
\begin{description}
%%%%%%%%%%%%%%%%%%%%%%%%%%%%%%%%%%%%%%%%%%%%%%%%%%%%%%%%%%%%%%%%%%%%%%%%%%%%%%%%%%%%%%
\item[Generative models]
%%%%%%%%%%%%%%%%%%%%%%%%%%%%%%%%%%%%%%%%%%%%%%%%%%%%%%%%%%%%%%%%%%%%%%%%%%%%%%%%%%%%%%
The random graphs are generated so that they statistically match some key 
topological characteristics of the given graph such as 
node degree distributions for undirected graphs and 
distribution of in-degrees and out-degrees of nodes for directed graphs.
There are \emph{two} reasons that prevented us from using these methods for our hypergraphs.
First, there are \emph{no} broadly accepted evidences of topological characteristics 
such as degree distributions for metabolic and co-authorship hypergraphs.
Secondly, it is \emph{not} clear 
how we will generate random hypergraphs so that they statistically match key 
topological characteristics of the given hypergraph, \EG, the methods outlined 
in~\cite{LN08,NG04,comm1,comm2,comm5}
for generating random graphs with prescribed degree-distributions 
are \emph{not} easily generalizable to hypergraphs.
%%%%%%%%%%%%%%%%%%%%%%%%%%%%%%%%%%%%%%%%%%%%%%%%%%%%%%%%%%%%%%%%%%%%%%%%%%%%%%%%%%%%%%
\item[Random-swap models]
%%%%%%%%%%%%%%%%%%%%%%%%%%%%%%%%%%%%%%%%%%%%%%%%%%%%%%%%%%%%%%%%%%%%%%%%%%%%%%%%%%%%%%
For graphs, these kind of random graphs are generated using a 
Markov-chain algorithm~\cite{KTV99} by starting with the real graph and repeatedly swapping randomly chosen ``compatible''
pairs of edges.
However, it is not very clear if there is an useful generalization of this hypergraphs.
For graphs, an edge contributes exactly $1$ to the degrees of nodes at its endpoints, leading to
many compatible edges as candidates for swap and thus providing statistical validity of the model.
In contrast, hyperedges may contribute to the degree of an \emph{arbitrary} number of nodes in a more 
complicated fashion.
%%%%%%%%%%%%%%%%%%%%%%%%%%%%%%%%%%%%%%%%%%%%%%%%%%%%%%%%%%%%%%%%%%%%%%%%%%%%%%%%%%%%%%
\end{description}
%%%%%%%%%%%%%%%%%%%%%%%%%%%%%%%%%%%%%%%%%%%%%%%%%%%%%%%%%%%%%%%%%%%%%%%%%%%%%%%%%%%%%%
Based on the above observations, 
we design the following method to generate the $p$-values.
Let $H=(V,E,w)$ be the (directed or undirected) hypergraph, 
let $\cS\subset V$ be the core in question with 
$\alpha_1,\dots,\alpha_r$ being the values of its quality parameters
(for directed hypergraphs 
$\alpha_1,\alpha_2,\alpha_3,\alpha_4$ are the values of 
$\mathscr{r}_{\mathrm{in}}^{deg}$, 
$\mathscr{r}_{\mathrm{out}}^{deg}$, 
$\mathscr{r}_{\mathrm{dist\_stretch}}^{directed}$ 
and 
$\mathscr{r}_{\mathrm{disconnected}}^{\mathrm{ordered\_pairs}}$, respectively;
for undirected hypergraphs
$\alpha_1,\alpha_2,\alpha_3$ are the values of 
$\mathscr{r}^{deg}$,
$\mathscr{r}_{\mathrm{dist\_stretch}}^{undirected}$
and
$\mathscr{r}_{\mathrm{disconnected}}^{\mathrm{unordered\_pairs}}$, respectively).
We generate $100$ \emph{random} subsets of $V$, say 
$\cB_1,\dots,\cB_{100}$, 
such that 
$|\cB_1|=\dots=|\cB_{100}|=|\cS|$ 
and 
compute the values of 
$\beta_{i,j}$ for $i\in\{1,\dots,100\}$ and $j\in\{1,\dots,r\}$, 
where $\beta_{i,j}$ is the value of the $j\tx$ property for $\cB_i$.
We calculate the $p$-value for the $j\tx$ property 
by performing a \emph{one-sample} \texttt{t}-test 
with $\beta_{1,j},\dots,\beta_{100,j}$ as the values of the samples
and $\alpha_j$ as the value of the hypothesis. 

The $p$-value is a real number between $0$ and $1$; lower $p$-values indicate better 
statistical significance.
Following standard practice in network science, in this article we adopt a 
strict constraint on the acceptable $p$-values:
{\bf a $\pmb{p}$-value that is more than $\pmb{10^{-5}}$
even for a single quality measure for a core will invalidate the selection of that core.}

%%%%%%%%%%%%%%%%%%%%%%%%%%%%%%%%%%%%%%%%%%%%%%%%%%%%%%%%%%%%%%%%%%%%%%%%%%%%%%%%%%%%%%
\subsection{\label{sec-data-all}Data Sources}
%%%%%%%%%%%%%%%%%%%%%%%%%%%%%%%%%%%%%%%%%%%%%%%%%%%%%%%%%%%%%%%%%%%%%%%%%%%%%%%%%%%%%%

%%%%%%%%%%%%%%%%%%%%%%%%%%%%%%%%%%%%%%%%%%%%%%%%%%%%%%%%%%%%%%%%%%%%%%%%%%%%%%%%%%%%%%
\subsubsection{Metabolic Systems (for Directed Hypergraphs)}
%%%%%%%%%%%%%%%%%%%%%%%%%%%%%%%%%%%%%%%%%%%%%%%%%%%%%%%%%%%%%%%%%%%%%%%%%%%%%%%%%%%%%%

We collected seven metabolic systems from BiGG Models~\cite{king2016bigg}, 
a comprehensive public repository managed by the Systems Biology Research Group at UC San Diego. 
These seven metabolic systems pertained to the seven species 
Escherichia Coli,
Homo Sapiens,
Helicobacter Pylori,
Methanosarcina berkeri str.\ Fusaro,
Mycobacterium Tuberculosis, 
Synechococcus elongatus 
and
Synechocystis.

%%%%%%%%%%%%%%%%%%%%%%%%%%%%%%%%%%%%%%%%%%%%%%%%%%%%%%%%%%%%%%%%%%%%%%%%%%%%%%%%%%%%%%
%%\subsubsection{\label{sec-coauth-data}Co-authorships data (for Undirected Hypergraphs)}
\subsubsection{Co-authorships data (for Undirected Hypergraphs)}
%%%%%%%%%%%%%%%%%%%%%%%%%%%%%%%%%%%%%%%%%%%%%%%%%%%%%%%%%%%%%%%%%%%%%%%%%%%%%%%%%%%%%%
\newcommand{\csp}{{\sc Csp}}
\newcommand{\nsp}{{\sc Nsp}}

We build two undirected hypergraphs corresponding to two co-authorship datasets, 
which we will call the 
\emph{Computer Science Papers} (\csp) dataset and the \emph{Network Science Papers} (\nsp) dataset.
Each individual item in each dataset is a peer-reviewed publication in the respective (computer science 
or network science) research field.
Our datasets are constructed following similar approaches used by prior researchers 
such as~\cite{molontay2021twenty}.

%%%%%%%%%%%%%%%%%%%%%%%%%%%%%%%%%%%%%%%%%%%%%%%%%%%%%%%%%%%%%%%%%%%%%%%%%%%%%%%%%%%%%%
\paragraph{\csp\ dataset} 
%%%%%%%%%%%%%%%%%%%%%%%%%%%%%%%%%%%%%%%%%%%%%%%%%%%%%%%%%%%%%%%%%%%%%%%%%%%%%%%%%%%%%%
%
We selected three influential papers~\cite{Goldwasser1984ProbabilisticE,karp2010reducibility,nisan1994hardness}, 
by three \emph{Turing Award} winner researchers
S. Goldwasser, R. M. Karp and A. Wigderson working in the same general research area (\emph{Theoretical Computer Science}).
We then selected $300$ of the \emph{most} cited papers that cite each of these $3$ papers giving us a list of $900$ papers. 

%%%%%%%%%%%%%%%%%%%%%%%%%%%%%%%%%%%%%%%%%%%%%%%%%%%%%%%%%%%%%%%%%%%%%%%%%%%%%%%%%%%%%%
\paragraph{\nsp\ dataset}
%%%%%%%%%%%%%%%%%%%%%%%%%%%%%%%%%%%%%%%%%%%%%%%%%%%%%%%%%%%%%%%%%%%%%%%%%%%%%%%%%%%%%%
%
We selected three influential papers~\cite{BA99,comm5,WS:1998}
in network science that have been used by previous researchers for related research works~\cite{molontay2021twenty}.
We then selected $200$ of the \emph{most} cited papers that cite each of these $3$ papers giving us a list of $600$ papers. 

We faced a situation regarding co-authorships in network science that is usually not encountered
in computer science, mathematics or theoretical physics: 
there are papers co-authored by a \emph{large} number of authors. 
For example, 
there are $355$ and $141$ co-authors (after including corresponding consortium authors) 
respectively 
in the following two papers: 
\textbf{(\emph{a})}
\emph{Gene expression imputation across multiple brain regions provides insights into schizophrenia risk}, 
Nature Genetics 51, 659--674, 2019, 
and 
\textbf{(\emph{b})}
\emph{Brain structural covariance networks in obsessive-compulsive disorder: 
a graph analysis from the ENIGMA Consortium},
Brain 143(2), 684--700, 2020.
These kind of papers act as a bottleneck in the calculation of the Ricci curvature via 
Equation~\eqref{eq2} that was adopted from~\cite{CDR23}
(graph-theoretic methods as illustrated in \FI{fig333}
will \emph{not} be helpful either since they will include all or almost all of the $355$ or $141$ 
authors in the core).

Fortunately, there are \emph{only} $18$ of the $600$ papers in our collection ($3\%$ of all the papers) that were
co-authored by $15$ or more authors. 
We removed these papers from our dataset.

%%%%%%%%%%%%%%%%%%%%%%%%%%%%%%%%%%%%%%%%%%%%%%%%%%%%%%%%%%%%%%%%%%%%%%%%%%%%%%%%%%%%%%
\section{Results and Discussions}
%%%%%%%%%%%%%%%%%%%%%%%%%%%%%%%%%%%%%%%%%%%%%%%%%%%%%%%%%%%%%%%%%%%%%%%%%%%%%%%%%%%%%%

Before presenting our results in details in the following subsections, we 
provide a brief \emph{synopsis} of them below:
%%%%%%%%%%%%%%%%%%%%%%%%%%%%%%%%%%%%%%%%%%%%%%%%%%%%%%%%%%%%%%%%%%%%%%%%%%%%%%%%%%%%%%%%
\begin{enumerate}[label=$\triangleright$]
%%%%%%%%%%%%%%%%%%%%%%%%%%%%%%%%%%%%%%%%%%%%%%%%%%%%%%%%%%%%%%%%%%%%%%%%%%%%%%%%%%%%%%%%
\item
In Section~\ref{sec-algm-overall} we present our algorithm for finding core(s) 
in a hypergraph.
%%%%%%%%%%%%%%%%%%%%%%%%%%%%%%%%%%%%%%%%%%%%%%%%%%%%%%%%%%%%%%%%%%%%%%%%%%%%%%%%%%%%%%%%
\item
In Section~\ref{hyper-construct} we show how to construct 
our directed and undirected hypergraphs from the corresponding datasets mentioned 
in Section~\ref{sec-data-all}.
%%%%%%%%%%%%%%%%%%%%%%%%%%%%%%%%%%%%%%%%%%%%%%%%%%%%%%%%%%%%%%%%%%%%%%%%%%%%%%%%%%%%%%%%
\item
In Section~\ref{sec-res-normprob}
we prove a theorem
showing that 
there are infinitely many graphs for which 
the iterative process in~\eqref{eq-ricci-norm}
is \emph{impossible} to execute beyond the first step.
%%%%%%%%%%%%%%%%%%%%%%%%%%%%%%%%%%%%%%%%%%%%%%%%%%%%%%%%%%%%%%%%%%%%%%%%%%%%%%%%%%%%%%%%
\item
In Section~\ref{sec-res-conv}
we show that the initial convergence of Ricci flows 
based on the values of 
$\Delta_{\mathrm{AVE}}$
occurs within a \emph{small} number of iterations for all 
our hypergraphs. 
We then show that increasing the number of iterations further 
improves the values of 
$\Delta_{\mathrm{STD}}$
thereby making the cores \emph{more} stable.
%%%%%%%%%%%%%%%%%%%%%%%%%%%%%%%%%%%%%%%%%%%%%%%%%%%%%%%%%%%%%%%%%%%%%%%%%%%%%%%%%%%%%%%%
\item
In Section~\ref{sec-finmod}
we provide the final cores with their quality parameter values for all hypergraphs
as determined by our algorithm and observe that these values \emph{are} satisfactory.
%%%%%%%%%%%%%%%%%%%%%%%%%%%%%%%%%%%%%%%%%%%%%%%%%%%%%%%%%%%%%%%%%%%%%%%%%%%%%%%%%%%%%%%%
\begin{enumerate}[label=$\triangleright$]
%%%%%%%%%%%%%%%%%%%%%%%%%%%%%%%%%%%%%%%%%%%%%%%%%%%%%%%%%%%%%%%%%%%%%%%%%%%%%%%%%%%%%%%%
\item
In Subsection~\ref{sec-finmod-dir-int}
we discuss how to \emph{interpret} the cores and its quality parameters for the seven metabolic systems
corresponding to the seven directed hypergraphs.
We also briefly  
comment here on optimally designing the biological experimental mechanism to remove a core
in subsection~\ref{sec-rem-cor}.
%%%%%%%%%%%%%%%%%%%%%%%%%%%%%%%%%%%%%%%%%%%%%%%%%%%%%%%%%%%%%%%%%%%%%%%%%%%%%%%%%%%%%%%%
\item
In Subsection~\ref{sec-finmod-undir-int}
we discuss how to \emph{interpret} the cores and its quality parameters for the 
two co-authorship data corresponding to the two undirected hypergraphs.
%%%%%%%%%%%%%%%%%%%%%%%%%%%%%%%%%%%%%%%%%%%%%%%%%%%%%%%%%%%%%%%%%%%%%%%%%%%%%%%%%%%%%%%%
\end{enumerate}
%%%%%%%%%%%%%%%%%%%%%%%%%%%%%%%%%%%%%%%%%%%%%%%%%%%%%%%%%%%%%%%%%%%%%%%%%%%%%%%%%%%%%%%%
\end{enumerate}
%%%%%%%%%%%%%%%%%%%%%%%%%%%%%%%%%%%%%%%%%%%%%%%%%%%%%%%%%%%%%%%%%%%%%%%%%%%%%%%%%%%%%%%%

%%%%%%%%%%%%%%%%%%%%%%%%%%%%%%%%%%%%%%%%%%%%%%%%%%%%%%%%%%%%%%%%%%%%%%%%%%%%%%%%%%%%%%
\subsection{\label{sec-algm-overall}Overall Algorithmic Approach}
%%%%%%%%%%%%%%%%%%%%%%%%%%%%%%%%%%%%%%%%%%%%%%%%%%%%%%%%%%%%%%%%%%%%%%%%%%%%%%%%%%%%%%

%%%%%%%%%%%%%%%%%%%%%%%%%%%%%%%%%%%%%%%%%%%%%%%%%%%%%%%%%%%%%%%%%%%%%%%%%%%%%%%%%%%%%%
\begin{table*}
\caption{\label{algo-mod}High-level overview of our algorithmic approach to find core(s) in a directed or undirected 
hypergraph. 
See Section~\ref{sec-algm-overall} for comments on the adjustable parameters.
In our experiments, $\eta=40$, $\tau=2$, $\kappa=2$ and $\delta=8$.} 
%%%%%%%%%%%%%%%%%%%%%%%%%%%%%%%%%%%%%%%%%%%%%%%%%%%%%%%%%%%%%%%%%%%%%%%%%%%%%%%%%%%%%%
\begin{tabular}{r l l}
\toprule
\toprule
\textbf{Input:} & 
   \multicolumn{2}{l}{A directed or undirected hypergraph $H=(V,E,w)$.}
\\
\begin{tabular}{r}
\textbf{Adjustable}
\\
\textbf{parameters} 
\\
\end{tabular}:
              &  $\eta$, $\kappa$, $\tau$, $\delta$
\\
\textbf{Output:} & 
   \multicolumn{2}{l}{A set of $\mu\leq\kappa$ mutually disjoint node subsets (cores) $\cS_1,\dots,\cS_\mu\subset V$.}
%%%%%%%%%%%%%%%%%%%%%%%%%%%%%%%%%%%%%%%%%%%%%%%%%%%%%%%%%%%%%%%%%%%%%%%%%%%%%%%%%%%%%%
\\ 
\midrule
\midrule
%%%%%%%%%%%%%%%%%%%%%%%%%%%%%%%%%%%%%%%%%%%%%%%%%%%%%%%%%%%%%%%%%%%%%%%%%%%%%%%%%%%%%%
$1$. & 
   \multicolumn{2}{l}{
		 Starting with $H$, perform Ricci flow iterations 
		 (see Equation~\eqref{eq-ricci} and Section~\ref{sec-flow-surg}) 
		 for a suitable 
	   }
\\
     & 
   \multicolumn{2}{l}{
		 large number 
		 $\eta$ of steps 
		 such that the edge-weights has converged at the end of iterations
     (see
	   }
\\
     & 
   \multicolumn{2}{l}{
		 Equation~\eqref{eq-converg} and Section~\ref{sec-flow-surg}).
	   }
%%%%%%%%%%%%%%%%%%%%%%%%%%%%%%%%%%%%%%%%%%%%%%%%%%%%%%%%%%%%%%%%%%%%%%%%%%%%%%%%%%%%%%
\\
%%%%%%%%%%%%%%%%%%%%%%%%%%%%%%%%%%%%%%%%%%%%%%%%%%%%%%%%%%%%%%%%%%%%%%%%%%%%%%%%%%%%%%
$1.1$. & 
   \multicolumn{2}{l}{
	   \begin{tabular}{l}
		 \hspace*{0.3in}
		 $\bullet$
		 Perform topological surgery with ``surgery amount'' $\delta\%$ after every $\tau$ iterations
		 \\
		 \hspace*{0.4in}
		 (see Section~\ref{sec-flow-surg}). 
	   \end{tabular}
	   }
%%%%%%%%%%%%%%%%%%%%%%%%%%%%%%%%%%%%%%%%%%%%%%%%%%%%%%%%%%%%%%%%%%%%%%%%%%%%%%%%%%%%%%
\\
%%%%%%%%%%%%%%%%%%%%%%%%%%%%%%%%%%%%%%%%%%%%%%%%%%%%%%%%%%%%%%%%%%%%%%%%%%%%%%%%%%%%%%
$1.2$. & 
   \multicolumn{2}{l}{
	   \begin{tabular}{l}
		 \hspace*{0.3in}
		 $\bullet$
		 Perform edge-weight normalization before the start of the next iteration
		 \\
		 \hspace*{0.4in}
		 (see Section~\ref{sec-flow-surg}). 
	   \end{tabular}
	   }
%%%%%%%%%%%%%%%%%%%%%%%%%%%%%%%%%%%%%%%%%%%%%%%%%%%%%%%%%%%%%%%%%%%%%%%%%%%%%%%%%%%%%%
\\
%%%%%%%%%%%%%%%%%%%%%%%%%%%%%%%%%%%%%%%%%%%%%%%%%%%%%%%%%%%%%%%%%%%%%%%%%%%%%%%%%%%%%%
$2$. & 
   \multicolumn{2}{l}{
	 If $G$ is an undirected hypergraph (respectively, directed hypergraph) 
	 then output up to $\kappa$ 
	   }
%%%%%%%%%%%%%%%%%%%%%%%%%%%%%%%%%%%%%%%%%%%%%%%%%%%%%%%%%%%%%%%%%%%%%%%%%%%%%%%%%%%%%%
\\
%%%%%%%%%%%%%%%%%%%%%%%%%%%%%%%%%%%%%%%%%%%%%%%%%%%%%%%%%%%%%%%%%%%%%%%%%%%%%%%%%%%%%%
     & 
   \multicolumn{2}{l}{
	 connected 
	 (respectively, weakly connected) components of $G$
	 that best satisfy the quality 
	   }
%%%%%%%%%%%%%%%%%%%%%%%%%%%%%%%%%%%%%%%%%%%%%%%%%%%%%%%%%%%%%%%%%%%%%%%%%%%%%%%%%%%%%%
\\
%%%%%%%%%%%%%%%%%%%%%%%%%%%%%%%%%%%%%%%%%%%%%%%%%%%%%%%%%%%%%%%%%%%%%%%%%%%%%%%%%%%%%%
     & 
   \multicolumn{2}{l}{
	 parameters
	 for the particular application 
	 (see Section~\ref{sec-quality}). 
	   }
%%%%%%%%%%%%%%%%%%%%%%%%%%%%%%%%%%%%%%%%%%%%%%%%%%%%%%%%%%%%%%%%%%%%%%%%%%%%%%%%%%%%%%
\\ 
\bottomrule
\bottomrule
%%%%%%%%%%%%%%%%%%%%%%%%%%%%%%%%%%%%%%%%%%%%%%%%%%%%%%%%%%%%%%%%%%%%%%%%%%%%%%%%%%%%%%
\end{tabular}
%%%%%%%%%%%%%%%%%%%%%%%%%%%%%%%%%%%%%%%%%%%%%%%%%%%%%%%%%%%%%%%%%%%%%%%%%%%%%%%%%%%%%%
\end{table*}
%%%%%%%%%%%%%%%%%%%%%%%%%%%%%%%%%%%%%%%%%%%%%%%%%%%%%%%%%%%%%%%%%%%%%%%%%%%%%%%%%%%%%%

Based on our discussions in Sections~\ref{sec-def-curv}--\ref{sec-quality}, we 
design an algorithm for finding core(s) in a (directed or undirected) hypergraph 
whose high-level overview is presented in Table~\ref{algo-mod}. 
Below we provide brief comments on the adjustable parameters in 
Table~\ref{algo-mod}:
%%%%%%%%%%%%%%%%%%%%%%%%%%%%%%%%%%%%%%%%%%%%%%%%%%%%%%%%%%%%%%%%%%%%%%%%%%%%%%%%%%%%%%%%
\begin{enumerate}[label=$\triangleright$]
%%%%%%%%%%%%%%%%%%%%%%%%%%%%%%%%%%%%%%%%%%%%%%%%%%%%%%%%%%%%%%%%%%%%%%%%%%%%%%%%%%%%%%%%
\item
$\eta$ controls the \emph{number} of iterations.
Selecting larger $\eta$ will make the algorithm slower but is likely to generate 
smaller cores, although smaller cores may \emph{not} necessarily have better values of
other quality parameters. We recommend selecting $\eta$ sufficiently high to ensure that
the diffusion process has actually converged in several successive iterations 
based on the value of 
$\Delta_{\mathrm{AVE}}$ (and, optionally, $\Delta_{\mathrm{STD}}$)
and that the quality parameters are within acceptable bounds.
%%%%%%%%%%%%%%%%%%%%%%%%%%%%%%%%%%%%%%%%%%%%%%%%%%%%%%%%%%%%%%%%%%%%%%%%%%%%%%%%%%%%%%%%
\item
$\kappa$ controls the \emph{number} of cores selected for further analysis. We suggest a \emph{small}
value for this parameter.
%%%%%%%%%%%%%%%%%%%%%%%%%%%%%%%%%%%%%%%%%%%%%%%%%%%%%%%%%%%%%%%%%%%%%%%%%%%%%%%%%%%%%%%%
\item
$\delta$ and $\tau$ control the \emph{frequency} and the \emph{amount} of the topological surgery
operation, respectively.
Selecting larger values of $\delta$ and smaller values of 
$\tau$ may help with faster convergence, but may also 
end up providing smaller cores by subdividing them. 
%%%%%%%%%%%%%%%%%%%%%%%%%%%%%%%%%%%%%%%%%%%%%%%%%%%%%%%%%%%%%%%%%%%%%%%%%%%%%%%%%%%%%%%%
\end{enumerate}
%%%%%%%%%%%%%%%%%%%%%%%%%%%%%%%%%%%%%%%%%%%%%%%%%%%%%%%%%%%%%%%%%%%%%%%%%%%%%%%%%%%%%%%%
%
Readers, especially from the algorithms or the computational complexity community,
may be curious to have an expression in the theoretical worst-case running time 
in the big-O notation
for the above algorithm.
Unfortunately, it does not seem possible to derive such an expression that would provide 
meaningful bound to the reader since it involves too many parameters for the hypergraph. 
We illustrate this point for an undirected hypergraph $H=(V,E)$.
Let 
\textsc{Time}$_{\!\!\text{\tiny\emd}}\!\!(p)$
denote the running time for solving 
the Earth Mover's distance on a hypergraph with $p$ nodes;
note that obviously 
\textsc{Time}$_{\!\!\text{\tiny\emd}}\!\!(p)=\Omega(p)$.
Consider a hyperedge $e\in E$. For every pair of nodes $x,y\in\cA_e$, 
the total time taken to compute 
$\PPL^x(y)$, $\PPR^x(y)$, 
$\PPL^y(x)$ and $\PPR^y(x)$
is 
$
O \left(
\sum_{e':x\in\cA_{e'}} |\cA_{e'}| + 
\sum_{e':y\in\cA_{e'}} |\cA_{e'}|
\right)
$, and the time to compute the corresponding 
$\mbox{\emd}_H(\PPL^p,\PPR^q)$ value is 
\textsc{Time}$_{\!\!\text{\tiny\emd}}\!\!
\left(
\sum_{e':x\in\cA_{e'}} |\cA_{e'}| + 
\sum_{e':y\in\cA_{e'}} |\cA_{e'}|
\right)$.
Summing over all pairs of nodes in $e$ and then summing over all hyperedges gives 
us the following time bound for one iteration of Ricci flow:
%%%%%%%%%%%%%%%%%%%%%%%%%%%%%%%%%%%%%%%%%%%%%%%%%%%%%%%%%%%%%%%%%%%%%%%%%%%%%%%%%%%%%%%%
\begin{gather*}
O
\left(
\!\!
\!\!
\text{\textsc{Time}}_{\!\!\!\!\text{\tiny\emd}}
\!\!
\!
\left(
\sum_{e\in E}
\sum_{x,y\in \cA_e}
\left(
\sum_{e':x\in\cA_{e'}} 
\!\!
|\cA_{e'}| + 
\!\!
\sum_{e':y\in\cA_{e'}} 
\!\!
|\cA_{e'}|
\right)
\right)
\right)
\end{gather*}
%%%%%%%%%%%%%%%%%%%%%%%%%%%%%%%%%%%%%%%%%%%%%%%%%%%%%%%%%%%%%%%%%%%%%%%%%%%%%%%%%%%%%%%%
The above bound depends in a non-trivial manner on the frequencies of nodes and 
the lengths of hyperedges, and further simplification is not possible without 
making additional assumptions.
For the very special case when every node occurs in \emph{exactly} $f$ hyperedges 
and every hyperedge has \emph{exactly} the same length $\ell$,
letting $n$ denote the number of nodes in the hypergraph
the above running time can be simplified to
$
O
\left(
\!\!
\text{\textsc{Time}}_{\!\!\!\!\text{\tiny\emd}}
\!\!
\left(
\frac{n\,f}{\ell}
\times 
\ell^2 
\times 
f
\right)
\right)
=
O
\left(
\!\!
%%\text{\textsc{Lnprg}}
\text{\textsc{Time}}_{\!\!\!\!\text{\tiny\emd}}
\!\!
\left(
n\,f^2\,\ell
\right)
\right)
$.
A somewhat more complicated expression for the running time for 
directed hypergraphs can also be calculated in a similar manner.

%%%%%%%%%%%%%%%%%%%%%%%%%%%%%%%%%%%%%%%%%%%%%%%%%%%%%%%%%%%%%%%%%%%%%%%%%%%%%%%%%%%%%%
\paragraph{Implementation and source codes}
%%%%%%%%%%%%%%%%%%%%%%%%%%%%%%%%%%%%%%%%%%%%%%%%%%%%%%%%%%%%%%%%%%%%%%%%%%%%%%%%%%%%%%

We implemented our algorithm in python. 
The linear program for calculation of \emd\ was solved using the python library of the Gurobi Optimizer
whose $\LP$ solver is known for its superior performance, often solving optimization models faster than 
other $\LP$ solvers in the industry (we used the free academic license to use the optimizer for our work).
The source codes for our implementation 
are freely available via GitHub at the link \url{https://github.com/iamprith/Ricci-Flow-on-Hypergraphs}.

%%%%%%%%%%%%%%%%%%%%%%%%%%%%%%%%%%%%%%%%%%%%%%%%%%%%%%%%%%%%%%%%%%%%%%%%%%%%%%%%%%%%%%
\subsection{\label{hyper-construct}Hypergraph Construction}
%%%%%%%%%%%%%%%%%%%%%%%%%%%%%%%%%%%%%%%%%%%%%%%%%%%%%%%%%%%%%%%%%%%%%%%%%%%%%%%%%%%%%%

%%%%%%%%%%%%%%%%%%%%%%%%%%%%%%%%%%%%%%%%%%%%%%%%%%%%%%%%%%%%%%%%%%%%%%%%%%%%%%%%%%%%%%
\subsubsection{\label{sec-data-metabolic}Directed Hypergraphs}
%%%%%%%%%%%%%%%%%%%%%%%%%%%%%%%%%%%%%%%%%%%%%%%%%%%%%%%%%%%%%%%%%%%%%%%%%%%%%%%%%%%%%%

%%%%%%%%%%%%%%%%%%%%%%%%%%%%%%%%%%%%%%%%%%%%%%%%%%%%%%%%%%%%%%%%%%%%%%%%%%%%%%%%%%%%%%
\begin{table*}
%%%%%%%%%%%%%%%%%%%%%%%%%%%%%%%%%%%%%%%%%%%%%%%%%%%%%%%%%%%%%%%%%%%%%%%%%%%%%%%%%%%%%%
\caption{\label{tab:hypergraph_data}Some first-order statistics of constructed 
directed hypergraphs from metabolic systems in~\cite{king2016bigg}.} 
%%%%%%%%%%%%%%%%%%%%%%%%%%%%%%%%%%%%%%%%%%%%%%%%%%%%%%%%%%%%%%%%%%%%%%%%%%%%%%%%%%%%%%
\begin{tabular}{l @{\hskip 0.3in} c @{\hskip 0.3in} c @{\hskip 0.3in} 
             c c c  c c c }\toprule
%%%%%%%%%%%%%%%%%%%%%%%%%%%%%%%%%%%%%%%%%%%%%%%%%%%%%%%%%%%%%%%%%%%%%%%%%%%%%%%%%%%%%%
\multicolumn{1}{c}{} 
      & \multicolumn{8}{c}{Constructed directed hypergraph $(V,E,w)$}
\\ \cmidrule{2-9}
\multicolumn{1}{c}{} 
      &      &       &
      \multicolumn{3}{c|}{in-degree}
			&
      \multicolumn{3}{c}{out-degree}
\\ 
\multicolumn{1}{c}{Name of metabolic system} 
      & $|V|$ & $|E|$ &
			average & max & \multicolumn{1}{c|}{min} & 
			average & max & min 
%%%%%%%%%%%%%%%%%%%%%%%%%%%%%%%%%%%%%%%%%%%%%%%%%%%%%%%%%%%%%%%%%%%%%%%%%%%%%%%%%%%%%%
\\ \midrule
%%%%%%%%%%%%%%%%%%%%%%%%%%%%%%%%%%%%%%%%%%%%%%%%%%%%%%%%%%%%%%%%%%%%%%%%%%%%%%%%%%%%%%
Escherichia Coli & $762$ & $1335$ 
     & $3.63$ & $380$ & $0$ & $3.56$ & $231$ & $0$ \\ 
\midrule
Homo Sapiens & $343$ & $672$ 
     & $3.58$ & $165$ & $0$ & $3.27$ & $114$ & $1$ \\ 
\midrule
Helicobacter Pylori & $486$ & $740$ 
     & $3.20$ & $191$ & $0$ & $3.10$ & $129$ & $0$ \\ 
\midrule
Methanosarcina berkeri str. Fusaro & $629$ & $905$ 
     & $3.22$ & $255$ & $0$ & $3.09$ & $167$ & $0$ \\ 
\midrule
Mycobacterium Tuberculosis & $826$ & $1297$ 
     & $3.68$ & $364$ & $0$ & $3.49$ & $253$ & $0$ \\ 
\midrule
Synechococcus elongatus & $769$ & $849$ 
     & $3.09$ & $244$ & $0$ & $2.97$ & $193$ & $0$ \\ 
\midrule
Synechocystis & $796$ & $863$
     & $3.00$ & $277$ & $0$ & $3.26$ & $222$ & $0$ 
\\ \bottomrule
%%%%%%%%%%%%%%%%%%%%%%%%%%%%%%%%%%%%%%%%%%%%%%%%%%%%%%%%%%%%%%%%%%%%%%%%%%%%%%%%%%%%%%
\end{tabular}
%%%%%%%%%%%%%%%%%%%%%%%%%%%%%%%%%%%%%%%%%%%%%%%%%%%%%%%%%%%%%%%%%%%%%%%%%%%%%%%%%%%%%%
\end{table*}
%%%%%%%%%%%%%%%%%%%%%%%%%%%%%%%%%%%%%%%%%%%%%%%%%%%%%%%%%%%%%%%%%%%%%%%%%%%%%%%%%%%%%%

%%%%%%%%%%%%%%%%%%%%%%%%%%%%%%%%%%%%%%%%%%%%%%%%%%%%%%%%%%%%%%%%%%%%%%%%%%%%%%%%%%%%%%
\begin{figure}
\caption{\label{fig222}A visual illustration of the hypergraph-theoretic representation 
(Section~\ref{sec-data-metabolic})
vs.\ two common graph-theoretic representations of biochemical reactions.
If the reaction times are known accurately then they can be used as the weights of the 
corresponding hyperedges.}
\includegraphics[width=0.45\textwidth]{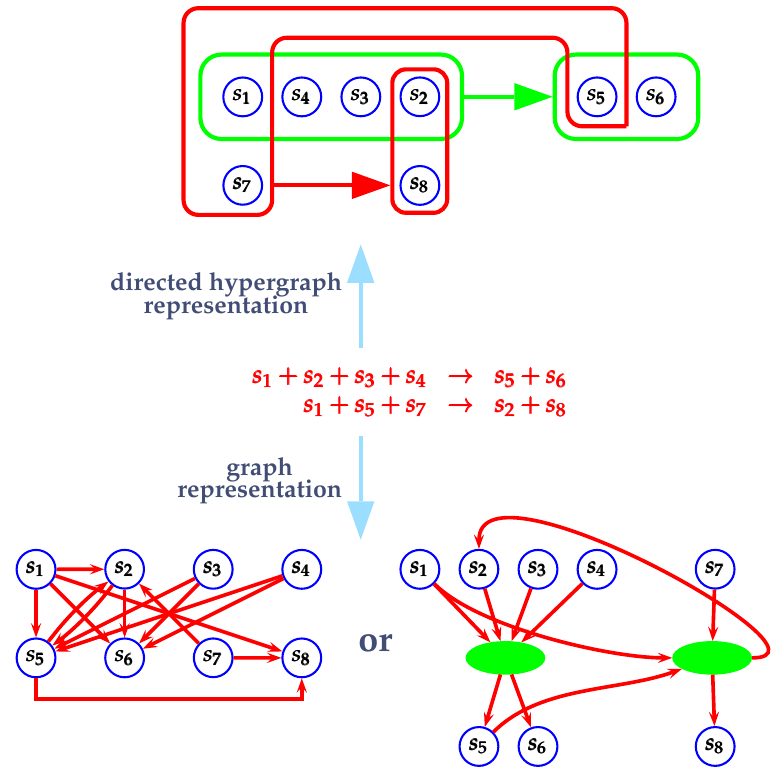}
\end{figure}  
%%%%%%%%%%%%%%%%%%%%%%%%%%%%%%%%%%%%%%%%%%%%%%%%%%%%%%%%%%%%%%%%%%%%%%%%%%%%%%%%%%%%%%

We modeled various metabolic and biochemical reactions as directed hypergraphs. 
Each reaction in the data 
had always at least one reactant but sometimes did not have a product.
We represent a reaction of the form 
``$R_1+\cdots+R_k\to P_1+\cdots+P_\ell$'',
where $R_i$'s are the reactants, $P_j$'s are the products,
as a directed hyperedge $e$ (see \FI{fig222})
with 
$\cTa_e=\{R_1,\dots,R_k\}$ 
and
$\cHe_e=\{P_1,\dots,P_\ell\}$\footnote{{\emph{If the reaction times are known accurately then they can be used as the 
weights of the corresponding hyperedges}}.}.
In other words, each hyperedge points from the set of reactants to the set of 
products of the corresponding reaction.
For reactions of the form 
``$R_1+\cdots+R_k\to$'' without a product
we use 
$\cHe_e=\{\mathsf{sink}\}$
for a unique node named 
``$\mathsf{sink}$'' 
following the \emph{same} conventions used in the network science literature 
(note that there is exactly \emph{one} $\mathsf{sink}$ node in the entire hypergraph 
and it does \emph{not} appear in $\cTa_f$ for any hyperedge $f$ in the hypergraph). 
\emph{All} the directed hypergraphs that we construct are weakly connected.
We computed 
some first-order statistics of the constructed directed hypergraphs
as shown in Table~\ref{tab:hypergraph_data}.
We show in \FI{fig222}
a visual comparison of our hypergraph-theoretic representation 
%%(Section~\ref{sec-data-metabolic})
with two common graph-theoretic representations of biochemical reactions.

%%%%%%%%%%%%%%%%%%%%%%%%%%%%%%%%%%%%%%%%%%%%%%%%%%%%%%%%%%%%%%%%%%%%%%%%%%%%%%%%%%%%%%
\subsubsection{\label{sec-coauth-data}Undirected Hypergraphs}
%%%%%%%%%%%%%%%%%%%%%%%%%%%%%%%%%%%%%%%%%%%%%%%%%%%%%%%%%%%%%%%%%%%%%%%%%%%%%%%%%%%%%%

%%%%%%%%%%%%%%%%%%%%%%%%%%%%%%%%%%%%%%%%%%%%%%%%%%%%%%%%%%%%%%%%%%%%%%%%%%%%%%%%%%%%%%
\begin{figure}
\caption{\label{fig333}A visual illustration of the hypergraph-theoretic representation 
(Section~\ref{sec-coauth-data})
vs. graph-theoretic representation of co-authorships.}
\includegraphics[width=0.4\textwidth]{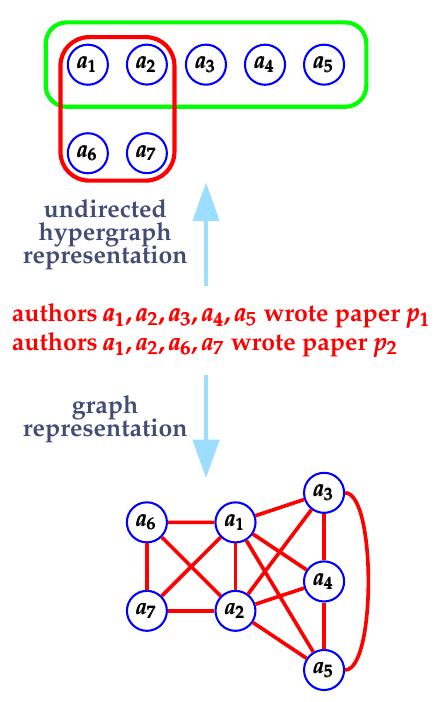}
\end{figure}  
%%%%%%%%%%%%%%%%%%%%%%%%%%%%%%%%%%%%%%%%%%%%%%%%%%%%%%%%%%%%%%%%%%%%%%%%%%%%%%%%%%%%%%

The hypergraph has a node corresponding to every authors and 
an undirected hyperedge $e$ with $\cA_e=\{a_1,\dots,a_k\}$
corresponding to 
each paper co-authored by authors $a_1,\dots,a_k$. 
For the \csp\ dataset,
we build an undirected hypergraph out of the $900$ papers and take the largest connected component
as our input undirected hypergraph.
For the \nsp\ dataset,
we found the largest connected component in the 
resulting undirected hypergraph of $582$ papers, and took that connected component as our input undirected hypergraph.
%
%%Both of our undirected hypergraphs are connected.
We computed 
some first-order statistics of the constructed undirected hypergraphs
as shown in Table~\ref{tab:undir-hypergraph_data}.
We show in \FI{fig333}
a visual comparison of our hypergraph-theoretic representation 
%%(Section~\ref{sec-data-metabolic})
with a common graph-theoretic representation of co-authorship relationships.

%%%%%%%%%%%%%%%%%%%%%%%%%%%%%%%%%%%%%%%%%%%%%%%%%%%%%%%%%%%%%%%%%%%%%%%%%%%%%%%%%%%%%%
\begin{table}
%%%%%%%%%%%%%%%%%%%%%%%%%%%%%%%%%%%%%%%%%%%%%%%%%%%%%%%%%%%%%%%%%%%%%%%%%%%%%%%%%%%%%%
\caption{\label{tab:undir-hypergraph_data}Some first-order statistics of constructed 
undirected hypergraphs in Section~\ref{sec-coauth-data}.
} 
%%%%%%%%%%%%%%%%%%%%%%%%%%%%%%%%%%%%%%%%%%%%%%%%%%%%%%%%%%%%%%%%%%%%%%%%%%%%%%%%%%%%%%
\begin{tabular}{l @{\hskip 0.3in} c @{\hskip 0.3in} c @{\hskip 0.3in} c c c }\toprule
%%%%%%%%%%%%%%%%%%%%%%%%%%%%%%%%%%%%%%%%%%%%%%%%%%%%%%%%%%%%%%%%%%%%%%%%%%%%%%%%%%%%%%
\multicolumn{1}{c}{} 
      & \multicolumn{5}{c}{
			    \begin{tabular}{c}
			    Constructed undirected 
					\\
					hypergraph $(V,E,w)$
			    \end{tabular}
					}
\\ \cmidrule{2-6}
      &      &       &
      \multicolumn{3}{c}{degree}
\\ 
      & $|V|$ & $|E|$ &
			average & max & min 
%%%%%%%%%%%%%%%%%%%%%%%%%%%%%%%%%%%%%%%%%%%%%%%%%%%%%%%%%%%%%%%%%%%%%%%%%%%%%%%%%%%%%%
\\ \midrule
%%%%%%%%%%%%%%%%%%%%%%%%%%%%%%%%%%%%%%%%%%%%%%%%%%%%%%%%%%%%%%%%%%%%%%%%%%%%%%%%%%%%%%
\csp\ dataset &
    $496$ & $609$ & $2.97$  & $39$ & $1$ 
\\ \midrule
%%%%%%%%%%%%%%%%%%%%%%%%%%%%%%%%%%%%%%%%%%%%%%%%%%%%%%%%%%%%%%%%%%%%%%%%%%%%%%%%%%%%%%
\nsp\ dataset &
    $518$ & $213$ &   $1.61$ & $20$  &  $1$
%%%%%%%%%%%%%%%%%%%%%%%%%%%%%%%%%%%%%%%%%%%%%%%%%%%%%%%%%%%%%%%%%%%%%%%%%%%%%%%%%%%%%%
\\ \bottomrule
%%%%%%%%%%%%%%%%%%%%%%%%%%%%%%%%%%%%%%%%%%%%%%%%%%%%%%%%%%%%%%%%%%%%%%%%%%%%%%%%%%%%%%
\end{tabular}
%%%%%%%%%%%%%%%%%%%%%%%%%%%%%%%%%%%%%%%%%%%%%%%%%%%%%%%%%%%%%%%%%%%%%%%%%%%%%%%%%%%%%%
\end{table}
%%%%%%%%%%%%%%%%%%%%%%%%%%%%%%%%%%%%%%%%%%%%%%%%%%%%%%%%%%%%%%%%%%%%%%%%%%%%%%%%%%%%%%

%%%%%%%%%%%%%%%%%%%%%%%%%%%%%%%%%%%%%%%%%%%%%%%%%%%%%%%%%%%%%%%%%%%%%%%%%%%%%%%%%%%%%%
\subsection{\label{sec-res-normprob}Proof of Inapplicability of Normalized Ricci 
Flow equation~\eqref{eq-ricci-norm} to graphs}
%%%%%%%%%%%%%%%%%%%%%%%%%%%%%%%%%%%%%%%%%%%%%%%%%%%%%%%%%%%%%%%%%%%%%%%%%%%%%%%%%%%%%%

%%%%%%%%%%%%%%%%%%%%%%%%%%%%%%%%%%%%%%%%%%%%%%%%%%%%%%%%%%%%%%%%%%%%%%%%%%%%%%%%%%%%%%
\begin{figure*}
\caption{\label{fig111}The graph $G_n$ used in the proof of Theorem~\ref{thm-ex}.}
\includegraphics[width=\textwidth]{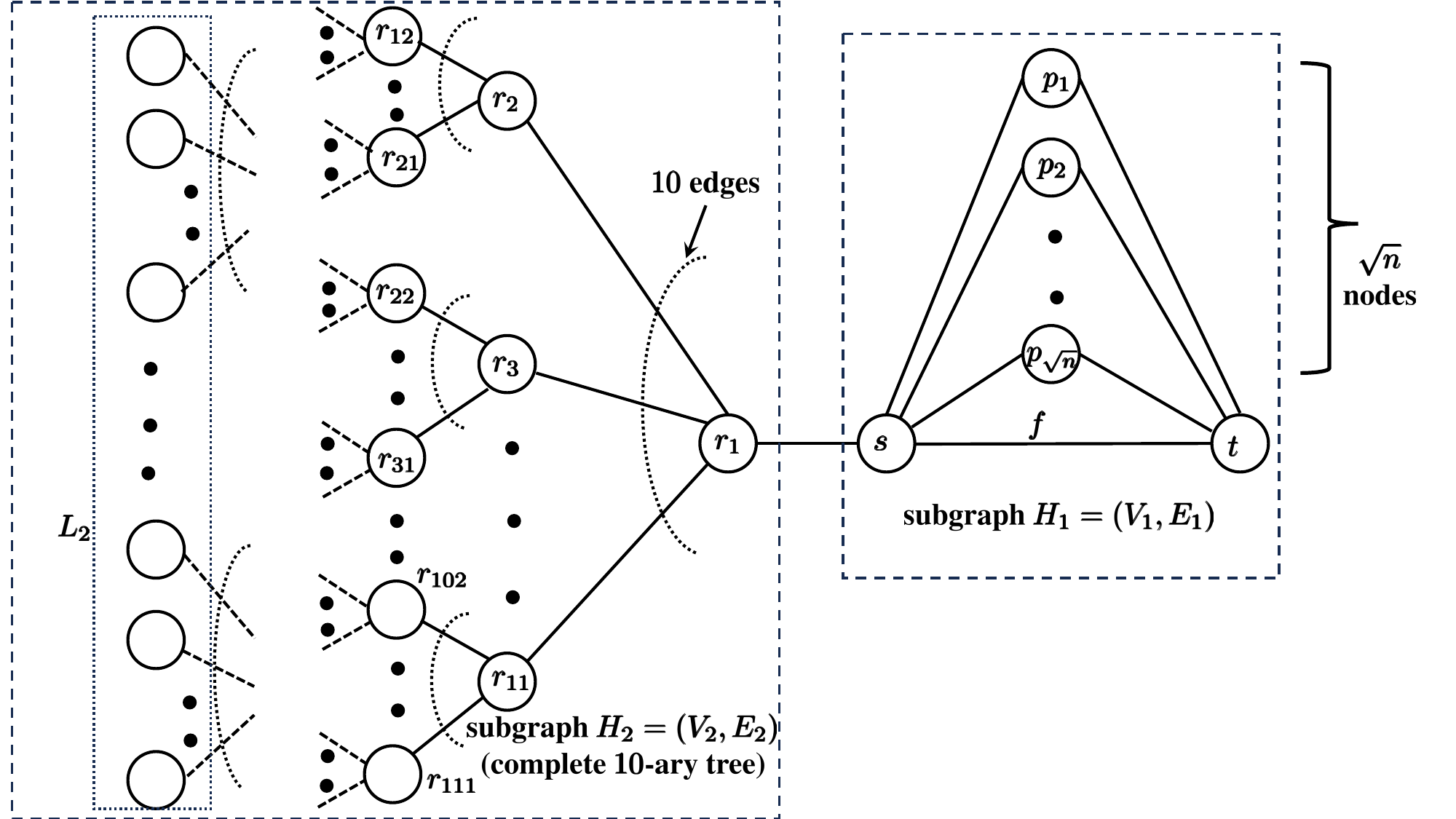}
\end{figure*}  
%%%%%%%%%%%%%%%%%%%%%%%%%%%%%%%%%%%%%%%%%%%%%%%%%%%%%%%%%%%%%%%%%%%%%%%%%%%%%%%%%%%%%%

The following theorem shows that there are infinitely many graphs for which the normalized Ricci 
flow equation~\eqref{eq-ricci-norm} will make $w^{(1)}(f)$ negative for some edge $f$ 
thus rendering the iterative process in~\eqref{eq-ricci-norm}
\emph{impossible} to execute beyond the first step.

%%%%%%%%%%%%%%%%%%%%%%%%%%%%%%%%%%%%%%%%%%%%%%%%%%%%%%%%%%%%%%%%%%%%%%%%%%%%%%%%
\begin{theorem}\label{thm-ex}
For all sufficiently large $n$,
there exists an undirected graph $G_n$ on $n$ nodes for which 
$w^{(1)}(f)<0$ for some edge $f$ of $G_n$.
\end{theorem}
%%%%%%%%%%%%%%%%%%%%%%%%%%%%%%%%%%%%%%%%%%%%%%%%%%%%%%%%%%%%%%%%%%%%%%%%%%%%%%%%

%%%%%%%%%%%%%%%%%%%%%%%%%%%%%%%%%%%%%%%%%%%%%%%%%%%%%%%%%%%%%%%%%%%%%%%%%%%%%%%%
\begin{proof}
For convenience and ease of proof, we first explicitly state the standard definition of 
the Ricci curvature for an undirected graph $G=(V,E,w)$ where the weight $w(e)$ is $1$
for every edge $e\in E$~\cite{LBL22,Oll11,Oll09,Oll10,Oll07,ASD20,DGM23}. 
Consider an edge $e=\{u,v\}\in E$ of our input (undirected unweighted) graph $G=(V,E)$. 
For a node $u$ of $G$, 
let
$\nbr_G(u)=\{u\} \bigcup \{ v \,|\, \{u,v\}\in E\}$
and 
$\deg_G(u)= |\, \nbr_H(u) \setminus \{u\} \,|$ 
denote 
the \emph{closed} neighborhood
and the \emph{degree} of $u$
in $G$, respectively.
Let 
$\PP_{\nbr_G(u)}$ and $\PP_{\nbr_G(v)}$
denote the two uniform distributions over the nodes in 
$\nbr_G(u)$ and $\nbr_G(v)$, respectively.
Extend the distributions 
$\PP_{\nbr_G(u)}$ and $\PP_{\nbr_G(v)}$
to all nodes in $G$ by assigning zero probabilities to nodes in 
$V\setminus\nbr_G(u)$ and 
$V\setminus\nbr_G(v)$, respectively.
The Ollivier-Ricci curvature 
$\mathfrak{C}_{G}(e)$ 
of the edge $e=\{u,v\}$ is then defined as
%%%%%%%%%%%%%%%%%%%%%%%%%%%%%%%%%%%%%%%%%%%%%%%%%%%%%%%%%%%%%%%%%%%%%%%%%%%%%%%%%%%%%%%%
\begin{gather}
\mathfrak{C}_{G}(e)= 1 - \text{\emd}_{H}(\PP_{\nbr_G(u)},\PP_{\nbr_G(v)})
\label{eq10}
\end{gather}
%%%%%%%%%%%%%%%%%%%%%%%%%%%%%%%%%%%%%%%%%%%%%%%%%%%%%%%%%%%%%%%%%%%%%%%%%%%%%%%%%%%%%%%%
where we use the same calculations of \emd\ as in Section~\ref{sec-def-curv}.

We will show the graph $G_n=(V,E,w)$ of $m$ edges where $w(e)=1$ 
for every edge $e$
(and thus the sum of all edge weights is $m$).
For this case, for $t=0$ and an edge $e$ of $G_n=G_n^{(0)}=(V^{(0)},E^{(0)},w^{(0)})$ equation~\eqref{eq-ricci-norm} simplifies to 
%%%%%%%%%%%%%%%%%%%%%%%%%%%%%%%%%%%%%%%%%%%%%%%%%%%%%%%%%%%%%%%%%%%%%%%%%%%%%%%%
\begin{gather*}
w^{(1)}(e)
=
1 - 
\mathfrak{C}_{G_n}(e)
+
\frac{s}{m}
\sum_{h\in E} 
\mathfrak{C}_{G_n}(h)
\end{gather*}
%%%%%%%%%%%%%%%%%%%%%%%%%%%%%%%%%%%%%%%%%%%%%%%%%%%%%%%%%%%%%%%%%%%%%%%%%%%%%%%%
It thus suffices to show the graph $G_n$ with an edge $e$ satisfying 
$
\mathfrak{C}_{G_n}(e)
-
\frac{s}{m}
\sum_{h\in E} 
\mathfrak{C}_{G_n}(h)
>1
$.
We show this by showing a graph $G_n$ of $n$ nodes in which 
$
\mathfrak{C}_{G_n}(e) 
\geq 1 - o(1)$ 
and 
$
\frac{s}{m}
\sum_{h\in E} 
\mathfrak{C}_{G_n}(h)
\leq 
-\eps
$
for some positive \emph{constant} $\eps>0$.
Our graph $G_n$, as shown in \FI{fig111}, 
consists of two subgraphs $H_1=(V_1,E_1)$ and $H_2=(V_2,E_2)$ connected by an edge, where 
$H_1$ has $n_1=\sqrt{n}+2$ nodes and $m_1=2\sqrt{n}+1$ edges 
$H_2$ is a complete $10$-ary tree
having 
$n_2=n-n_1$ nodes and $m_2=n_2-1=n-\sqrt{n}-3$ edges.
(since the graph is unweighted, we omit mentioning the weights of the edges).
Thus, the total number of edges of $G_n$ is $m=m_1+m_2+1=n+\sqrt{n}+2$.
Let $L_2$ be the set of leaf nodes of $H_2$.
We calculate the number of leaves $\ell_2 = |L_2|$ of $H_2$ in the following simple manner. 
Letting $k$ be the depth of $H_2$, we have 
$\sum_{j=0}^k 10^j =n_2$.
This gives the number of leaves $\ell_2$ of $H_2$ as 
$\ell_2=10^k=\frac{9n_2+1}{10}=\frac{9n-9\sqrt{n}-17}{10}$.
The edge $f=\{s,t\}$ shown in \FI{fig111} is the edge for which we will show that 
$w^{(1)}(f)<0$.

Consider any edge 
$e=\{u,v\}\in E$ of $G_n$ 
and the associated distributions 
$\PP_{\nbr_{G_n}(u)}$ and $\PP_{\nbr_{G_n}(v)}$
as defined before.
The (standard) \emph{total variation distance} 
$|| \PP_{\nbr_{G_n}(u)} - \PP_{\nbr_{G_n}(v)} ||_{\mathrm{TVD}}$
between the two distributions 
$\PP_{\nbr_{G_n}(u)}$ and $\PP_{\nbr_{G_n}(v)}$
is defined as
%%%%%%%%%%%%%%%%%%%%%%%%%%%%%%%%%%%%%%%%%%%%%%%%%%%%%%%%%%%%%%%%%%%%%%%%%%%%%%%%%%%%%%%%
\begin{multline*}
|| e ||_{\mathrm{TVD}}
\eqdef
|| \PP_{\nbr_{G_n}(u)} - \PP_{\nbr_{G_n}(v)} ||_{\mathrm{TVD}}
=
\\
\frac{1}{2} \times 
\Bigg(
\hspace*{-0.5in}
\sum_{\hspace*{0.5in} \alpha \in \nbr_{G_n}(u) \cap \nbr_{G_n}(v)} 
    \hspace*{-0.5in}
      \left|\, 
         \PP_{\nbr_{G_n}(u)}(\alpha) - \PP_{\nbr_{G_n}(v)} (\alpha)
			\,\right|
\\
+
\hspace*{-0.5in}
\sum_{\hspace*{0.5in} \beta \in \nbr_{G_n}(u) \setminus \nbr_{G_n}(v)} 
    \hspace*{-0.5in}
         \PP_{\nbr_{G_n}(u)}(\beta)
+
\hspace*{-0.5in}
\sum_{\hspace*{0.5in} \gamma \in \nbr_{G_n}(v) \setminus \nbr_{G_n}(u)} 
    \hspace*{-0.5in}
         \PP_{\nbr_{G_n}(v)}(\gamma)
\Bigg)
\end{multline*}
%%%%%%%%%%%%%%%%%%%%%%%%%%%%%%%%%%%%%%%%%%%%%%%%%%%%%%%%%%%%%%%%%%%%%%%%%%%%%%%%%%%%%%%%
By Proposition~1 of~\cite{ASD20}
we have 
%%%%%%%%%%%%%%%%%%%%%%%%%%%%%%%%%%%%%%%%%%%%%%%%%%%%%%%%%%%%%%%%%%%%%%%%%%%%%%%%%%%%%%%%
\begin{gather}
1- 3 \times || e ||_{\mathrm{TVD}}
\leq
\mathfrak{C}_{{G_n}}(e)
\leq
1- || e ||_{\mathrm{TVD}}
\label{eqb1}
\end{gather}
%%%%%%%%%%%%%%%%%%%%%%%%%%%%%%%%%%%%%%%%%%%%%%%%%%%%%%%%%%%%%%%%%%%%%%%%%%%%%%%%%%%%%%%%
Now suppose that 
the following condition holds for the edge $e$:
%%%%%%%%%%%%%%%%%%%%%%%%%%%%%%%%%%%%%%%%%%%%%%%%%%%%%%%%%%%%%%%%%%%%%%%%%%%%%%%%%%%%%%%%
\begin{multline}
\forall \, \alpha\in \nbr_{G_n}(u)\setminus\{u,v\} \,\,\, \forall \, \beta\in\nbr_{G_n}(v)\setminus\{u,v\}
:
\\
\dist_{G_n}(\alpha,\beta)=3
\label{eqc1}
\tag{C1}
\end{multline}
%%%%%%%%%%%%%%%%%%%%%%%%%%%%%%%%%%%%%%%%%%%%%%%%%%%%%%%%%%%%%%%%%%%%%%%%%%%%%%%%%%%%%%%%
Using the discussions surrounding Proposition~1 and 
Proposition~2 of~\cite{ASD20}, 
we get the following bound for this case: 
%%%%%%%%%%%%%%%%%%%%%%%%%%%%%%%%%%%%%%%%%%%%%%%%%%%%%%%%%%%%%%%%%%%%%%%%%%%%%%%%%%%%%%%%
\begin{multline}
\text{\emd}_{e}(\PP_{\nbr_{G_n}(u)},\PP_{\nbr_{G_n}(v)})
\geq
\\
3 \times || e ||_{\mathrm{TVD}}
-
         \left| \, \PP_{\nbr_{G_n}(u)}(u) - \PP_{\nbr_{G_n}(v)} (u) \, \right|
\\
-
         \left| \, \PP_{\nbr_G(u)}(v) - \PP_{\nbr_G(v)} (v) \, \right|
\label{eqb2}
\end{multline}
%%%%%%%%%%%%%%%%%%%%%%%%%%%%%%%%%%%%%%%%%%%%%%%%%%%%%%%%%%%%%%%%%%%%%%%%%%%%%%%%%%%%%%%%
Finally, suppose that edge $e$ satisfies 
$\nbr_{G_n}(u)=\{u,v\}$. 
Since 
$
\dist_{G_n}(u,\alpha)=2
$
for all $\alpha\in 
\nbr_{G_n}(v)\setminus\{u,v\}
$,
in this case we get the following bound:
%%%%%%%%%%%%%%%%%%%%%%%%%%%%%%%%%%%%%%%%%%%%%%%%%%%%%%%%%%%%%%%%%%%%%%%%%%%%%%%%%%%%%%%%
\begin{multline}
\text{\emd}_{e}(\PP_{\nbr_{G_n}(u)},\PP_{\nbr_{G_n}(v)})
\geq
\\
\frac{1}{2} \times
\left [
\frac{
  2 (\deg_{G_n}(v)-1) 
}{
  \deg_{G_n}(v)+1
}
+
2 
\left(
\frac{1}{2} - \frac{1}{ 
  \deg_{G_n}(v)+1
   }
\right)
\right ]
\\
=
\frac{3}{2} - \frac{1}{
  2\deg_{G_n}(v)+2
}
\label{eqb3}
\end{multline}
%%%%%%%%%%%%%%%%%%%%%%%%%%%%%%%%%%%%%%%%%%%%%%%%%%%%%%%%%%%%%%%%%%%%%%%%%%%%%%%%%%%%%%%%
We now use relatively straightforward calculations to calculate
the Ricci curvature values of various edges of $G_n$:
%%%%%%%%%%%%%%%%%%%%%%%%%%%%%%%%%%%%%%%%%%%%%%%%%%%%%%%%%%%%%%%%%%%%%%%%%%%%%%%%%%%%%%%%
\begin{description}
\item[(\emph{i})]
For the edge 
$f=\{s,t\}$, 
we get  
%%%%%%%%%%%%%%%%%%%%%%%%%%%%%%%%%%%%%%%%%%%%%%%%%%%%%%%%%%%%%%%%%%%%%%%%%%%%%%%%%%%%%%%%
\begin{multline*}
\textstyle
|| f ||_{\mathrm{TVD}}
=
\frac{\sqrt{n}+2}{2}\times
\left(
  \frac{1}{\sqrt{n}+2} - \frac{1}{\sqrt{n}+3}  
\right)
+
  \frac{1}{2(\sqrt{n}+3)}
\\
\textstyle
=
\frac{1}{2} - 
  \frac{ \sqrt{n}+1 }{2\sqrt{n}+6}
=
  \frac{ 1 }{\sqrt{n}+3}
\end{multline*}
%%%%%%%%%%%%%%%%%%%%%%%%%%%%%%%%%%%%%%%%%%%%%%%%%%%%%%%%%%%%%%%%%%%%%%%%%%%%%%%%%%%%%%%%
and thus using~\eqref{eqb1} we get
$\mathfrak{C}_{G_n}(f)
\geq 
1 - 
  \frac{ 3 }{\sqrt{n}+3}
=1 - o(1)
$, as required.
%%%%%%%%%%%%%%%%%%%%%%%%%%%%%%%%%%%%%%%%%%%%%%%%%%%%%%%%%%%%%%%%%%%%%%%%%%%%%%%%%%%%%%%%
\item[(\emph{ii})]
For any edge $e=\{p_i,s\}\in E_1$ for $i\in\{1,\dots,\sqrt{n}\,\}$, we have
%%%%%%%%%%%%%%%%%%%%%%%%%%%%%%%%%%%%%%%%%%%%%%%%%%%%%%%%%%%%%%%%%%%%%%%%%%%%%%%%%%%%%%%%
\begin{gather*}
|| e ||_{\mathrm{TVD}}
=
\frac{1}{2} \times 
\frac{1}{4} 
+ 3\times 
\frac{1}{2} \times 
\left( \frac{1}{3}-\frac{1}{4} \right) = \frac{1}{4}
\end{gather*}
%%%%%%%%%%%%%%%%%%%%%%%%%%%%%%%%%%%%%%%%%%%%%%%%%%%%%%%%%%%%%%%%%%%%%%%%%%%%%%%%%%%%%%%%
and thus using~\eqref{eqb1} we get
%%%%%%%%%%%%%%%%%%%%%%%%%%%%%%%%%%%%%%%%%%%%%%%%%%%%%%%%%%%%%%%%%%%%%%%%%%%%%%%%%%%%%%%%
\begin{gather*}
\Lambda_1 = 
\sum_{
{ e=\{p_i,s\}\in E_1,\, i\in\{1,\dots,\sqrt{n}\,\} }
}
\mathfrak{C}_{G_n}(e)
\leq
\frac{3}{4}\sqrt{n}
\end{gather*}
%%%%%%%%%%%%%%%%%%%%%%%%%%%%%%%%%%%%%%%%%%%%%%%%%%%%%%%%%%%%%%%%%%%%%%%%%%%%%%%%%%%%%%%%
\item[(\emph{iii})]
Since trivially 
$\mathfrak{C}_{G_n}(e) \leq 1$ 
for any edge $e=\{p_i,t\}\in E_1$ for $i\in\{1,\dots,\sqrt{n}\,\}$, 
we get
%%%%%%%%%%%%%%%%%%%%%%%%%%%%%%%%%%%%%%%%%%%%%%%%%%%%%%%%%%%%%%%%%%%%%%%%%%%%%%%%%%%%%%%%
\begin{gather*}
\Lambda_2 = 
\sum_{
{ e=\{p_i,t\}\in E_1,\, i\in\{1,\dots,\sqrt{n}\,\} }
}
\mathfrak{C}_{G_n}(e)
\leq
\sqrt{n}
\end{gather*}
%%%%%%%%%%%%%%%%%%%%%%%%%%%%%%%%%%%%%%%%%%%%%%%%%%%%%%%%%%%%%%%%%%%%%%%%%%%%%%%%%%%%%%%%
\item[(\emph{iv})]
For any edge $e=\{\, \{r_i,r_j \} \,|\, r_i\in L_2 \, \}\in E_2$ connecting a leaf node $r_i$ to another non-leaf node $r_j$ in $H_2$, 
since 
$\nbr_{G_n}(r_i)=\{r_i,r_j\}$ and  
$\deg_{G_n}(v)=11$, using~\eqref{eqb3} we get  
$
\text{\emd}_{e}(\PP_{\nbr_{G_n}(r_i)},\PP_{\nbr_{G_n}(r_j)})
\geq
\frac{3}{2} - \frac{1}{24}=
\frac{35}{24}
$,
and thus
%%%%%%%%%%%%%%%%%%%%%%%%%%%%%%%%%%%%%%%%%%%%%%%%%%%%%%%%%%%%%%%%%%%%%%%%%%%%%%%%%%%%%%%%
\begin{multline*}
\Lambda_3 = 
\sum_{ e=\{\, \{r_i,r_j \} \,|\, r_i\in L_2 \, \}\in E_2 }
\mathfrak{C}_{G_n}(e)
\leq
\left( 1 - 
\frac{35}{24}
\right) \times \ell_2
\\
=
\frac{-99n+99\sqrt{n}+187}{240}
\end{multline*}
%%%%%%%%%%%%%%%%%%%%%%%%%%%%%%%%%%%%%%%%%%%%%%%%%%%%%%%%%%%%%%%%%%%%%%%%%%%%%%%%%%%%%%%%
\item[(\emph{v})]
Condition~\eqref{eqc1} applies
to any edge $e=\{\, \{r_i,r_j \} \,|\, r_i,r_j\notin L_2 \}\in E_2$ connecting a pair of non-leaf nodes $r_i,r_j$ in $H_2$, 
or to the edge $\{u,r_1\}$.
Thus, by using~\eqref{eqb2} we get the following bounds: 
%%%%%%%%%%%%%%%%%%%%%%%%%%%%%%%%%%%%%%%%%%%%%%%%%%%%%%%%%%%%%%%%%%%%%%%%%%%%%%%%%%%%%%%%
\begin{itemize}
\item
For an edge $e=\{\, \{r_i,r_j \} \,|\, r_i,r_j\notin L_2 \}\in E_2$, 
$
\text{\emd}_{e}(\PP_{\nbr_{G_n}(r_i)},\PP_{\nbr_{G_n}(r_j)}) = \frac{5}{4}
$
and thus
%%%%%%%%%%%%%%%%%%%%%%%%%%%%%%%%%%%%%%%%%%%%%%%%%%%%%%%%%%%%%%%%%%%%%%%%%%%%%%%%%%%%%%%%
\begin{multline*}
\Lambda_4 = 
\hspace*{-0.5in}
\sum_{ 
\hspace*{0.5in}
e=\{\, \{r_i,r_j \} \,|\, r_i,r_j\notin L_2 \}\in E_2
}
\hspace*{-0.5in}
\mathfrak{C}_{G_n}(e)
\\
\leq
- \frac{1}{4}\times (n_2-\ell_2)
=
\frac{-n+\sqrt{n}+3}{10}
\end{multline*}
%%%%%%%%%%%%%%%%%%%%%%%%%%%%%%%%%%%%%%%%%%%%%%%%%%%%%%%%%%%%%%%%%%%%%%%%%%%%%%%%%%%%%%%%
\item
If $e=\{u,r_1\}$ then
\begin{multline*}
\text{\emd}_{e}(\PP_{\nbr_{G_n}(r_1)},\PP_{\nbr_{G_n}(u)}) = 
\\
3\times \frac{1}{2} \times 
\left(
\frac{10}{12} + \frac{\sqrt{n}+1}{\sqrt{n}+3}
\right)
=\frac{5}{4} + \frac{3}{2} \times \frac{\sqrt{n}+1}{\sqrt{n}+3}
\end{multline*}
and thus
%%%%%%%%%%%%%%%%%%%%%%%%%%%%%%%%%%%%%%%%%%%%%%%%%%%%%%%%%%%%%%%%%%%%%%%%%%%%%%%%%%%%%%%%
\begin{gather*}
\Lambda_5 = 
\mathfrak{C}_{G_n}(e) \leq
-\frac{1}{4} - \frac{3}{2} \times \frac{\sqrt{n}+1}{\sqrt{n}+3}
\end{gather*}
%%%%%%%%%%%%%%%%%%%%%%%%%%%%%%%%%%%%%%%%%%%%%%%%%%%%%%%%%%%%%%%%%%%%%%%%%%%%%%%%%%%%%%%%
\end{itemize}
%%%%%%%%%%%%%%%%%%%%%%%%%%%%%%%%%%%%%%%%%%%%%%%%%%%%%%%%%%%%%%%%%%%%%%%%%%%%%%%%%%%%%%%%
\end{description}
%%%%%%%%%%%%%%%%%%%%%%%%%%%%%%%%%%%%%%%%%%%%%%%%%%%%%%%%%%%%%%%%%%%%%%%%%%%%%%%%%%%%%%%%
Adding up the relevant quantities, we get 
%%%%%%%%%%%%%%%%%%%%%%%%%%%%%%%%%%%%%%%%%%%%%%%%%%%%%%%%%%%%%%%%%%%%%%%%%%%%%%%%%%%%%%%%
$
\sum_{h\in E} \mathfrak{C}_{G_n}(h)
\leq 
\mathfrak{C}_{G_n}(f) + \sum_{j=1}^5 \Lambda_j
\leq 
1 + \sum_{j=1}^5 \Lambda_j
$.
%%%%%%%%%%%%%%%%%%%%%%%%%%%%%%%%%%%%%%%%%%%%%%%%%%%%%%%%%%%%%%%%%%%%%%%%%%%%%%%%%%%%%%%%
Since 
$m=m_1+m_2+1=n+\sqrt{n}+2$
and $s>0$ is a constant, 
it now follows that 
$
\lim_{n\to\infty}
\frac{s}{m}
\sum_{h\in E} 
\mathfrak{C}_{G_n}(h)
\leq 
-\frac{123\,s}{240}
$,
thus there 
exists a constant $\eps>0$ such that 
$
\frac{s}{m}
\sum_{h\in E} 
\mathfrak{C}_{G_n}(h)
<
-\eps
$
for all sufficiently large $n$.
%%%%%%%%%%%%%%%%%%%%%%%%%%%%%%%%%%%%%%%%%%%%%%%%%%%%%%%%%%%%%%%%%%%%%%%%%%%%%%%%
\end{proof}
%%%%%%%%%%%%%%%%%%%%%%%%%%%%%%%%%%%%%%%%%%%%%%%%%%%%%%%%%%%%%%%%%%%%%%%%%%%%%%%%

%%%%%%%%%%%%%%%%%%%%%%%%%%%%%%%%%%%%%%%%%%%%%%%%%%%%%%%%%%%%%%%%%%%%%%%%%%%%%%%%%%%%%%
\subsection{\label{sec-res-conv}Rapid Initial Convergence of the Ricci Flow for Directed and Undirected Hypergraphs}
%%%%%%%%%%%%%%%%%%%%%%%%%%%%%%%%%%%%%%%%%%%%%%%%%%%%%%%%%%%%%%%%%%%%%%%%%%%%%%%%%%%%%%

As we shown in Table~\ref{tab-fastconv} below, the first iteration in which the edge-weights converge 
under Equation~\eqref{eq-converg}
is a \emph{small} number for all of our (directed and undirected) hypergraphs.

%%%%%%%%%%%%%%%%%%%%%%%%%%%%%%%%%%%%%%%%%%%%%%%%%%%%%%%%%%%%%%%%%%%%%%%%%%%%%%%%%%%%%%
\begin{table}[h]
%%%%%%%%%%%%%%%%%%%%%%%%%%%%%%%%%%%%%%%%%%%%%%%%%%%%%%%%%%%%%%%%%%%%%%%%%%%%%%%%%%%%%%
\caption{\label{tab-fastconv}Fast convergence of the Ricci flows for all 
constructed hypergraphs (cf.\ Section~\ref{sec-flow-surg}).
$\eta_{\mathrm{first}}$
is the \emph{smallest} iteration after which $\Delta_{\mathrm{AVE}}\leq\eps$; 
$\eps=0.005$ for directed hypergraphs
and $\eps=0.000005$ for undirected hypergraphs. 
} 
%%%%%%%%%%%%%%%%%%%%%%%%%%%%%%%%%%%%%%%%%%%%%%%%%%%%%%%%%%%%%%%%%%%%%%%%%%%%%%%%%%%%%%
\begin{tabular}{l @{\hskip 0.3in} c }\toprule
\toprule
%%%%%%%%%%%%%%%%%%%%%%%%%%%%%%%%%%%%%%%%%%%%%%%%%%%%%%%%%%%%%%%%%%%%%%%%%%%%%%%%%%%%%%
Hypergraph name  & $\eta_{\mathrm{first}}$
%%%%%%%%%%%%%%%%%%%%%%%%%%%%%%%%%%%%%%%%%%%%%%%%%%%%%%%%%%%%%%%%%%%%%%%%%%%%%%%%%%%%%%
\\ \midrule
\midrule
%%%%%%%%%%%%%%%%%%%%%%%%%%%%%%%%%%%%%%%%%%%%%%%%%%%%%%%%%%%%%%%%%%%%%%%%%%%%%%%%%%%%%%
Escherichia Coli & $4$ \\ \midrule
Homo Sapiens & $4$ \\ \midrule
Helicobacter Pylori & $4$ \\ \midrule
Methanosarcina berkeri str. Fusaro & $4$ \\ \midrule
Mycobacterium Tuberculosis & $4$ \\ \midrule
Synechococcus elongatus & $4$ \\ \midrule
Synechocystis & $4$ \\[4pt] \midrule
\midrule
\csp\ dataset & $10$ \\ \midrule
\nsp\ dataset & $5$ 
%%%%%%%%%%%%%%%%%%%%%%%%%%%%%%%%%%%%%%%%%%%%%%%%%%%%%%%%%%%%%%%%%%%%%%%%%%%%%%%%%%%%%%
\\ \bottomrule
\bottomrule
%%%%%%%%%%%%%%%%%%%%%%%%%%%%%%%%%%%%%%%%%%%%%%%%%%%%%%%%%%%%%%%%%%%%%%%%%%%%%%%%%%%%%%
\end{tabular}
%%%%%%%%%%%%%%%%%%%%%%%%%%%%%%%%%%%%%%%%%%%%%%%%%%%%%%%%%%%%%%%%%%%%%%%%%%%%%%%%%%%%%%
\end{table}
%%%%%%%%%%%%%%%%%%%%%%%%%%%%%%%%%%%%%%%%%%%%%%%%%%%%%%%%%%%%%%%%%%%%%%%%%%%%%%%%%%%%%%

Table~\ref{tab-fastconv} show fast initial convergence of Ricci flows
via small value of the parameter $\Delta_{\mathrm{AVE}}$.
However, as mentioned in Section~\ref{sec-algm-overall}, 
it is preferable to execute more iterations to 
ensure that
the diffusion process has actually converged in several successive iterations 
based on the value of 
$\Delta_{\mathrm{AVE}}$
and that the quality parameters are within acceptable bounds.
Moreover, 
even if $\Delta_{\mathrm{AVE}}$ is within acceptable bounds, 
some individual edge weights may still change significantly in subsequent iterations
since the value of 
$\Delta_{\mathrm{STD}}$ (\emph{cf}.\ Equation~\eqref{eq-converg-2})
may not be acceptably low, and this may lead to 
changes in the cores.
In our case, we indeed found that 
further iterations produced 
decreasing values of 
$\Delta_{\mathrm{STD}}$
leading to more stable cores (see Table~\ref{tab-stabconv}).

%%%%%%%%%%%%%%%%%%%%%%%%%%%%%%%%%%%%%%%%%%%%%%%%%%%%%%%%%%%%%%%%%%%%%%%%%%%%%%%%%%%%%%
\begin{table*}
%%%%%%%%%%%%%%%%%%%%%%%%%%%%%%%%%%%%%%%%%%%%%%%%%%%%%%%%%%%%%%%%%%%%%%%%%%%%%%%%%%%%%%
\caption{\label{tab-stabconv}
Values of 
$\Delta_{\mathrm{STD}}$ (\emph{cf}.\ Equation~\eqref{eq-converg-2})
value \textbf{at the end} of $\eta$ Ricci flows iterations for all 
constructed hypergraphs (cf.\ Section~\ref{sec-flow-surg}).
For all these cases, as observed in Table~\ref{tab-fastconv}, 
the value of 
$\Delta_{\mathrm{AVE}}$ is 
at most 
$\eps=0.005$ for directed hypergraphs
and 
is at most
$\eps=0.000005$ for undirected hypergraphs. 
} 
%%%%%%%%%%%%%%%%%%%%%%%%%%%%%%%%%%%%%%%%%%%%%%%%%%%%%%%%%%%%%%%%%%%%%%%%%%%%%%%%%%%%%%
\begin{tabular}{l @{\hskip 0.3in} 
                   c @{\hskip 0.3in} 
                   c @{\hskip 0.3in} 
                   c @{\hskip 0.3in} 
                   c @{\hskip 0.3in} 
									 c
									 }\toprule
\toprule
%%%%%%%%%%%%%%%%%%%%%%%%%%%%%%%%%%%%%%%%%%%%%%%%%%%%%%%%%%%%%%%%%%%%%%%%%%%%%%%%%%%%%%
                 & 
          \multicolumn{5}{c}{
$\Delta_{\mathrm{STD}}$ values at end of $\eta=$  
          }
\\ \cmidrule{2-6}
Hypergraph name  & 
            $\eta_{\mathrm{conv}}^{\,\ddagger}$ & $10$ & $20$ & $30$ & $40$ 
%%%%%%%%%%%%%%%%%%%%%%%%%%%%%%%%%%%%%%%%%%%%%%%%%%%%%%%%%%%%%%%%%%%%%%%%%%%%%%%%%%%%%%
\\ \midrule
\midrule
%%%%%%%%%%%%%%%%%%%%%%%%%%%%%%%%%%%%%%%%%%%%%%%%%%%%%%%%%%%%%%%%%%%%%%%%%%%%%%%%%%%%%%
Escherichia Coli & 
    $0.0818$ & $0.0650$ & 
     $0.0498$ & $0.0411$ & $0.0358$  
\\ \midrule
Homo Sapiens & 
    $0.0286$ & $0.0199$ & 
     $0.0136$ & $0.0131$ & $0.0123$  
\\ \midrule
Helicobacter Pylori & 
    $0.0037$ & $0.0025$ &  
     $0.0020$ & $0.0017$ & $0.0015$  
\\ \midrule
Methanosarcina berkeri str. Fusaro & 
    $0.0872$  & $0.0592$ & 
     $0.0554$ & $0.0461$ & $0.0403$  
\\ \midrule
Mycobacterium Tuberculosis & 
    $0.0884$ & $0.0600$ &  
     $0.0530$ & $0.0441$ & $0.0385$  
\\ \midrule
Synechococcus elongatus & 
    $0.0184$ & $0.0132$ &  
     $0.0100$ & $0.0082$ & $0.0071$  
\\ \midrule
Synechocystis & 
    $0.1171$ & $0.0837$ &  
     $0.0628$ & $0.0515$ & $0.0446$  
\\[4pt] \midrule
\midrule
%%%%%%%%%%%%%%%%%%%%%%%%%%%%%%%%%%%%%%%%%%%%%%%%%%%%%%%%%%%%%%%%%%%%%%%%%%%%%%%%%%%%%%
\csp\ dataset & 
    $0.0024$ & $0.0024$ & 
     $0.0017$ & $0.0014$ & $0.0012$  
\\ \midrule
\nsp\ dataset & 
    $0.0189$ & $0.0138$ &  
     $0.0103$ & $0.0091$ & $0.0085$  
%%%%%%%%%%%%%%%%%%%%%%%%%%%%%%%%%%%%%%%%%%%%%%%%%%%%%%%%%%%%%%%%%%%%%%%%%%%%%%%%%%%%%%
\\ \bottomrule
\bottomrule
%%%%%%%%%%%%%%%%%%%%%%%%%%%%%%%%%%%%%%%%%%%%%%%%%%%%%%%%%%%%%%%%%%%%%%%%%%%%%%%%%%%%%%
\\
[-6pt]
\multicolumn{6}{c}{
$^\ddagger$the smallest value of $\eta$ for which 
$\Delta_{\mathrm{AVE}}\leq\eps$
}
%%%%%%%%%%%%%%%%%%%%%%%%%%%%%%%%%%%%%%%%%%%%%%%%%%%%%%%%%%%%%%%%%%%%%%%%%%%%%%%%%%%%%%
\end{tabular}
%%%%%%%%%%%%%%%%%%%%%%%%%%%%%%%%%%%%%%%%%%%%%%%%%%%%%%%%%%%%%%%%%%%%%%%%%%%%%%%%%%%%%%
\end{table*}
%%%%%%%%%%%%%%%%%%%%%%%%%%%%%%%%%%%%%%%%%%%%%%%%%%%%%%%%%%%%%%%%%%%%%%%%%%%%%%%%%%%%%%

%%%%%%%%%%%%%%%%%%%%%%%%%%%%%%%%%%%%%%%%%%%%%%%%%%%%%%%%%%%%%%%%%%%%%%%%%%%%%%%%%%%%%%
\subsection{\label{sec-finmod}Final Cores and Their Qualities Determined by Our Algorithm}
%%%%%%%%%%%%%%%%%%%%%%%%%%%%%%%%%%%%%%%%%%%%%%%%%%%%%%%%%%%%%%%%%%%%%%%%%%%%%%%%%%%%%%

%%%%%%%%%%%%%%%%%%%%%%%%%%%%%%%%%%%%%%%%%%%%%%%%%%%%%%%%%%%%%%%%%%%%%%%%%%%%%%%%%%%%%%
\begin{table*}
%%%%%%%%%%%%%%%%%%%%%%%%%%%%%%%%%%%%%%%%%%%%%%%%%%%%%%%%%%%%%%%%%%%%%%%%%%%%%%%%%%%%%%
\caption{\label{tab-finmod-dir}
For directed hypergraphs,
cores with their quality parameters 
(\emph{cf}.\ 
Equations~\eqref{eq-rin},\eqref{eq-rout},\eqref{eq-dir-dist},\eqref{eq-dir-discon})
as found by our algorithm.
\textbf{The $\pmb{p}$-values for the core quality parameters (\emph{cf}.\ Section~\ref{sec-pcalc}) for all cores 
were found to be significantly less than $\pmb{10^{-5}}$, so these values are not listed explicitly}.
} 
%%%%%%%%%%%%%%%%%%%%%%%%%%%%%%%%%%%%%%%%%%%%%%%%%%%%%%%%%%%%%%%%%%%%%%%%%%%%%%%%%%%%%%
\begin{tabular}{l 
    @{\hskip 0.2in} c 
    @{\hskip 0.2in} c 
    @{\hskip 0.2in} c 
    @{\hskip 0.2in} c 
    @{\hskip 0.2in} c 
    @{\hskip 0.2in} c 
    @{\hskip 0.2in} c 
    @{\hskip 0.2in} c 
	 }\toprule
\toprule
%%%%%%%%%%%%%%%%%%%%%%%%%%%%%%%%%%%%%%%%%%%%%%%%%%%%%%%%%%%%%%%%%%%%%%%%%%%%%%%%%%%%%%
\multirow{2}{*}{Hypergraph $(V,E,w)$}  & \multirow{2}{*}{$|V|$} 
    & \multirow{2}{*}{$\mathrm{core}_\#^\dagger$} & \multirow{2}{*}{$\mathrm{core\_size}^\ddagger$} 
    & \multicolumn{4}{c}{core quality parameters}
\\ \cmidrule{5-8}
                                       & 
    &                                             & 
		& $\mathscr{r}_{\mathrm{in}}^{deg}$ & $\mathscr{r}_{\mathrm{out}}^{deg}$ 
		& $\mathscr{r}_{\mathrm{dist\_stretch}}^{directed}$ 
    & $\mathscr{r}_{\mathrm{disconnected}}^{\mathrm{ordered\_pairs}}$
%%%%%%%%%%%%%%%%%%%%%%%%%%%%%%%%%%%%%%%%%%%%%%%%%%%%%%%%%%%%%%%%%%%%%%%%%%%%%%%%%%%%%%
\\ \midrule
\midrule
%%%%%%%%%%%%%%%%%%%%%%%%%%%%%%%%%%%%%%%%%%%%%%%%%%%%%%%%%%%%%%%%%%%%%%%%%%%%%%%%%%%%%%
Escherichia Coli &  $762$ 
   & $1$ 
	 & $326$ & $0.7259$ & $0.7340$ & $2.7043$ & $0.1839$ %%$69911$
\\ \midrule
Homo Sapiens &  $343$ 
   & $1$ 
	 & $144$ & $0.6960$ & $0.7400$ & $2.7065$ & $0.4303$ %%$16785$
\\ \midrule
Helicobacter Pylori & $486$ 
   & $1$ 
	 & $230$ & $0.7859$ & $0.7501$ & $2.7065$ & $0.3215$ %%$20991$
\\ \midrule
Methanosarcina berkeri str. Fusaro &  $629$ 
   & $1$ 
	 & $254$ & $0.7829$ & $0.7590$ & $2.7337$ & $0.2463$ %%$34554$
\\ \midrule
Mycobacterium Tuberculosis &  $826$ 
   & $1$ 
	 & $377$ & $0.7993$ & $0.7092$ & $2.7236$ & $0.2886$ %%$71719$
\\ \midrule
Synechococcus elongatus &  $769$ 
   & $1$ 
	 & $305$ & $0.8063$ & $0.7907$ & $2.7211$ & $0.7345$ %%$157801$
\\ \midrule
Synechocystis & $796$
   & $1$ 
	 & $305$ & $0.8058$ & $0.7694$ & $2.6035$ & $0.1622$ %%$62681$
%%%%%%%%%%%%%%%%%%%%%%%%%%%%%%%%%%%%%%%%%%%%%%%%%%%%%%%%%%%%%%%%%%%%%%%%%%%%%%%%%%%%%%
\\ \bottomrule
\bottomrule
%%%%%%%%%%%%%%%%%%%%%%%%%%%%%%%%%%%%%%%%%%%%%%%%%%%%%%%%%%%%%%%%%%%%%%%%%%%%%%%%%%%%%%
\\
[-6pt]
\multicolumn{8}{c}{
\begin{tabular}{ r l}
$^{\dagger}\mathrm{core}_\#\,$ : & number of cores 
\\
$^{\ddagger}\mathrm{core\_size}\,$ : & number of nodes in the cores 
\end{tabular}
}
%%%%%%%%%%%%%%%%%%%%%%%%%%%%%%%%%%%%%%%%%%%%%%%%%%%%%%%%%%%%%%%%%%%%%%%%%%%%%%%%%%%%%%
\end{tabular}
%%%%%%%%%%%%%%%%%%%%%%%%%%%%%%%%%%%%%%%%%%%%%%%%%%%%%%%%%%%%%%%%%%%%%%%%%%%%%%%%%%%%%%
\end{table*}
%%%%%%%%%%%%%%%%%%%%%%%%%%%%%%%%%%%%%%%%%%%%%%%%%%%%%%%%%%%%%%%%%%%%%%%%%%%%%%%%%%%%%%

%%%%%%%%%%%%%%%%%%%%%%%%%%%%%%%%%%%%%%%%%%%%%%%%%%%%%%%%%%%%%%%%%%%%%%%%%%%%%%%%%%%%%%
\begin{table*}
%%%%%%%%%%%%%%%%%%%%%%%%%%%%%%%%%%%%%%%%%%%%%%%%%%%%%%%%%%%%%%%%%%%%%%%%%%%%%%%%%%%%%%
\caption{\label{tab-finmod-undir}
For undirected hypergraphs,
cores with their quality parameters 
(\emph{cf}.\ 
Equations~\eqref{eq-r},\eqref{eq-undir-dist},\eqref{eq-undir-discon})
as found by our algorithm.
\textbf{The $\pmb{p}$-values for the core quality parameters (\emph{cf}.\ Section~\ref{sec-pcalc}) for all cores 
were found to be significantly less than $\pmb{10^{-5}}$, so these values are not listed explicitly}.
} 
%%%%%%%%%%%%%%%%%%%%%%%%%%%%%%%%%%%%%%%%%%%%%%%%%%%%%%%%%%%%%%%%%%%%%%%%%%%%%%%%%%%%%%
\begin{tabular}{l 
    @{\hskip 0.3in} c 
    @{\hskip 0.3in} c 
    @{\hskip 0.3in} c 
    @{\hskip 0.3in} c 
    @{\hskip 0.3in} c 
    @{\hskip 0.3in} c 
	 }\toprule
\toprule
%%%%%%%%%%%%%%%%%%%%%%%%%%%%%%%%%%%%%%%%%%%%%%%%%%%%%%%%%%%%%%%%%%%%%%%%%%%%%%%%%%%%%%
\multirow{2}{*}{Hypergraph $(V,E,w)$}  & \multirow{2}{*}{$|V|$} 
    & \multirow{2}{*}{$\mathrm{core}_\#^\dagger$} & \multirow{2}{*}{$\mathrm{core\_size}^\ddagger$} 
    & \multicolumn{3}{c}{core quality parameters} 
\\ \cmidrule{5-7}
                                       & 
		&                                             & 
    & $\mathscr{r}^{deg}$
		& $\mathscr{r}_{\mathrm{dist\_stretch}}^{undirected}$
    & $\mathscr{r}_{\mathrm{disconnected}}^{\mathrm{unordered\_pairs}}$
%%%%%%%%%%%%%%%%%%%%%%%%%%%%%%%%%%%%%%%%%%%%%%%%%%%%%%%%%%%%%%%%%%%%%%%%%%%%%%%%%%%%%%
\\ \midrule
\midrule
%%%%%%%%%%%%%%%%%%%%%%%%%%%%%%%%%%%%%%%%%%%%%%%%%%%%%%%%%%%%%%%%%%%%%%%%%%%%%%%%%%%%%%
\csp\ dataset & $496$  
   & $1$ 
	 & $62^{a}$ & $0.6411$ & $2.8964$ & $0.0006$ %%$161308$
\\ \midrule
\multirow{2}{*}{\nsp\ dataset} & \multirow{2}{*}{$518$}
   & \multirow{2}{*}{$2$} & $24^{b}$  
	                 & $0.7530$ &  $3.4574$ & $0.0002$ %%$213680$
\\ \cmidrule{4-7}
                                    &              
	 &                      & $23^{c}$ 
	                 & $0.7600$ &  $1.6309$ & $0.0001$ %%$219066$
%%%%%%%%%%%%%%%%%%%%%%%%%%%%%%%%%%%%%%%%%%%%%%%%%%%%%%%%%%%%%%%%%%%%%%%%%%%%%%%%%%%%%%
\\ \bottomrule
\bottomrule
%%%%%%%%%%%%%%%%%%%%%%%%%%%%%%%%%%%%%%%%%%%%%%%%%%%%%%%%%%%%%%%%%%%%%%%%%%%%%%%%%%%%%%
\\
[-6pt]
\multicolumn{7}{c}{
\begin{tabular}{ r l}
$^{\dagger}\mathrm{core}_\#\,$ : & number of cores 
\\
$^{\ddagger}\mathrm{core\_size}\,$ : & number of nodes in the cores 
\end{tabular}
}
%%%%%%%%%%%%%%%%%%%%%%%%%%%%%%%%%%%%%%%%%%%%%%%%%%%%%%%%%%%%%%%%%%%%%%%%%%%%%%%%%%%%%%
\\
%%%%%%%%%%%%%%%%%%%%%%%%%%%%%%%%%%%%%%%%%%%%%%%%%%%%%%%%%%%%%%%%%%%%%%%%%%%%%%%%%%%%%%
\multicolumn{7}{l}{
\begin{tabular}{l}
\small
$^a$Authors in the core:
\\
\hspace*{0.2in}
%%\begin{tabular}{p{0.9\textwidth}}
\begin{tabular}{l}
\small
M. Bellare, R. Impagliazzo, M. Szegedy, L. Fortnow, M. Yung, B. Waters, D. Dolev, L. Babai, M. Luby, 
\\
\small
A. Lysyanskaya, A. Wigderson, A. Sahai, M. Sudan, J. Naor, T. Okamoto, L. Levin, Omer Reingold,
\\
\small
N. Buchbinder, P. Rogaway, C. Dwork, H. Krawczyk, H. Buhrman, O. Goldreich, B. Barak, D. Boneh,
\\
\small
N. Nisan, Y. Lindell, S. Halevi, S. Keelveedhi, A. Herzberg, S. Wolf, G. Segev, R. Rivest, S. Vadhan,
\\
\small
C. Peikert, R. Santhanam, J. Stern, E. Fujisaki, V. Kabanets, S. Goldwasser, L. Trevisan,
\\
\small
S. Jarecki, A. O'Neill, G. Rothblum, M. Kharitonov, R. Schwartz, D. Melkebeek, M. Naor, R. Canetti,
\\
\small
A. Goldberg, S. Micali, A. Boldyreva, D. Pointcheval, S. Arora, K. Yang, N. Linial, J. Hastad, 
\\
\small
S. Rudich, E. Sweedyk, R. Ostrovsky, A. Desai, M. Feldman
\end{tabular}
%%%%%%%%%%%%%%%%%%%%%%%%%%%%%%%%%%%%%%%%%%%%%%%%%%%%%%%%%%%%%%%%%%%%%%%%%%%%%%%%%%%%%%
\\
%%%%%%%%%%%%%%%%%%%%%%%%%%%%%%%%%%%%%%%%%%%%%%%%%%%%%%%%%%%%%%%%%%%%%%%%%%%%%%%%%%%%%%
\small
$^b$Authors in the core:
\\
\hspace*{0.2in}
%%\begin{tabular}{p{0.9\textwidth}}
\begin{tabular}{l}
\small
J. Zhuang, C. Bauch, M. Perc, Y. Tian, A. Sheikhahmadi, C. Hens, M. Duh, Y. Mu, Y. Xia, W. Lin,
\\
\small
P. Ji, K. Skok, M. Milojevic, J. Ye, J. Sun, J. Kurths, Z. Cheng, A. Zareie, J. Cao, Y. Tang,
\\
\small
L. Guerrini, M. S. K. Fasaei, L. Tang, M Gosak
\end{tabular}
%%%%%%%%%%%%%%%%%%%%%%%%%%%%%%%%%%%%%%%%%%%%%%%%%%%%%%%%%%%%%%%%%%%%%%%%%%%%%%%%%%%%%%
\\
%%%%%%%%%%%%%%%%%%%%%%%%%%%%%%%%%%%%%%%%%%%%%%%%%%%%%%%%%%%%%%%%%%%%%%%%%%%%%%%%%%%%%%
\small
$^c$Authors in the core:
\\
\hspace*{0.2in}
%%\begin{tabular}{p{0.9\textwidth}}
\begin{tabular}{l}
\small
X. Wu, W. Han, D. Zhao, C. Lv, E. M. Ruiz, P. A. C. Sousa, L.-L. Jiang, A. Nicchi, S. Boccaletti,
\\
\small
J. Gao, Z. Wang, D. Duan, E. Kubik, H. Stanley, S. Havlin, L. Wang, S. Li, Q. Su, A. Li, S. Si,
\\
\small
D. Li, M. Zanin, D. Papo
\end{tabular}
%%%%%%%%%%%%%%%%%%%%%%%%%%%%%%%%%%%%%%%%%%%%%%%%%%%%%%%%%%%%%%%%%%%%%%%%%%%%%%%%%%%%%%
\end{tabular}
}
%%%%%%%%%%%%%%%%%%%%%%%%%%%%%%%%%%%%%%%%%%%%%%%%%%%%%%%%%%%%%%%%%%%%%%%%%%%%%%%%%%%%%%
\end{tabular}
%%%%%%%%%%%%%%%%%%%%%%%%%%%%%%%%%%%%%%%%%%%%%%%%%%%%%%%%%%%%%%%%%%%%%%%%%%%%%%%%%%%%%%
\end{table*}
%%%%%%%%%%%%%%%%%%%%%%%%%%%%%%%%%%%%%%%%%%%%%%%%%%%%%%%%%%%%%%%%%%%%%%%%%%%%%%%%%%%%%%

We report in Table~\ref{tab-finmod-dir} and Table~\ref{tab-finmod-undir}
the cores found by our algorithm with their quality parameters via 
Equations~\eqref{eq-rin}--\eqref{eq-undir-discon}.
As can be seen, taken together 
the parameters 
$\mathscr{r}_{\mathrm{in}}^{deg}$, $\mathscr{r}_{\mathrm{out}}^{deg}$, 
$\mathscr{r}_{\mathrm{dist\_stretch}}^{directed}$,  
$\mathscr{r}_{\mathrm{disconnected}}^{\mathrm{ordered\_pairs}}$
for directed hypergraphs
and 
the parameters 
$\mathscr{r}^{deg}$, $\mathscr{r}_{\mathrm{dist\_stretch}}^{undirected}$, 
$\mathscr{r}_{\mathrm{disconnected}}^{\mathrm{unordered\_pairs}}$
for undirected hypergraphs 
indicate a good quality of modularity and centrality of the cores
and satisfy all the validity criteria set forth in 
Sections \ref{sec-conn}--\ref{sec-pcalc}.
For example, for 
\emph{Synechococcus elongatus} 
hyperedges with only nodes from the core both in their head and tail
contributed $80\%$ on average to the in-degrees of the nodes  
and similarly 
hyperedges with only nodes from the core both in their head and tail
contributed $79\%$ on average to the out-degrees of the nodes.
Furthermore, 
for \emph{Synechococcus elongatus} 
about $73\%$ of ordered pairs of nodes not in the core that were connected by a path lose this path
when the core is removed, 
and ordered node pairs that stay connected increase their distance by about $272\%$ on average.

An inspection of the values in 
Table~\ref{tab-finmod-dir} and Table~\ref{tab-finmod-undir}
shows that 
removal of the cores disconnect 
\emph{fewer} pairs of nodes in the core
for undirected hypergraphs
as compared to the directed hypergrahs (\IE, 
the $\mathscr{r}_{\mathrm{disconnected}}^{\mathrm{unordered\_pairs}}$
values are smaller than the 
$\mathscr{r}_{\mathrm{disconnected}}^{\mathrm{ordered\_pairs}}$ values),
even though the stretch factors of shortest paths surviving after core removals 
are comparable (\IE,
the $\mathscr{r}^{deg}$, $\mathscr{r}_{\mathrm{dist\_stretch}}^{undirected}$
values are of similar magnitudes to the values of 
$\mathscr{r}_{\mathrm{dist\_stretch}}^{directed}$).
There are \emph{two} possible reasons for this behavior.
Firstly, 
the sizes of cores for undirected hypergraphs are \emph{smaller} than the 
typical core sizes of  
directed hypergraphs, and 
intuitively one would expect removal of \emph{fewer} nodes to disconnect less number of
remaining paths.
Secondly,
the \emph{directionality} constraints on paths for directed hypergraphs (\IE, 
edges may \emph{not} be traversed in the wrong direction) 
allow fewer avenues to ``bypass'' the nodes in the core; indeed,
majority of biochemical reactions in our dataset are \emph{not} bidirectional reactions.

In the following two sections, we comment on interpreting the cores and their 
%significances in reference
implications
to future research works. 

%%%%%%%%%%%%%%%%%%%%%%%%%%%%%%%%%%%%%%%%%%%%%%%%%%%%%%%%%%%%%%%%%%%%%%%%%%%%%%%%%%%%%%
\subsubsection{\label{sec-finmod-dir-int}Interpretation and Usefulness of Cores for Directed Hypergraphs (Metabolic Systems)}
%%%%%%%%%%%%%%%%%%%%%%%%%%%%%%%%%%%%%%%%%%%%%%%%%%%%%%%%%%%%%%%%%%%%%%%%%%%%%%%%%%%%%%

A core $\cS\subset V$ of the directed hypergraph $G=(V,E,w)$
corresponds to a set of 
% I think it will be clearer to just use "molecules"
%\emph{biochemical entities} 
molecules
(reactants and/or products) in the 
biochemical system under consideration. 
Our analysis indicates the following properties for this subset of molecules:
%%%%%%%%%%%%%%%%%%%%%%%%%%%%%%%%%%%%%%%%%%%%%%%%%%%%%%%%%%%%%%%%%%%%%%%%%%%%%%%%%%%%%%
\begin{enumerate}[label={\textbf{({\Roman*})}}]
%%%%%%%%%%%%%%%%%%%%%%%%%%%%%%%%%%%%%%%%%%%%%%%%%%%%%%%%%%%%%%%%%%%%%%%%%%%%%%%%%%%%%%
\item
The high values of 
$\mathscr{r}_{\mathrm{in}}^{deg}$ 
and
$\mathscr{r}_{\mathrm{out}}^{deg}$ 
in Table~\ref{tab-finmod-dir}
suggest that each 
molecule (node) $u$ in the core
is highly dependent on 
other 
molecules (nodes) in the core, \EG, 
for another 
molecule $v$ in the core
either $u$ and $v$ are both reactants in the same biochemical reaction, 
which itself is part of the core
or one of the them is a product produced in a core biochemical reaction 
in which the other one is a reactant. 
%%%%%%%%%%%%%%%%%%%%%%%%%%%%%%%%%%%%%%%%%%%%%%%%%%%%%%%%%%%%%%%%%%%%%%%%%%%%%%%%%%%%%%
\item
The values of 
$\mathscr{r}_{\mathrm{disconnected}}^{\mathrm{ordered\_pairs}}$
and 
$\mathscr{r}_{\mathrm{dist\_stretch}}^{directed}$ 
signify the importance of the 
molecules in the core
in the overall functioning of the biochemical system in the following manner.
%%%%%%%%%%%%%%%%%%%%%%%%%%%%%%%%%%%%%%%%%%%%%%%%%%%%%%%%%%%%%%%%%%%%%%%%%%%%%%%%%%%%%%
\begin{enumerate}[label={\textbf{({\Alph*})}}]
%%%%%%%%%%%%%%%%%%%%%%%%%%%%%%%%%%%%%%%%%%%%%%%%%%%%%%%%%%%%%%%%%%%%%%%%%%%%%%%%%%%%%%
\item
Consider an ordered pair $(u,v)$ which
contributes towards the value of 
$\mathscr{r}_{\mathrm{disconnected}}^{\mathrm{ordered\_pairs}}$ by losing their path when the core is removed.
This implies that the production of $v$ can be significantly reduced or perhaps 
completely disrupted by the removal of the core $\cS$.
%%%%%%%%%%%%%%%%%%%%%%%%%%%%%%%%%%%%%%%%%%%%%%%%%%%%%%%%%%%%%%%%%%%%%%%%%%%%%%%%%%%%%%
\item
Consider an ordered pair $(u,v)$ 
such that
$\frac{\dist_{H\setminus\cS}(u,v)}{\dist_{H\setminus\cS}(u,v)}$ 
is large enough to 
contribute significantly towards the value of 
$\mathscr{r}_{\mathrm{dist\_stretch}}^{directed}$. 
Then removal of the core $\cS$ 
may significantly delay
the production of $v$ by increasing the number of reactions needed for its production.
%%%%%%%%%%%%%%%%%%%%%%%%%%%%%%%%%%%%%%%%%%%%%%%%%%%%%%%%%%%%%%%%%%%%%%%%%%%%%%%%%%%%%%
\end{enumerate}
%%%%%%%%%%%%%%%%%%%%%%%%%%%%%%%%%%%%%%%%%%%%%%%%%%%%%%%%%%%%%%%%%%%%%%%%%%%%%%%%%%%%%%
\end{enumerate}
%%%%%%%%%%%%%%%%%%%%%%%%%%%%%%%%%%%%%%%%%%%%%%%%%%%%%%%%%%%%%%%%%%%%%%%%%%%%%%%%%%%%%%
Since removal of the core(s) 
\emph{significantly} affects the overall functioning of the biochemical system, 
one can conclude that the core plays a dominant role in the system. 
Following a standard practice used in computational biology, researchers may focus on the 
%biochemical entities
molecules in the core and their associated biochemical reactions for further computational 
analysis if the original system was too computationally intensive to analyse because of its size.

%%%%%%%%%%%%%%%%%%%%%%%%%%%%%%%%%%%%%%%%%%%%%%%%%%%%%%%%%%%%%%%%%%%%%%%%%%%%%%%%%%%%%%
\paragraph{\label{sec-rem-cor}\color{black}Removing a core by biological experiments}
%%%%%%%%%%%%%%%%%%%%%%%%%%%%%%%%%%%%%%%%%%%%%%%%%%%%%%%%%%%%%%%%%%%%%%%%%%%%%%%%%%%%%%
%
Perturbation experiments (\EG, genetic knockouts) are well-established 
biological method for probing the importance of biochemical entities. The best way to eliminate a 
chemical reaction is to knock out the enzyme that catalyzes the reaction.
We \emph{briefly} comment here on
\emph{optimally} designing the biological experimental mechanism to remove a core
$\cS$,
assuming that we have the data corresponding to the \emph{enzymes}
for each reaction (our datasets did \emph{not} provide this information).  
Let $\cE$ be the set of all biochemical reactions in which one or more
members of $\cS$ appear as either a reactant or a product (or both); 
thus, disabling the reactions in $\cE$ will effectively disconnect the core $\cS$.
Let $\cC$ be the set of enzymes catalyzing the reactions from $\cE$.
Then, a \emph{minimal} set of enzyme knockouts that can be used to disable all the reactions 
in $\cE$ can be determined by solving an appropriate \emph{minimum hitting set} problem~\cite{GJ79} defined such that 
the universe is $\cE$ and corresponding to each enzyme
$c\in\cC$ we have a set $\{e\in\cE \,|\,\text{$c$ is a catalyst for $e$} \}$.
The reader is referred to standard literatures in computer algorithms such as~\cite{V01}
for methods to solve a minimum hitting set problem efficiently.

%%%%%%%%%%%%%%%%%%%%%%%%%%%%%%%%%%%%%%%%%%%%%%%%%%%%%%%%%%%%%%%%%%%%%%%%%%%%%%%%%%%%%%
\subsubsection{\label{sec-finmod-undir-int}Interpretation and Usefulness of Cores for 
Undirected Hypergraphs (Co-author Relationships)}
%%%%%%%%%%%%%%%%%%%%%%%%%%%%%%%%%%%%%%%%%%%%%%%%%%%%%%%%%%%%%%%%%%%%%%%%%%%%%%%%%%%%%%

%%%%%%%%%%%%%%%%%%%%%%%%%%%%%%%%%%%%%%%%%%%%%%%%%%%%%%%%%%%%%%%%%%%%%%%%%%%%%%%%%%%%%%
\begin{figure}
\caption{\label{fig444}A visual illustration of the example discussed in 
item \textbf{(I)} in Section~\ref{sec-finmod-undir-int} for $k=8$.
The cores are shown in dotted black lines.
The hypergraph representation captures the cores correctly whereas the graph
representation does not.}
\includegraphics[width=0.4\textwidth]{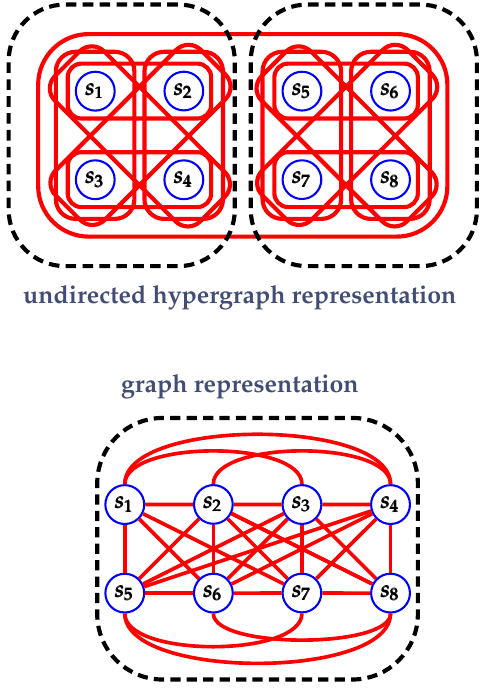}
\end{figure}  
%%%%%%%%%%%%%%%%%%%%%%%%%%%%%%%%%%%%%%%%%%%%%%%%%%%%%%%%%%%%%%%%%%%%%%%%%%%%%%%%%%%%%%

A core $\cS\subset V$ of the undirected hypergraph $G=(V,E,w)$
corresponds to a set of authors.
Below we explain what the core signifies in terms of 
its properties:
%%%%%%%%%%%%%%%%%%%%%%%%%%%%%%%%%%%%%%%%%%%%%%%%%%%%%%%%%%%%%%%%%%%%%%%%%%%%%%%%%%%%%%
\begin{enumerate}[label={\textbf{({\Roman*})}}]
%%%%%%%%%%%%%%%%%%%%%%%%%%%%%%%%%%%%%%%%%%%%%%%%%%%%%%%%%%%%%%%%%%%%%%%%%%%%%%%%%%%%%%
\item
The values of 
$\mathscr{r}^{deg}$
in Table~\ref{tab-finmod-undir}
suggest that 
the authors in $\cS$ form a ``close group'' of collaborators in the sense 
that they wrote more papers with each other as compared to with authors 
outside the core.

\textbf{We show by an illustrative example that the above need \emph{not} be the case if cores are
found by graph-theoretic approaches}.
For some large even $k>2$, suppose that authors $a_1,\dots,a_k$ wrote a paper,
authors $a_i,a_j$ wrote a paper for $1\leq i<j\leq \nicefrac{k}{2}$ 
and 
authors $a_i,a_j$ wrote a paper for $\nicefrac{k}{2}+1\leq i<j\leq k$ 
(see \FI{fig444} for a visual illustration when $k=10$).
The standard graph-theoretic approach (\emph{cf}.\ \FI{fig333})
will group all the $k$ nodes $a_1,\dots,a_k$ in the same core
since they form a large $k$-clique, 
whereas hypergraph-theoretic approach will correctly group them in two cores 
containing 
$a_1,\dots,a_{\nicefrac{k}{2}}$
and 
$a_{1+\nicefrac{k}{2}},\dots,a_k$, respectively.
%%%%%%%%%%%%%%%%%%%%%%%%%%%%%%%%%%%%%%%%%%%%%%%%%%%%%%%%%%%%%%%%%%%%%%%%%%%%%%%%%%%%%%
\item
The values of 
$\mathscr{r}_{\mathrm{dist\_stretch}}^{undirected}$
and 
$\mathscr{r}_{\mathrm{disconnected}}^{\mathrm{unordered\_pairs}}$
signify the importance of the 
authors in $\cS$ in \emph{promoting} collaborations between other researchers 
in the following manner.
Note that for a path 
$\cP_{x,y}=(x=v_1,e_1,\dots,v_k,e_k,v_{k+1}=y)$
between nodes $x$ and $y$ 
$v_i$ and $v_{i+1}$ are co-authors for $i=1,\dots,k$.
%%%%
Consider a pair of authors $\{x,y\}$ 
contributing towards the value of 
$\mathscr{r}_{\mathrm{disconnected}}^{\mathrm{unordered\_pairs}}$.
This implies complete disruption 
of collaborations between any pairs of authors where 
one of them is in the set of authors reachable from $u$ via chains of collaboration 
and 
the other one is in the set of authors reachable from $v$ via chains of collaboration. 
If for a pair of authors $\{x,y\}$ 
$\frac{\dist_{H\setminus\cS}(u,v)}{\dist_{H\setminus\cS}(u,v)}$ 
was large enough to 
contribute significantly towards the value of 
$\mathscr{r}_{\mathrm{dist\_stretch}}^{undirected}$
then, 
even though some collaborative chains connecting $u$ and $v$ are not completely disrupted, 
they are elongated leading to a decrease in productivity. 
%%%%%%%%%%%%%%%%%%%%%%%%%%%%%%%%%%%%%%%%%%%%%%%%%%%%%%%%%%%%%%%%%%%%%%%%%%%%%%%%%%%%%%
\end{enumerate}
%%%%%%%%%%%%%%%%%%%%%%%%%%%%%%%%%%%%%%%%%%%%%%%%%%%%%%%%%%%%%%%%%%%%%%%%%%%%%%%%%%%%%%

%%%%%%%%%%%%%%%%%%%%%%%%%%%%%%%%%%%%%%%%%%%%%%%%%%%%%%%%%%%%%%%%%%%%%%%%%%%%%%%%%%%%%%
\section{Conclusion}
%%%%%%%%%%%%%%%%%%%%%%%%%%%%%%%%%%%%%%%%%%%%%%%%%%%%%%%%%%%%%%%%%%%%%%%%%%%%%%%%%%%%%%

In this article, we have designed and implemented an algorithmic paradigm for finding
cores in edge-weighted directed and undirected hypergraphs
using a hypergraph-curvature guided discrete time diffusion process, 
and have \emph{successfully} 
applied our methods to 
seven metabolic hypergraphs 
and 
two social (co-authorship) hypergraphs.
\emph{En route}, 
we have also shown 
that an edge-weight re-normalization procedure in
a prior research work for Ricci flows has undesirable properties.
Finding cores of hypergraphs is \emph{still} a relatively new 
research topic that is growing rapidly, and we expect that our 
our work will provide further impetus and guidance to this
burgeoning research area. We conclude by listing a few research questions 
in this direction that may be of interest to 
%future 
researchers:
%%%%%%%%%%%%%%%%%%%%%%%%%%%%%%%%%%%%%%%%%%%%%%%%%%%%%%%%%%%%%%%%%%%%%%%%%%%%%%%%%%%%%%%%
\begin{enumerate}[label=$\triangleright$,leftmargin=*]
%%%%%%%%%%%%%%%%%%%%%%%%%%%%%%%%%%%%%%%%%%%%%%%%%%%%%%%%%%%%%%%%%%%%%%%%%%%%%%%%%%%%%%%%
\item
Our calculations of Ricci curvatures of 
directed hypergraphs in Section~\ref{sec-ricci-dir}
and 
directed hypergraphs in Section~\ref{sec-ricci-undir}
is a generalization of the corresponding calculations for undirected and directed
graphs, respectively.
However, it is possible to carry out these generalizations in different ways,
and thus it will be of interest to know if other generalizations produce better qualities
for cores of hypergraphs.
%%%%%%%%%%%%%%%%%%%%%%%%%%%%%%%%%%%%%%%%%%%%%%%%%%%%%%%%%%%%%%%%%%%%%%%%%%%%%%%%%%%%%%%%
\item
It would be of interest to see if the discrete time diffusion process using Ricci curvatures 
\emph{can} be combined with random walks in hypergraphs~\cite{Carletti_2021}
for application domains such as modular
decompositions of graphs.
%%%%%%%%%%%%%%%%%%%%%%%%%%%%%%%%%%%%%%%%%%%%%%%%%%%%%%%%%%%%%%%%%%%%%%%%%%%%%%%%%%%%%%%%
\item
In spite of writing our code in Python (which runs slower than codes written in C or similar 
other programming languages) and using a single older (slower) laptop,  
we were able to finish 
our experimental work within reasonable time, 
and therefore we made no attempt to rewrite the code in a more suitable programming language, 
optimize the code and use a faster computational device.
We believe this could be due to 
two reasons: 
\textbf{(\emph{i})}
our hypergraphs were of moderate size (number of nodes) and moderate density (number and size of hyperedges),
and 
\textbf{(\emph{ii})}
our Ricci flow iterations converged within a few steps. 
However, computational speed could become an issue with larger hypergraphs 
or if Ricci flows required more steps to converge.
Here we provide a few suggestions to 
researchers to overcome possible future
computational bottlenecks. Note that the most computational intensive part of the
curvature calculations involve the following two computational components:
\textbf{(\emph{a})}
computing the 
Earth Mover's Distance (\emd$_H(\PPL,\PPR)$) for each directed hyperedge or for each pair of nodes 
within an undirected hyperedge for every iteration, 
and 
\textbf{(\emph{b})}
the all-pairs distance calculations once for the entire hypergraph for every iteration.

For all-pairs distance calculations we either used 
a straightforward adoption of 
Floyd-Warshall's all-pair shortest path calculations with early terminations
for graphs to our hypergraphs, or used a simple breadth-first-search based approach.
However, the standard all-pairs-shortest-path problems for
graphs have a long and rich algorithmic history~\cite{CLRS09} 
with strong connections to matrix multiplication algorithms~\cite{S95,Zw02}
and other problems~\cite{WW18}, and future researchers could implement 
some of these advanced algorithmic methods for their computations. 
Note that in many applications it may even suffice to compute the distances 
\emph{approximately} 
and in that case algorithms such as in~\cite{Dor00} could be useful.

We have the following suggestions for future researchers regarding the 
\emd\ calculations:
%%%%%%%%%%%%%%%%%%%%%%%%%%%%%%%%%%%%%%%%%%%%%%%%%%%%%%%%%%%%%%%%%%%%%%%%%%%%%%%%%%%%%%%%
\begin{enumerate}[label=$\triangleright$,leftmargin=*]
%%%%%%%%%%%%%%%%%%%%%%%%%%%%%%%%%%%%%%%%%%%%%%%%%%%%%%%%%%%%%%%%%%%%%%%%%%%%%%%%%%%%%%%%
\item
Note that within each iteration, the calculations of the \emd\ values for the hyperedges
can be done in parallel, and thus a clustered computing could be used.
%%%%%%%%%%%%%%%%%%%%%%%%%%%%%%%%%%%%%%%%%%%%%%%%%%%%%%%%%%%%%%%%%%%%%%%%%%%%%%%%%%%%%%%%
\item
If it suffices to compute the \emd\ values approximately, one could implement 
the $\eps$-additive approximation algorithms in~\cite{Q18,DGK18}.
%%%%%%%%%%%%%%%%%%%%%%%%%%%%%%%%%%%%%%%%%%%%%%%%%%%%%%%%%%%%%%%%%%%%%%%%%%%%%%%%%%%%%%%%
\end{enumerate}
%%%%%%%%%%%%%%%%%%%%%%%%%%%%%%%%%%%%%%%%%%%%%%%%%%%%%%%%%%%%%%%%%%%%%%%%%%%%%%%%%%%%%%%%
\end{enumerate}
\appendix*
%%%%%%%%%%%%%%%%%%%%%%%%%%%%%%%%%%%%%%%%%%%%%%%%%%%%%%%%%%%%%%%%%%%%%%%%%%%%%%%%%%%%%%%%
\section{Details of Calculations of $\PPR$ for Weighted Directed Hypergraphs in Section~\ref{sec-ricci-dir}}
%%%%%%%%%%%%%%%%%%%%%%%%%%%%%%%%%%%%%%%%%%%%%%%%%%%%%%%%%%%%%%%%%%%%%%%%%%%%%%%%%%%%%%%%

%%%%%%%%%%%%%%%%%%%%%%%%%%%%%%%%%%%%%%%%%%%%%%%%%%%%%%%%%%%%%%%%%%%%%%%%%%%%%%%%%%%%%%%%%%%%%%%%%%%%%%%%%%%%%%%%
\begin{enumerate}[label=$\triangleright$]
%%%%%%%%%%%%%%%%%%%%%%%%%%%%%%%%%%%%%%%%%%%%%%%%%%%%%%%%%%%%%%%%%%%%%%%%%%%%%%%%%%%%%%%%%%%%%%%%%%%%%%%%%%%%%%%%
\item
Initially, $\PPR(u)=0$ for all $u\in V$. In our subsequent steps, we will add to these values as appropriate.
%%%%%%%%%%%%%%%%%%%%%%%%%%%%%%%%%%%%%%%%%%%%%%%%%%%%%%%%%%%%%%%%%%%%%%%%%%%%%%%%%%%%%%%%%%%%%%%%%%%%%%%%%%%%%%%%
\item
We divide the total probability $1$ equally among the nodes in $\cHe_e$, thus ``allocating'' a value of 
$(| \cHe_e |)^{-1}$ 
to each node in question.
%%%%%%%%%%%%%%%%%%%%%%%%%%%%%%%%%%%%%%%%%%%%%%%%%%%%%%%%%%%%%%%%%%%%%%%%%%%%%%%%%%%%%%%%%%%%%%%%%%%%%%%%%%%%%%%%
\item
For every node $x\in\cHe_e$ with 
$\deg_x^{out}=0$, 
we add 
$(| \cHe_e |)^{-1}$ 
to 
$\PPL(x)$.
%%%%%%%%%%%%%%%%%%%%%%%%%%%%%%%%%%%%%%%%%%%%%%%%%%%%%%%%%%%%%%%%%%%%%%%%%%%%%%%%%%%%%%%%%%%%%%%%%%%%%%%%%%%%%%%%
\item
For every node $x\in\cHe_e$ with 
$\deg_x^{out}>0$, 
we perform the following:
%%%%%%%%%%%%%%%%%%%%%%%%%%%%%%%%%%%%%%%%%%%%%%%%%%%%%%%%%%%%%%%%%%%%%%%%%%%%%%%%%%%%%%%%%%%%%%%%%%%%%%%%%%%%%%%%
\begin{enumerate}[label=$\triangleright$]
%%%%%%%%%%%%%%%%%%%%%%%%%%%%%%%%%%%%%%%%%%%%%%%%%%%%%%%%%%%%%%%%%%%%%%%%%%%%%%%%%%%%%%%%%%%%%%%%%%%%%%%%%%%%%%%%
\item
We divide the probability 
$(| \cHe_e |)^{-1}$ 
equally among the hyperedges 
$e'$ such that $x\in\cTa_{e'}$, 
thus ``allocating'' a value of 
$(| \cHe_e |\times\deg_x^{out})^{-1}$ 
to each hyperedge in question.
%%%%%%%%%%%%%%%%%%%%%%%%%%%%%%%%%%%%%%%%%%%%%%%%%%%%%%%%%%%%%%%%%%%%%%%%%%%%%%%%%%%%%%%%%%%%%%%%%%%%%%%%%%%%%%%%
\item
For each such hyperedge 
$e'$ such that $x\in\cTa_{e'}$, 
we divide the allocated value equally among the nodes in 
$\cHe_{e'}$ and add this values to the probabilities of these nodes.
In other words, 
for every node $y\in\cHe_{e'}$
we add 
$(| \cHe_e |\times\deg_x^{out} \times |\cHe_{e'}|)^{-1}$ 
to 
$\PPL(y)$.
%%%%%%%%%%%%%%%%%%%%%%%%%%%%%%%%%%%%%%%%%%%%%%%%%%%%%%%%%%%%%%%%%%%%%%%%%%%%%%%%%%%%%%%%%%%%%%%%%%%%%%%%%%%%%%%%
\end{enumerate}
%%%%%%%%%%%%%%%%%%%%%%%%%%%%%%%%%%%%%%%%%%%%%%%%%%%%%%%%%%%%%%%%%%%%%%%%%%%%%%%%%%%%%%%%%%%%%%%%%%%%%%%%%%%%%%%%
\end{enumerate}
%%%%%%%%%%%%%%%%%%%%%%%%%%%%%%%%%%%%%%%%%%%%%%%%%%%%%%%%%%%%%%%%%%%%%%%%%%%%%%%%%%%%%%%%%
Note that the final probability for each node is calculated by 
summing all the contributions from each bullet point.

\begin{acknowledgments}
We thank Katie Kruzan for useful discussions and help in debugging our software.
\end{acknowledgments}
%%%%%%%%%%%%%%%%%%%%%%%%%%%%%%%%%%%%%%%%%%%%%%%%%%%%%%%%%%%%%%%%%%%%%%%%%%%%%%%%%%%%%%%%

\bibliography{references-journ}
%%%%%%%%%%%%%%%%%%%%%%%%%%%%%%%%%%%%%%%%%%%%%%%%%%%%%%%%%%%%%%%%%%%%%%%%%%%%%%%%%%%%%%%%

\end{document}